\documentclass[11pt]{article}

\usepackage{amsfonts,amsmath,amssymb,amsthm,dsfont,boxedminipage,color,url,fullpage,mathtools}
\usepackage[margin=1in, papersize={8.5in,11in}]{geometry}
\usepackage{algorithm}
\usepackage{algpseudocode}
\usepackage{bbm}
\usepackage[inline]{enumitem}
\usepackage{stmaryrd}
\usepackage{float}
\usepackage[numbers]{natbib} 
\usepackage[pdfstartview=FitH,colorlinks,linkcolor=blue,filecolor=blue,citecolor=blue,urlcolor=blue]{hyperref} 
\usepackage[labelfont=bf]{caption}
\usepackage{aliascnt,cleveref}
\usepackage{graphicx}
\usepackage{makecell}
\usepackage[table]{xcolor}
\usepackage{xparse}
\usepackage{upgreek}

\usepackage{titlesec}
\usepackage{setspace}
\titlespacing{\paragraph}{0pt}{0.2cm}{1em}

\newtheorem{theorem}{Theorem}[section]
\newtheorem{corollary}[theorem]{Corollary}
\newtheorem{fact}[theorem]{Fact}
\newtheorem{lemma}[theorem]{Lemma}
\newtheorem{definition}[theorem]{Definition}
\newtheorem{remark}[theorem]{Remark}
\newtheorem{claim}[theorem]{Claim}
\newtheorem{example}[theorem]{Example}
\newtheorem{invariant}[theorem]{Invariant}

\newcommand{\K}{\mathcal{K}}
\newcommand{\Q}{\mathcal{Q}}

\newcommand{\norm}[1]{\left\| #1 \right\|}

\newcommand\numeq[1]{\stackrel{\scriptscriptstyle(\mkern-1.5mu#1\mkern-1.5mu)}{=}}
\newcommand\numge[1]{\stackrel{\scriptscriptstyle(\mkern-1.5mu#1\mkern-1.5mu)}{\ge}}
\newcommand\numle[1]{\stackrel{\scriptscriptstyle(\mkern-1.5mu#1\mkern-1.5mu)}{\le}}

\NewDocumentCommand\F{gg}{\ensuremath{F\IfNoValueTF{#1}{}{_{#1}}\IfNoValueTF{#2}{}{(#2)}}}
\NewDocumentCommand\M{gg}{\ensuremath{M\IfNoValueTF{#1}{}{_{#1}}\IfNoValueTF{#2}{}{(#2)}}}
\NewDocumentCommand\N{gg}{\ensuremath{N\IfNoValueTF{#1}{}{_{#1}}\IfNoValueTF{#2}{}{(#2)}}}
\NewDocumentCommand\W{gg}{\ensuremath{W\IfNoValueTF{#1}{}{_{#1}}\IfNoValueTF{#2}{}{(#2)}}}
\NewDocumentCommand\D{gg}{\ensuremath{D\IfNoValueTF{#1}{}{_{#1}}\IfNoValueTF{#2}{}{(#2)}}}
\NewDocumentCommand\A{gg}{\ensuremath{A\IfNoValueTF{#1}{}{_{#1}}\IfNoValueTF{#2}{}{(#2)}}}
\NewDocumentCommand\B{gg}{\ensuremath{B\IfNoValueTF{#1}{}{_{#1}}\IfNoValueTF{#2}{}{(#2)}}}
\NewDocumentCommand\X{gg}{\ensuremath{X\IfNoValueTF{#1}{}{_{#1}}\IfNoValueTF{#2}{}{(#2)}}}
\NewDocumentCommand\Y{gg}{\ensuremath{Y\IfNoValueTF{#1}{}{_{#1}}\IfNoValueTF{#2}{}{(#2)}}}
\NewDocumentCommand\Z{gg}{\ensuremath{Z\IfNoValueTF{#1}{}{_{#1}}\IfNoValueTF{#2}{}{(#2)}}}
\NewDocumentCommand\R{gg}{\ensuremath{R\IfNoValueTF{#1}{}{_{#1}}\IfNoValueTF{#2}{}{(#2)}}}
\NewDocumentCommand\Smat{gg}{\ensuremath{S\IfNoValueTF{#1}{}{_{#1}}\IfNoValueTF{#2}{}{(#2)}}}
\NewDocumentCommand\U{gg}{\ensuremath{U\IfNoValueTF{#1}{}{_{#1}}\IfNoValueTF{#2}{}{(#2)}}}
\NewDocumentCommand\Hmat{gg}{\ensuremath{H\IfNoValueTF{#1}{}{_{#1}}\IfNoValueTF{#2}{}{(#2)}}}
\NewDocumentCommand\Pmat{gg}{\ensuremath{P\IfNoValueTF{#1}{}{_{#1}}\IfNoValueTF{#2}{}{(#2)}}}
\NewDocumentCommand\I{gg}{\ensuremath{I\IfNoValueTF{#1}{}{_{#1}}\IfNoValueTF{#2}{}{(#2)}}}
\NewDocumentCommand\dG{gg}{\ensuremath{d\IfNoValueTF{#1}{}{_{#1}}\IfNoValueTF{#2}{}{(#2)}}}
\NewDocumentCommand\dGG{g}{\ensuremath{d_G\IfNoValueTF{#1}{}{(#1)}}}

\NewDocumentCommand\muG{gg}{\ensuremath{\mu\IfNoValueTF{#1}{}{_{#1}}\IfNoValueTF{#2}{}{(#2)}}}
\NewDocumentCommand\muGG{g}{\ensuremath{\mu\IfNoValueTF{#1}{}{(#1)}}}

\newcommand{\RR}{\mathop{\mathbb R}}

\newcommand{\1}{\mathbbm{1}}

\newcommand{\eg}{{\it e.g.}}
\newcommand{\ie}{{\it i.e.}}

\DeclareMathOperator{\vol}{\mathbf{vol}}
\DeclareMathOperator{\poly}{poly}

\DeclareMathOperator{\dem}{dem}
\DeclareMathOperator{\charge}{charge}
\DeclareMathOperator{\capacity}{cap}
\DeclareMathOperator{\mincut}{mincut}
\DeclareMathOperator{\congestion}{cong}
\newcommand{\defeq}{\mathrel{\mathop:}=}

\counterwithin*{equation}{theorem}

\definecolor{PrintGreen}{cmyk}{0.8, 0.0, 1.0, 0.0}

\begin{document}
    \title{Improved Tree Sparsifiers in Near-Linear Time
    \thanks{This work was supported in part by Israel Science Foundation grant no.\ 1156-23 and the Blavatnik Family Foundation.}
	}
    \date{\today}
	
	\author{Daniel Agassy\thanks{Tel Aviv University,
danielagassy@mail.tau.ac.il} \and Dani Dorfman\thanks{Max Planck Institute for Informatics, 
ddorfman@mpi-inf.mpg.de} 	\and Haim Kaplan\thanks{Tel Aviv University, haimk@tauex.tau.ac.il}}
    \date{\today}

	\maketitle
	
	\noindent
	
    \begin{abstract}

A \emph{tree cut-sparsifier} $T$ of quality $\alpha$ of a graph $G$ is a single tree that preserves the capacities of all cuts in the graph up to a factor of $\alpha$. A \emph{tree flow-sparsifier} $T$ of quality $\alpha$ guarantees that every demand that can be routed in $T$ can also be routed in $G$ with congestion at most $\alpha$.

We present a near-linear time algorithm that, for any undirected capacitated graph $G=(V,E,c)$, constructs a tree cut-sparsifier $T$ of quality $O(\log^{2} n \log\log n)$, where $n=|V|$. This nearly matches the quality of the best known polynomial construction of a tree cut-sparsifier, of quality  $O(\log^{1.5} n \log\log n)$ [Räcke and Shah, ESA~2014]. By the flow-cut gap, our result yields a tree flow-sparsifier (and congestion-approximator) of quality $O(\log^{3} n \log\log n)$. This improves on the celebrated result of [Räcke, Shah, and Täubig, SODA~2014] (RST) that gave a near-linear time construction of a tree flow-sparsifier of quality $O(\log^{4} n)$.

    Our algorithm builds on a recent \emph{expander decomposition} algorithm by [Agassy, Dorfman, and Kaplan, ICALP~2023], 
    which we use as a black box to obtain a clean and modular foundation for tree cut-sparsifiers. This yields an improved and simplified version of the RST construction for cut-sparsifiers with quality $O(\log^{3} n)$. We then introduce a near-linear time \emph{refinement phase} that controls the load accumulated on boundary edges of the sub-clusters across the levels of the tree. Combining the improved framework with this refinement phase leads to our final $O(\log^{2} n \log\log n)$ tree cut-sparsifier.
	\end{abstract}
	
	\newpage
	
	\section{Introduction}\label{sec:introduction}
    A cut-sparsifier of a weighted undirected graph $G=(V,E,w)$ is a sparse graph that preserves the cuts of $G$ up to some factor. Formally $H = (V_H,E_H,w_H)$ is a cut-sparsifier of $G$ of quality $\alpha$ if $V\subseteq V_H$ and for every $S\subseteq V$ it holds that $\capacity(S,V\setminus S) \le \mincut_H(S, V\setminus S) \le \alpha \cdot \capacity(S,V\setminus S) $, where 
    \[\mincut_H(S, V\setminus S) \defeq \min_{U:\ S\subseteq U, \; V\setminus S \subseteq V_H \setminus U} \capacity_H(U, V_H\setminus U).\]

    A stronger notion is of flow-sparsifiers. We say that $H = (V_H,E_H,w_H)$, where $V\subseteq V_H$, is a flow-sparsifier of $G=(V,E,w)$ of quality $\alpha$ if every multi-commodity flow problem that can be routed in $G$ with congestion $1$, can be also routed in $H$ with congestion $1$ and, moreover, every multi-commodity flow problem that can be routed in $H$ with congestion $1$, can be also routed in $G$ with congestion $\alpha$.

     A concept closely related to tree flow-sparsifiers is that of \emph{congestion-approximators}. An $\alpha$-congestion-approximator of $G$ is a collection of cuts $\{(S_i , V\setminus S_i) \}_{i=1}^k$ that approximately capture the optimal congestion needed to route multi-commodity flows in $G$. Formally, for any demand matrix $Q$, consider the minimum ratio (over all cuts $(S_i, V \setminus S_i)$) between the capacity of the cut and the total demand separated by it, \ie,
     $\frac{\capacity(S_i,V\setminus S_i)}{\dem_Q(S_i, V\setminus S_i)}$. The collection $\{(S_i, V \setminus S_i)\}$ is an $\alpha$-congestion-approximator if this ratio $\alpha$-approximates the minimum congestion required to route $Q$ in $G$. Tree flow-sparsifiers arising from hierarchical decompositions yield congestion approximators of the same quality, consisting of only a near-linear number of cuts.
     In a breakthrough result, Sherman~\cite{sherman2013nearly} showed that congestion-approximators of polylogarithmic quality suffice to compute approximate maximum flow in near-linear time.

    R{\"a}cke~\cite{racke2002minimizing} proved the existence of tree flow-sparsifiers of quality $O(\log^3 n)$. The construction is based on a hierarchical decomposition of $G$. This result was soon improved to a polynomial time algorithm for computing a tree flow-sparsifier of quality $O(\log^2 n \log \log n)$~\cite{harrelson2003polynomial}. 
    This series of works has culminated in an almost linear time (that is, $O\!\left(m^{1+o(1)}\right)$ time) 
    construction of a tree flow-sparsifier of quality $O(\log ^4 n)$ by R{\"a}cke, Shah, and T{\"a}ubig~\cite{racke2014computing}, which was later improved to a near-linear time (that is, $\tilde{O}(m)$ time\footnote{The $\tilde{O}$ notation is used to hide $\log^c n$ factors.}) 
    construction due to Peng's near-linear time approximate max flow algorithm~\cite{peng2016approximate}. 
    This $O(\log^4 n)$ bound is the best known for near-linear time constructions of a tree flow-sparsifier and of congestion-approximators of near-linear size.

    Even though flow-sparsifiers are robust,  for many applications 
    a cut-sparsifier suffices (see, \eg, \cite{chekuri2004all,engelberg2007cut,andreev2009simultaneous,bansal2014min,van2024almost}). 
    The best known polynomial time  construction by R{\"a}cke and Shah~\cite{RackeS14} gives a tree cut-sparsifier of quality  $O(\log^{1.5} n\log\log n)$. Grid graphs impose an $\Omega(\log n)$ lower bound on the quality of tree cut-sparsifiers and flow-sparsifiers~\cite{bartal1997line,maggs1997exploiting}.

    A tree flow-sparsifier $T$ of quality $\alpha$ is also a tree cut-sparsifier of the same quality. Therefore, \cite{racke2014computing, peng2016approximate} implies a near-linear time construction of a tree cut-sparsifier of quality $O(\log^4 n)$, and this is the best near-linear time construction to date.
    We give an $\tilde{O}(m)$ time construction of tree cut-sparsifier of quality $O(\log^2 n \log \log n)$, a tree flow-sparsifier of quality $O(\log^3 n \log\log n)$, and a congestion-approximator of near-linear size of quality $O(\log^3 n \log\log n)$. This improves for the first time upon the 2014 bound of \cite{racke2014computing,peng2016approximate}. The following theorem is our main result.

    \begin{theorem}
    \label{theorem:main}
        Given a graph $G=(V,E)$ of $m$ edges, there exists a randomized algorithm which takes $\tilde{O}(m)$ time and with high probability finds a weighted tree $T=(V_T, E_T, w_T)$ with $V\subseteq V_T$, such that $T$ is a cut-sparsifier of $G$ with quality $O(\log^2 n \log \log n)$. The same tree is also a flow-sparsifier of $G$ with quality $O(\log^3 n \log \log n)$, which in particular induces a congestion-approximator of $G$ with quality $O(\log^3 n \log \log n)$ and near-linear size. 
    \end{theorem}

   \subsection{Techniques}\label{subsec:intro-techniques}
    We now elaborate on the tools and techniques that allow us to obtain Theorem  \ref{theorem:main}. First, we develop a way to work directly with cuts rather than flows.
    \smallskip

    \noindent\textbf{Demand State Formulation.} 
    Certifying that a hierarchical tree $T$ is a flow-sparsifier of a graph $G$ with quality $\alpha$, amounts to showing that every multi-commodity flow problem $Q$ routable in $T$ with congestion $1$, is routable in $G$ with congestion $\alpha$. The natural approach is to route the mass in $G$ according to the structure of $T$: we partition~$G$ to clusters according to the laminar set family given by the nodes of~$T$, and iteratively route flow in~$G$ from the boundary edges of each cluster to the boundary edges of its parent cluster. The process begins at the leaf clusters (corresponding to the vertices~$V$) and proceeds upward, canceling demands whenever both endpoints lie inside the same cluster.

    For the cut-sparsifier setting, we replace explicit routings by a notion of $G$ \emph{respecting} demand matrices. We say that $G$ $\alpha$-respects a demand matrix if for every cut $(S,\bar{S})$ the total amount of that demand that must cross $(S,\bar{S})$ is at most $\frac{1}{\alpha} \cdot \capacity_G(S,\bar{S})$. 
    To prove that $T$ is a tree cut-sparsifier we need to show that $G$ $\frac{1}{\alpha}$-respects every demand matrix that is $1$-respected by $T$ (see Claim~\ref{claim:equiv-def-cut}).
    To carry out this task we use the notion of \emph{demand states} (Definition~\ref{def:flow_state}). A demand state assigns to each vertex a vector of commodity masses, and each update to the state corresponds to moving mass. Demand states allow us to formally define and maintain invariants, reason about how mass is redistributed across clusters, and rigorously charge these movements to the capacity of cuts. In contrast to flow-sparsifiers, these movements and corresponding charges are not necessarily achieved by explicit routing. 
    While similar techniques were used in previous work, using this demand state formulation allows us to provide a clearer and more systematic presentation of the algorithms, making their structure explicit. 

    \smallskip
    \noindent\textbf{Expander Decomposition.} Leveraging the ``balanced cut or expander" primitive from~\cite{ADK25} with adaptively chosen expansion parameters, we develop a unified approach that allows us to simplify and strengthen the~\cite{racke2014computing} framework to obtain a tree cut-sparsifier of quality $O(\log^3 n)$. Afterwards, we use the same approach to implement a variant of the refinement phase of~\cite{harrelson2003polynomial,RackeS14} in near-linear time, as we describe below.

    Using this primitive from \cite{ADK25}, together with the fair-cuts computation \cite{LNPSsoda23}, allows us to abstract away the details of the underlying algorithmic machinery and significantly simplify the construction of the tree cut-sparsifier, when compared to \eg, \cite{racke2014computing}. This connects tree cut-sparsifier constructions to modern sparse-cut tools and shows how they can be constructed using these tools.  

    \smallskip
    \noindent\textbf{A Near-Linear Time Refinement Phase.}
    A main caveat in the construction of~\cite{racke2014computing} is that the boundary edges of the clusters are not \emph{flow-linked}. The notion of a flow-linked set, introduced by~\cite{chekuri2005multicommodity}, 
    means that it is possible to route an all-to-all multi-commodity flow between the nodes (or edges) of this set. The lack of this property in~\cite{racke2014computing} adds a logarithmic factor to the quality of the resulting tree flow-sparsifier.
    For more details, see Section~\ref{subsec:recap-refinement}.
    In the context of tree cut–sparsifiers, flow-linkedness corresponds to \emph{expansion}. By applying the “balanced cut or expander’’ primitive of~\cite{ADK25} to a given cluster~$S$, we can test whether the boundary of~$S$ expands. If not, we split~$S$ along a balanced cut and recur on both sides. This recursive process, which we call the \emph{refinement phase}, allows us to partition a cluster $S$ into $\{S_1,\ldots,S_t\}$ such that in each $S_i$, the boundary edges expand. 

    Additionally, similar to~\cite{harrelson2003polynomial,RackeS14}, we show that there exists a low-congestion flow from the inter-cluster edges of the refinement partition to the boundary of~$S$. We do this by carefully using the guarantees supplied by \cite{ADK25}. This requires overcoming several technical obstacles, as we need to meticulously control both the number of boundary edges and the number of vertices in the resulting sub-clusters. 
    Importantly, we are able to perform this step in near-linear time, instead of polynomial time.
    
    Previous constructions~\cite{harrelson2003polynomial, RackeS14} relied on intricate recursive schemes that handled “bad child events’’, a subcall that violated the recursion invariants, by restarting parts of the computation. In contrast, our approach is entirely modular: we alternately apply a 
    \emph{merge phase} (on which we elaborate in Section \ref{subsec:first-step-summary}) and a refinement phase, each performing a single, independent partition of the current cluster (no backtracking is needed). This simple alternation yields a much cleaner and modular tree cut-sparsifier construction of quality $O(\log^2 n \log\log n)$.

    \subsection{Further Background}
    Our work builds on the tree flow-sparsifier constructions of~\cite{racke2014computing} and~\cite{harrelson2003polynomial}, combined with an expander decomposition result by~\cite{ADK25}.

    A crucial component in \cite{racke2014computing} was a variant of the cut-matching game, introduced by Khandekar, Rao, and  Vazirani \cite{khandekar2009graph}.
    The cut-matching game~\cite{khandekar2009graph} provides an $O(\log^2 n)$-approximate  sparsest cut in near-linear time.
 By running a ``non-stop" version of this game, \cite{racke2014computing} obtained
    \emph{balanced} $O(\log^2 n)$-approximate sparse cuts, or alternatively, constructed clusters whose inter-cluster edges (not including boundary edges) are \emph{flow-linked}.
    
    A stronger cut strategy was developed by~\cite{orecchia2008partitioning}, achieving an $O(\log n)$-approximation for the sparsest cut problem. The expansion produced by this cut strategy follows from Cheeger's inequality~\cite{cheeger1970lower} and therefore, at least intuitively, is unsuitable for flow routing purposes. This stands in contrast to the cut player of~\cite{khandekar2009graph}, whose expansion follows from embedding a uniform demands matrix (all $1$'s) into the expanding set.

    Agassy, Dorfman, and Kaplan~\cite{agassy2023expander} were able to translate \cite{orecchia2008partitioning}'s improved cut strategy into a ``balanced cut or expander" procedure of approximation $O(\log n)$, resulting in a near-linear time expander decomposition algorithm which improved on the prior expander decomposition algorithm by Saranurak and Wang~\cite{saranurak2019expander} that uses~\cite{khandekar2009graph}'s cut strategy. In~\cite{ADK25}, Agassy, Dorfman and Kaplan show how to extend their results to the case of generalized conductance. By generalized conductance, they refer to conductance with respect to a vertex measure. Formally, given a vertex measure $\mu : V \to \RR_{\ge 0}$, and a set $S\subseteq V$, the expansion of a cut $(S,\bar{S})$ with respect to $\muG$ is defined as $\Phi^{\mu}_{G}(S,\bar{S}) \defeq \frac{\capacity(S,\bar{S})}{\min(\mu(S), \mu(\bar{S}))}$, where $\mu(U)$, for $U\subseteq V$, is defined as $\mu(U) \defeq \sum_{v\in U}{\mu(v)}$.  We use the result by~\cite{ADK25} in a black-box manner as explained in Section~\ref{subsec:first-step-summary}. We note that for the purpose of obtaining our  $O(\log^3 \log \log n)$ quality tree flow-sparsifier, one could alternatively rely on~\cite{saranurak2019expander}, together with our refinement phase. However, achieving the improved $O(\log^2 \log \log n)$ quality tree cut-sparsifier crucially relies on the spectral cut player underlying~\cite{ADK25}.

    Typically, as in this work, hierarchical decompositions are computed \emph{top-down} by recursively splitting clusters starting from the root $V$. 
    Li, Rao, and Wang~\cite{li2025congestion} take the opposite, \emph{bottom-up} approach, constructing a congestion approximator of quality $O(\log^{10} n)$. 
    Their main tool is a weak expander decomposition: they build a sequence of partitions 
    $P_1, P_2, \ldots, P_L$ of $V$, where each partition corresponds to a weak expander decomposition, and $P_1$ is the partition into singleton clusters and $P_L = \{V\}$. 
    The collection of cuts from these partitions does not form a laminar family. Instead, the congestion-approximator’s set of cuts consists of all intersections of sub-clusters drawn from different levels. 
    A technical challenge is that expander-decomposition routines typically invoke approximate max-flow algorithms that themselves rely on the top-down congestion-approximator of~\cite{racke2014computing}. 
    The construction of~\cite{li2025congestion} avoids this potential circularity by maintaining ``pseudo''-congestion-approximators derived from the partial hierarchy 
    $P_1, \ldots, P_i$ and using them within the max-flow subroutines required for the weak expander decomposition.

    We note that the ``balanced cut or expander" primitive of~\cite{ADK25}, which we use extensively, internally computes fair-cuts (approximate maximum flows) via the algorithm of~\cite{LNPSsoda23}.  This algorithm requires a congestion–approximator for the underlying graph~$G$. By instantiating the congestion–approximator required by the fair-cut routine~\cite{LNPSsoda23} with the construction of \cite{li2025congestion}, our tree cut–sparsifiers can be implemented \emph{without} invoking the congestion–approximator of~\cite{racke2014computing}.
    This means that one can also think of our work (in particular, Sections \ref{section:merge-phase} and \ref{section:charging-scheme}) as a simplification of the construction of \cite{racke2014computing} (which does not \emph{depend} on it).

    \subsection{Related Work}

    When allowing to sparsify with a convex combination of trees, R{\"a}cke achieved a polynomial time algorithm for computing a flow-sparsifier~\cite{racke2008optimal} of optimal $O(\log n)$ quality. However, some problems require a flow-sparsifier which is a single tree (\eg, \cite{khandekar2014advantage,andreev2009simultaneous,bansal2014min}).

    Using modern techniques of (dynamic) expander decomposition, and expander-hierarchies, \cite{goranci2021expander} achieved a fully dynamic, $n^{o(1)}$ update and $O(\log^{1/6}n)$ query time, data structure for maintaining a tree flow-sparsifier of quality $n^{o(1)}$. 

    Tree flow–sparsifiers have been extensively used in the context of oblivious routing schemes, that is, routings that are independent of the input demands and preassign flow paths for every pair of vertices. 
    There has also been progress on hop-constrained routing, where paths are kept short. Ghaffari, Haeupler, and \v{Z}u\v{z}i\'{c}
    obtain a \emph{static} oblivious scheme that, for any hop bound $h$, uses paths of length at most $h\cdot \mathrm{poly}(\log n)$ and achieves congestion within $n^{o(1)}$ of the best  routing that uses paths of length at most $h$~\cite{ghaffari2021hop}. Haeupler, R{\"a}cke, and Ghaffari introduce hop-constrained expander decompositions, from which they derive hop-constrained routing~\cite{haeupler2022hop}. These works are complementary to our goal of producing a single tree with provable cut guarantees and near-linear preprocessing time.

    \paragraph{Recent Work.} In an independent and parallel work Henzinger, M{\"u}nk, and R{\"a}cke~\cite{henzinger2025improved} achieved a tree cut sparsifier with the same $O(\log^2 n \log \log n)$ quality. They also show how to implement their construction in parallel with polylogarithmic span and near linear work.
    Their hierarchy construction follows the~\cite{harrelson2003polynomial} bad-child-event paradigm: when a cluster cannot be certified to satisfy the required expansion invariant, the algorithm dynamically restructures the current hierarchy level and revisits the affected subproblems. This event-driven backtracking leads to a more intricate control flow. 
    In contrast, our construction avoids such restart mechanisms entirely. The merge and refinement phases are applied as independent, one-shot procedures, resulting in a simpler and more modular algorithmic structure.

    \section{Overview of Our Results and Techniques}
    \label{section:techniques}

    This section is organized as follows. Section \ref{section:preliminary_preliminaries} introduces key definitions used throughout the section (additional definitions and notation are deferred to Section \ref{section:preliminaries}). In Section \ref{section:our_contribution}, we detail our results using these definitions.
    In Section~\ref{subsec:relating-cut-and-flow-sparsifiers}, we introduce a general framework 
    for analyzing the quality of tree cut-sparsifiers. 
    Section~\ref{subsec:recap-refinement} reviews the key ideas and techniques from previous work. In Section~\ref{subsec:first-step-summary}, we outline the approach used to prove Theorem~\ref{theorem:final}, and in Section~\ref{subsec:full_construction_intuition}, we describe the additional ideas that give the improved result stated in Theorem~\ref{theorem:final-improved}. 

    \subsection{Preliminaries}
    \label{section:preliminary_preliminaries}

     We use $\log$ to denote the base-2 logarithm. Throughout the paper, we assume all graphs are undirected. Given an undirected graph $G = (V,E)$ and a set $S\subseteq V$, we denote by $\deg_G(v)$ the \emph{degree} of $v$ (for a weighted graph, this is the sum of the weights of the edges adjacent to $v$). 
     The \emph{volume} of a set $S$ is defined as $\vol_G(S) \defeq \sum_{v\in S} \deg_G(v)$. Given vertex weights $\mu : V \to \RR_{\ge 0}$, and a set $U\subseteq V$, we denote the total weight of $U$ as $\mu(U) \defeq \sum_{v\in U}{\mu(v)}$. We refer to $\mu$ interchangeably as a \emph{vertex weighting} or a \emph{vertex measure}.
     In almost all of our use cases, $\mu$ will be an indicator measure. That is, $\mu(v) = \1_{\{v\in X\}}$ for some $X\subseteq V$, and $\mu$ simply distinguishes the set $X$ of vertices (formally, $\mu(v) = 1$ if $v\in X$ and otherwise $\mu(v) = 0$).
    For disjoint sets $A,B \subseteq V$, we denote by $E_G(A,B)$ the set of edges connecting $A$ and $B$.
    We denote by $\capacity_G(A,B)$ the number of edges in $E_G(A,B)$, or the sum of their weights if the graph is weighted. We sometimes omit the subscript when the graph is clear from the context. Denote $\bar{A} = V \setminus A$. If $A$ and $\bar{A}$ are both non-empty, then we call $(A, \bar{A})$ a \emph{cut}.

    \begin{definition}[Multi-Commodity Flow Problem, Demand Matrix]
    \label{def:flow_problem}
        Given a vertex set $V$, a \emph{multi-commodity flow problem} $Q$ on $V$ is specified by a \emph{demand matrix} $D_Q\in\RR_{\ge 0}^{V\times V}$. Each vertex $u\in V$ is associated with a single \emph{commodity}: the $u$'th row of $D_Q$ describes how much of $u$'s commodity is to be delivered to every vertex $v\in V$. In other words, $D_Q(u, v)$ denotes the amount of flow of commodity $u$ that must be sent from $u$ to $v$ (\ie, $u$ is the source of its commodity). 

        Given $S \subseteq V$, we denote the \emph{demand of $Q$ across the cut $(S,\bar{S})$} by 
        \[\dem_Q(S,\bar{S}) \defeq \sum_{u\in S}{\sum_{v\in \bar{S}}{(D_Q(u, v) + D_Q(v, u))}}.\]

        We identify a multi-commodity flow problem with its demand matrix (or simply demands), and use $Q$ and $D_Q$ interchangeably.

    \end{definition}

    A specific flow problem which will be used multiple times is the \emph{all-to-all} flow problem.

    \begin{definition}[All-to-all flow problem]
    \label{def:all_to_all_flow}
         Let $G=(V,E)$ and let $A\subseteq V$. The all-to-all multi-commodity flow problem on $A$ is defined by the demand matrix $D_A(u, v) \defeq 
        \begin{cases*}
            \frac{1}{|A|}, & if $u,v \in A$,\\
            0,             & otherwise.
        \end{cases*}
        $.
    \end{definition}

    We now define congestion, which quantifies the efficiency of a given flow. 

    \begin{definition}[Congestion]\label{def:congestion}
        Given a graph $G=(V,E)$ with edge capacities $c:E\to\mathbb{R}_{>0}$
        and a flow $f:E\to\mathbb{R}_{\ge0}$ routing a demand matrix $Q$ ($f$ does not necessarily abide capacity constraints),
        the \emph{congestion} of $f$ in $G$ is
        \[
            \congestion_G(f)
            := \max_{e\in E} \frac{f(e)}{c(e)},
        \]
        where $f(e)$ is the total flow that crosses $e$, of all commodities. The \emph{congestion of \(Q\) in \(G\)}, $\congestion_G(Q)$, is the minimum possible value of $\congestion_G(f)$ over all flows $f$ routing $Q$.
    \end{definition}

    For completeness, we restate below the definitions from Section~\ref{sec:introduction} of cut–sparsifiers and flow–sparsifiers.
    Let $G=(V,E)$ be an undirected graph, and let $H = (V_H,E_H,w_H)$ be a weighted undirected graph with $V \subseteq V_H$.

    \begin{definition}[Cut-Sparsifier]
        \label{def:cut-sparsifier}
        We say that $H=(V_H,E_H)$ is a cut-sparsifier of $G=(V,E)$ of quality $\alpha$ if $V\subseteq V_H$ and for every $S\subseteq V$ it holds that 
        \[\text{$\capacity_G(S,V\setminus S) \le \mincut_H (S, V\setminus S) \le \alpha \cdot \capacity_G(S, V\setminus S)$},\]
        where \text{$\mincut_H$}$(S, V\setminus S) =\min \{\capacity_H(U, V_H\setminus U) \mid S\subseteq U,\; V\setminus S \subseteq V_H \setminus U \}$.
    \end{definition}

    \begin{definition}[Flow-Sparsifier]\label{def:flow-sparsifier}
        Given a graph $G=(V,E)$, we say that a graph $H$ on the vertices $V_H$, where $V\subseteq V_H$, is a flow-sparsifier of $G$ of quality $\alpha$ if for every demand matrix $Q$ on the vertex set $V$, 
        it holds that
        \[
            \congestion_H(Q) \leq \congestion_G(Q) \leq \alpha \cdot \congestion_H(Q).
        \]

    \end{definition}

    \begin{definition}[Respecting Cut]
    \label{def:cut_respects_flow}
        Let $Q$ be a demand matrix on $G=(V,E)$ and let $(S, \bar{S})$ be a cut in $G$. 

        We say that $(S, \bar{S})$ \emph{$\alpha$-respects $Q$} if $\capacity(S, \bar{S}) \ge \alpha\cdot \dem_Q(S, \bar{S})$. We say that $G$ \emph{$\alpha$-respects $Q$} if
        \[
            \min_{\substack{S\subseteq V\\\dem_Q(S, \bar{S}) > 0}} \frac{\capacity(S, \bar{S})}{\dem_Q(S, \bar{S})} \ge \alpha.
        \]
        More generally, if $\tilde{G}=(\tilde{V},\tilde{E})$ satisfies $V \subseteq \tilde{V}$, then in the above definition we replace $\capacity(S, \bar{S})$ by $\mincut_{\tilde{G}}(S, \bar{S})$.
    \end{definition}

    \begin{fact}
    \label{fact:congestion_implies_respecting_cut}
        If the demand matrix $Q$ can be routed in $G$ with congestion $\alpha$, then $\dem_{Q}(S, V\setminus S) \le \alpha\cdot\capacity_G(S, V\setminus S)$, 
        for all cuts $(S, V\setminus S)$, \ie, $G$ $\frac{1}{\alpha}$-respects $Q$.
    \end{fact}

    Note that $G$ $\alpha$-respecting a demand matrix $Q$ does not necessarily imply that $Q$ can be routed with congestion $1/\alpha$ in $G$. Due to the flow–cut gap~\cite{leighton1999multicommodity}, an additional $O(\log n)$ factor may be required.

    \subsection{Our Contribution}
    \label{section:our_contribution}
    
    We prove the following theorem:
    \begin{theorem}
    \label{theorem:final-improved}
        Given a graph $G=(V,E)$ of $m$ edges, there exists a randomized algorithm which takes $\tilde{O}(m)$ time and with high probability finds a weighted tree $T=(V_T, E_T, w_T)$ with $V\subseteq V_T$, such that $T$ is a cut-sparsifier of $G$ with quality $O(\log^2 n \log \log n)$.
    \end{theorem}

    However, first, to gradually develop our techniques, we prove a weaker and simpler version of the theorem, in which we look for a tree cut-sparsifier of quality $O(\log^3 n)$:
    \begin{theorem}
    \label{theorem:final}
        Given a graph $G=(V,E)$ of $m$ edges, there exists a randomized algorithm which takes $\tilde{O}(m)$ time and with high probability finds a weighted tree $T=(V_T, E_T, w_T)$ with $V\subseteq V_T$, such that $T$ is a cut-sparsifier of $G$ with quality $O(\log^3 n)$.
    \end{theorem}

    \begin{remark}[Hierarchical trees]
    \label{remark:hierarchical-trees}
        The trees produced in Theorems~\ref{theorem:final-improved} and~\ref{theorem:final} are \emph{hierarchical trees}: each node of $T$ corresponds to a subset of vertices of $G$, and these subsets form a \emph{laminar family}\footnote{That is, the subsets corresponding to the children of a node in $T$ are pair-wise disjoint and contained in the subset corresponding to that node.} in which the root corresponds to $V$ and the leaves correspond to all singleton subsets $\{v\}_{v\in V}$.
        Each edge $e \in E_T$ connects a node, corresponding to a cluster $S_e \subseteq V$, with its parent node, corresponding to a larger cluster $S_e \subseteq P_T(S_e) \subseteq V$. This edge has a capacity equal to the capacity of the cut $(S_e, V\setminus S_e)$ in $G$, \ie, $\capacity_T(e) = \capacity_G(S_e, V\setminus S_e)$.
    \end{remark}
    
    The following claim follows from the flow-cut gap theorem~\cite{leighton1999multicommodity, aumann1998log}, together with the fact that the flow-cut gap in a tree is $1$~\cite{okamura1981multicommodity}.
    We defer the (standard) proof to appendix \ref{appendix:omitted-proofs}.
    \begin{claim}
    \label{claim:cut-sparsifier-implies-flow}
        The following relations hold for a (hierarchical) tree $T$ of $G$:
        \begin{itemize}
            \item If $T$ is a cut-sparsifier of $G$ of quality $\alpha$, then $T$ is also a flow-sparsifier of quality $O(\alpha \log n)$.
            \item If $T$ is a flow-sparsifier of $G$ of quality $\alpha$, then $T$ is also a cut-sparsifier of quality $\alpha$.
        \end{itemize}
    \end{claim}
    
    Hence, as a corollary of Theorem \ref{theorem:final-improved}, we get:
    \begin{theorem}
    \label{theorem:final-improved-flow}
        Given a graph $G=(V,E)$ of $m$ edges, there exists a randomized algorithm which takes $\tilde{O}(m)$ time and with high probability finds a weighted tree $T=(V_T, E_T, w_T)$ with $V\subseteq V_T$, such that $T$ is a flow-sparsifier of $G$ with quality $O(\log^3 n \log \log n)$, which in particular induces a congestion-approximator of $G$ with quality $O(\log^3 n \log \log n)$ and near-linear size.
    \end{theorem}

    The flow-cut gap for planar graphs is $O(\sqrt{\log n})$~\cite{rao1999small}. Therefore, for planar graphs, our construction achieves tree flow-sparsifier (and congestion-approximator) of quality $O(\log^{2.5} n \log \log n)$.

    For simplicity, we present our algorithms in the setting where $G$ is unweighted, but our constructions can be adapted to the case of polynomially bounded weights. 

    \subsection{Certifying a Tree Cut-Sparsifier Versus Certifying a Tree Flow-Sparsifier}\label{subsec:relating-cut-and-flow-sparsifiers}
    
    Our construction follows a recursive partitioning algorithm which takes as input a cluster $S\subseteq V$ and partitions it to sub-clusters $S_1,\ldots,S_t$. Applying this algorithm recursively gives a laminar decomposition of the set $V$, which naturally corresponds to a hierarchical tree $T$ (see Remark~\ref{remark:hierarchical-trees}). The root of $T$ corresponds to the set $V$ while the leaves correspond to individual vertices $v\in V$.
    The weight of an edge $(S, P_T(S))$ in $T$, between the node corresponding to the cluster $S$ and its parent $P_T(S)$ is $\capacity_G(S, V\setminus S)$.

    Consider first the task of certifying that such a tree $T$ is a flow-sparsifier. We can do this by showing that the congestion of routing a demand matrix $Q$ in $T$ is comparable to the congestion of routing $Q$ in $G$. Note that by choosing the capacities of the tree edges as we did, the inequality $\congestion_T(Q)\le\congestion_G(Q)$ is immediate: we route the demand matrix $Q$ in $T$ trivially along the unique path between each pair of vertices. If the flow problem can be routed in $G$ with congestion $c$, then by Fact~\ref{fact:congestion_implies_respecting_cut} it must be that, for every sub-cluster $S\in V_T$, the cut $(S, V\setminus S)$ $\frac{1}{c}$-respects $Q$:  $\frac{\capacity(S, V\setminus S)}{\dem_Q(S, V\setminus S)} \ge \frac{1}{c}$. 
    The demand across each cut $(S, V\setminus S)$ corresponds exactly to the total amount of flow routed along the edge $(S, P_T(S))$ in the tree, so the way the edge capacities in $T$ are assigned
    ensures that the flow problem can be routed in the tree with congestion $c$. Therefore, when building a tree flow-sparsifier, our goal is to choose the tree such that the reverse inequality $\congestion_G(Q)\le\alpha\cdot\congestion_T(Q)$ also holds. This means that  given  demands $Q$ which can be routed in $T$ we have to show how to route $Q$ in  $G$ with similar congestion. $Q$ being routable in $T$ means that only for every $S\in V_T$, the cut $(S, V\setminus S)$ in $G$, respects $Q$. Therefore, the intuition is that we look to include in the tree the sub-clusters which best capture the cut structure of the graph, in the sense that any demands respected by them can be routed in $G$.

    In the flow-sparsifier case, given demands  $Q$  routable in $T$, \cite{harrelson2003polynomial, bienkowski2003practical, racke2014computing} used the tree structure to describe the routing of $Q$ in the graph $G$. They did this by routing  $Q$  which originates at all the vertices of $V$ (corresponding to leaves in $T$), toward the boundaries of sub-clusters $S$ as they go up along the tree.\footnote{This is made formal using the subdivision graph (see Definition \ref{def:subdivision_graph}), where the demands  reside on the split nodes of the boundary edges. For the overview in this section, however, we omit these details for simplicity.}
    That is, when a sub-cluster $S\in T$ is partitioned into $\{S_1,\ldots,S_t\}$,
     they route (in $G$) the mass from the boundary edges of $\{S_i\}_{i=1}^t$ (where they arrived in the previous iteration) to the boundary edges of $S$. 
     We emphasize that this routing takes place in $G$, while the structure of $T$ is used only as a ``blueprint'' to direct the intermediate flows. 
     Then, source and target demands of vertex pairs that are both in $S$, are ``paired up'' and canceled before moving further upwards in the tree.

    Turning now back to cut-sparsifiers, in order to certify that $T$ is a tree cut-sparsifier,  we no longer need to route the demands  $Q$ (which can be routed in $T$) in  $G$, but rather only show that $G$ respects $Q$ (see Claim \ref{claim:equiv-def-cut} which essentially says that $T$ is a tree cut-sparsifier of quality $\alpha$ if it $1/\alpha$-respects every demands $Q$ that are $1$-respected by $T$). 
    To show that demands $Q$ are respected by $G$ we develop a technique of moving demands rather than flow. For this we define the notion of a \emph{valid demand state} (see Definition~\ref{def:flow_state}). 

    In a demand state $P$ each node $v\in V$ has its own vector of $|V|$ commodities (one for each vertex), where for each commodity it has some amount of mass that is waiting at $v$ to be shipped out (or, if the amount is negative, waiting to be shipped into $v$). This representation is particularly convenient in the hierarchical setting: 
    Even when we focus on a sub-cluster $S$, we maintain at each $v\in S$ information on the mass from all $|V|$ commodities including those of vertices that are not in $S$. In contrast, working directly with $|S|\times |S|$ demand matrices within a sub-cluster $S$ would be inconvenient, as the local view would lose information about demands involving vertices outside $S$. The demand state formulation, on the other hand, always retains the complete global demand structure, allowing us to reason locally within clusters while still maintaining awareness of the external interactions. 
    
    We use the structure of $T$ to direct the movement of demands. This corresponds to updating the demand state. Specifically, as we climb up the sparsifier $T$, we move the demands of the demand state from internal edges to the boundary edges of sub-clusters $S\in V_T$.
    However, unlike in the flow-sparsifier setting, these movements do not necessarily correspond to an actual multi-commodity flow that is routable in the graph. In fact, we do not ``pay''  for the congestion of these routings, as in the flow-sparsifier analysis, but rather according to how much the graph respects them (see Definition~\ref{def:cut_respects_flow}). This charging scheme is detailed in Section \ref{section:charging-scheme}.

    \subsection{Dealing With Bottleneck Cuts}
    \label{subsec:recap-refinement}

    As mentioned in Section~\ref{subsec:relating-cut-and-flow-sparsifiers}, a key property of tree flow-sparsifier $T$ is that for every cluster $S\in V_T$, we can move the demands currently in the inter-cluster edges $F$ between sub-clusters
    $\{S_1,\ldots,S_t\}$
    of $S$, to the boundary edges $B = E(S, V\setminus S)$ of $S$ while incurring small congestion. 
    An ideal scenario is when the algorithm guarantees that $F \cup B$ is \emph{flow-linked} inside $G[S]$, meaning that one can route an all-to-all multi-commodity  flow $Q_{F\cup B}$ (see Definition~\ref{def:all_to_all_flow}) 
    between the edges of $F \cup B$ with low congestion. In this scenario, we can transfer the demands from the inter-cluster edges $F$ to the boundary $B$ in a two-step process:
    \begin{itemize}
        \item Spread the demands of $F\cup B$ uniformly by routing them according to $Q_{F\cup B}$.
        \item Route the demands of $F$ to $B$ according to the flow paths of the sub-matrix $Q_{F\cup B}(F,B)$.
    \end{itemize}
    
    Note that the uniform spreading in the first step ensures that demands whose source and target are in $S$ are routed to the same uniform distribution over $F\cup B$ and are therefore canceled. Moreover, after the second step every boundary edge carries an equal share of the total demand. 
    Because $S\in V_T$, the cut $E(S, V\setminus S)$ must respect $Q$, so the amount of mass that has to leave the cluster $S$ is at most $E(S, V\setminus S)$. Therefore, the mentioned cancellation guarantees that the load (total demand per edge) is at most $1$ on each boundary edge.

    However, as discussed in~\cite{bienkowski2003practical, harrelson2003polynomial, RackeS14, racke2014computing}, 
    there may be sub-clusters $S\subseteq V$, for which we cannot find a partition $\{S_i\}_{i=1}^t$ such that $F\cup B$ is flow-linked in $G[S]$. A simple counter-example is that it might be the case that $B$ itself is not flow-linked in $G[S]$, or maybe even disconnected there.
    In this case, no partition of $G[S]$ will work.
    We use the name \emph{bottleneck cut} to refer to a cut which certifies that the boundary edges are not flow-linked. This name was introduced by~\cite{RackeS14}, though their bottleneck cuts only had to satisfy a weaker requirement.
    A similar (yet not equivalent) notion appeared earlier in \cite{bienkowski2003practical, harrelson2003polynomial}. The algorithms of \cite{bienkowski2003practical, harrelson2003polynomial, RackeS14} had to make sure that the clusters have no bottleneck cuts. They achieved this by cutting clusters that have bottleneck cuts. We call this step the \emph{refinement phase} \cite{RackeS14}. Previous work implement the refinement phase in different ways, either by ensuring a ``precondition'' on the sub-clusters they recur into \cite{bienkowski2003practical}, or encountering these cuts and handling them ``on the fly'' \cite{harrelson2003polynomial, RackeS14}. In either case dealing with these cuts slowed down these early algorithms substantially.

   Since ensuring that there is no bottleneck cut in each sub-cluster before recurring into it may be expensive (in terms of running time),  \cite{racke2014computing} settled for something slightly weaker: They first ignore the boundary edges of $S$ and find a partition
    of $S$ into  $\{S_1,\ldots,S_t\}$ such that the inter-cluster edges $F$ are flow-linked. 
    Then they find an approximate min-cut $Y$ for the (single commodity) flow problem which routes flow from $F$ to $B$. 
    By adding to $T$ a new sub-cluster whose boundary is $Y$,
    \cite{racke2014computing} 
    are able to bound the amount of mass that has to cross the cut $Y$, since now the demand matrix $Q$ is respected by $Y$ in $G$. 
    
    Adding $Y$ as a boundary of a sub-cluster in $T$ is enough to show how to move the demands to the boundary edges $B$, such that the load on each boundary edge is multiplied by $1+\uptau$ units ($\uptau$ is a parameter $0 < \uptau \le 1$).
    However, the distribution of demands on the boundary edges is not uniform. Hence, demands whose source and target are in the sub-cluster do not cancel-out. 
    Consequently, the load accumulates when going up along the tree, multiplied by a factor of at most $1+\uptau$ in each iteration. Choosing $\uptau = \frac{1}{\log n}$, such that the load is multiplied by a factor of $(1+\frac{1}{\log n})$ along the $O(\log n)$ levels of the tree, ensures that the load is constant. 
    This method, however, incurs a degradation of $O(\frac{1}{\uptau})$ per level in the quality of the tree flow-sparsifier,
    leading to an overall loss of $O(\frac{\log n}{\uptau})=O(\log^2 n)$.\footnote{Choosing $\uptau = \frac{1}{\log n}$ is optimal, up to a multiplicative constant. Let $k=\Theta(\log n)$ be the depth of the tree. We omit the proof of the following claim: For any sequence of $\uptau_1,\ldots,\uptau_k \le 1$, it holds that $\sum_{i=1}^k{\frac{1}{\uptau_i}\prod_{j=i+1}^k{(1+\uptau_j)}} = \Omega(k^2)$.} The  additional $O(\log^2 n)$ factor in their quality comes from the usage of~\cite{khandekar2009graph}'s cut-matching game for flow (so in total the quality is $O(\log^4 n)$).

    In our first step below (Section \ref{subsec:first-step-summary}) we deal with bottleneck cuts similarly to \cite{racke2014computing}, but using primitives from the expander decomposition algorithm in \cite{ADK25}, which are the cut-analogue of~\cite{khandekar2009graph}'s cut-matching game. This allows us to get Theorem \ref{theorem:final}.
    Afterwards, in Section \ref{subsec:full_construction_intuition} we describe a better way to deal with bottleneck cuts, by implementing a version of the refinement phase in near-linear time. This allows us to improve the quality of the resulting tree cut-sparsifier and get Theorem \ref{theorem:final-improved}.

    \subsection{A Tree Cut-Sparsifier of Quality \texorpdfstring{$O(\log^3 n)$}{O(log\^3 n)}}
    \label{subsec:first-step-summary}

    We now adapt~\cite{racke2014computing} to the \emph{cut-sparsifier} setting. As outlined in Section~\ref{subsec:relating-cut-and-flow-sparsifiers}, our goal here is not to route the demand matrix $Q$ within $G$, but rather to ensure that $G$ respects $Q$.
    In this formulation, the analogue of a set $F$ of inter-cluster edges (within a cluster $S$) being flow-linked in~\cite{racke2014computing} is that $G[S]$ respects the all-to-all demand matrix between edges of $F$ (see Definition~\ref{def:all_to_all_flow}), or equivalently that $F$ expands in $G[S]$ (see Definition \ref{def:expanding-set} below). 
    Therefore, our first goal is to obtain, efficiently, such a set $F$. For this step we ignore the boundary edges of $S$ (which may be separated by bottleneck cuts).
    We use the following algorithmic tool: 
   
    \paragraph{Balanced-cut-or-expander procedure.} To obtain such an $F$ efficiently we use the “balanced cut or expander” primitive from~\cite{ADK25},
    which allows us to certify that the sub-cluster $G[S]$ respects the all-to-all flow problem on a subset $F\subseteq E$, or find a balanced cut which refutes this.
    We start with $F_1=E(G[S])$ and iterate on $i=1,\ldots$: 
    \begin{itemize} 
        \item Invoke the primitive on $G[S]$ to check if it respects the all-to-all flow problem on $F_i$. 
        \item On an \emph{expander} outcome, we stop and set $F:=F_i$. This certifies the all-to-all flow problem on $F$ is respected by $G[S]$.
        \item On a \emph{balanced cut} outcome, we shrink $F_i$ to a new set $F_{i+1}\subset F_i$. The sparsity guarantee ensures $|F_{i+1}|\le (1-1/\log n)|F_i|$, so after $O(\log^2 n)$ iterations we must obtain the first case.
    \end{itemize} 
    This step is the cut-analogue of \cite[Lemma 3.1]{racke2014computing}, but powered by~\cite{ADK25}’s tool so that the certification we obtain when $F$ is good is expansion (according to Definition~\ref{def:expanding-set}) rather than being flow-linked. 

    \paragraph{Routing from $F$ to the boundary.} Once $F$ is fixed, we connect $F$ to the boundary edges similarly to~\cite{racke2014computing}. 
    We compute a cut $Y$ (via an auxiliary approximate maximum flow problem) that separates $F$ from $B$. For this purpose, we use the fair-cuts algorithm of~\cite{LNPSsoda23}, which serves as an efficient approximate max-flow solver and provides slightly better congestion (see Lemma~\ref{lemma:separator}) compared to \cite[Lemma 3.2]{racke2014computing}, as well as simplifying the argument. 
    
    \medskip
    The usage of~\cite{ADK25} rather than a procedure inspired by the cut-matching game of~\cite{khandekar2009graph} in~\cite{racke2014computing} allows us to improve the quality of the cut-sparsifier by a $\log n$ factor. Specifically, we still pay $O(\frac{\log n}{\uptau})=O(\log^2 n)$ like in Section \ref{subsec:recap-refinement}, but the additional factor from the cut-matching component is now $O(\log n)$ instead of $O(\log^2 n)$. This gives a tree cut-sparsifier of quality of $O(\log^3 n)$, and establishes Theorem \ref{theorem:final}. 
    In particular, as discussed after Theorem~\ref{theorem:final}, a tree cut-sparsifier of quality $\alpha$ implies a tree flow-sparsifier of quality $O(\alpha \log n)$.
    Therefore, this tree is also a tree flow-sparsifier of quality $O(\log^4 n)$, matching the quality of the flow-sparsifier of~\cite{racke2014computing}.

    The construction of the tree cut-sparsifier is detailed in Section \ref{section:merge-phase}. Importantly, we use the modular building blocks of the balanced sparse cut procedure~\cite{ADK25} and the fair-cuts algorithm \cite{LNPSsoda23} as a black box. This significantly simplifies the construction of the tree cut-sparsifier, when compared to, \eg, \cite{racke2014computing}, as it abstracts away many details of the underlying algorithmic tools. 

    In Section \ref{section:charging-scheme} we show that the constructed tree is indeed a tree cut-sparsifier. There, we explicitly define the invariant that holds when we progressively move the demands around the graph (following the structure of the tree). Additionally, we use a careful charging scheme to rigorously bound the quality of the tree cut-sparsifier.

    \subsection{A Tree Cut-Sparsifier of Quality \texorpdfstring{$O(\log^2 n \log\log n)$}{O(log\^2 n * log log n)}}\label{subsec:full_construction_intuition}
    Following the naming of \cite{harrelson2003polynomial, RackeS14}, we call the partitioning algorithm from Section~\ref{subsec:first-step-summary}, that takes $S$ and returns $\{S_i\}_{i=1}^t$, the \emph{merge phase}. As explained in Section \ref{subsec:first-step-summary}, this phase does \emph{not} guarantee that the inter-cluster edges $F$ and the boundary edges $B$ are expanding together, just that the edges of $F$ are expanding, and we can route from $F$ to $B$.
    Recall from Section \ref{subsec:recap-refinement} that because the distribution of demands on the boundary edges is not uniform, source and target demands do not cancel-out, so the load on the boundary edges accumulates as we go up along the tree. This degrades the quality of the tree cut-sparsifier by an $O(\log n)$ factor.
    To get the improved Theorem \ref{theorem:final-improved}, in Section \ref{section:improvement} we show how to implement a version of the \emph{refinement phase} (see Section~\ref{subsec:recap-refinement}) in near-linear time.  
    Specifically, we show that given a cluster $S$ we can partition it to sub-clusters $\{S_1,\ldots,S_t\}$ such that each sub-cluster respects the all-to-all flow problem on its boundary edges. This is a slightly weaker requirement than the precondition of~\cite{bienkowski2003practical} and the requirement of~\cite{harrelson2003polynomial, RackeS14}.
    Additionally, we show that we can route from the inter-cluster edges of this refinement partition to the boundary of $S$. To perform this partition in near-linear time, we again use the result of~\cite{ADK25}
    which allows us to either find a balanced sparse cut (with respect to the all-to-all flow problem on the boundary of the current sub-cluster) or ensure that the boundary edges are expanding. Additionally, it lets us route from the sparse cut's edges to the boundary edges of one of the sides of the cut. Crucially, however, contrary to the merge phase, we get no guarantee that the size $|S_i|$ is smaller than $|S|$ by a constant factor (and in fact it may be that the partition is simply $\{S\}$).
    
    Using the refinement phase, we amend the hierarchical decomposition tree by performing this refinement after each call to the merge phase algorithm. This gives a decomposition which alternates between a refinement phase partition, 
    that allows to spread the demand uniformly on the boundary,
    and a merge phase partition that shrinks the clusters. Because each sub-cluster $S_i$ in the refinement phase respects the all-to-all demand matrix on its boundary edges, we can move the demands to a uniform distribution on these boundary edges. Furthermore, since $S_i\in V_T$ (and by assumption $Q$ is routable in $T$), the cut $E(S_i, V\setminus S_i)$ respects $Q$,  we get that at most $1$ unit of mass per edge needs to leave the cluster. This allows us to ``reset'' the load on each boundary edge to be $1$ and avoid accumulating the loads on the boundary edges as we go up the tree. By setting $\uptau = 1$ in the merge phase, we improve the degradation (in quality) incurred by the merge phase by a factor of $O(\log n)$, to be $O(\log^2 n)$ (in total over all sub-clusters), instead of $O(\log^3 n)$. This is because we pay $O(\frac{\log n}{\uptau})=O(\log n)$, plus an additional $O(\log n)$ term for the cut-matching component like before.
    As for the degradation in quality incurred by the refinement phase, we adapt the analysis of \cite{harrelson2003polynomial, RackeS14} to show that it is bounded by $O(\log^2 n \log \log n)$ (in total over all sub-clusters), which gives Theorem \ref{theorem:final-improved}. In particular, our tree is also a tree flow-sparsifier of quality $O(\log^3 n \log\log n)$, which gives Theorem \ref{theorem:final-improved-flow}.

    The rest of the paper is structured as follows. Section \ref{section:preliminaries} contains additional definitions and notations, as well as some algorithmic tools we use. In Section \ref{section:merge-phase} we detail the merge phase partition and the hierarchical decomposition for Theorem \ref{theorem:final}. In Section \ref{section:charging-scheme} we show how to update the demand states as we go up the tree in each application of the merge phase, and show that the resulting tree indeed suffices for Theorem \ref{theorem:final}. In Section \ref{section:improvement} we detail the refinement phase and use it to prove Theorem \ref{theorem:final-improved}.

    \section{Additional Preliminaries}
	\label{section:preliminaries}
    \subsection{General Definitions}
    We begin with defining the subdivision graph, which allows us to formally reason about the edges of a graph carrying demand and expanding.
	
	\begin{definition}[Subdivision graph]
	\label{def:subdivision_graph}
	    Given a graph $G=(V,E)$, the \emph{subdivision graph} is denoted by $G'=(V', E')$. It is created by adding a \emph{split vertex} in the middle of each edge, replacing it with a path of length $2$. Formally, $V' \defeq V\cup X_E$, where $X_E$ are the split vertices $X_E \defeq \{x_e\mid e\in E\}$. The new edges are $E' \defeq \left\{(x_e, v)\mid v\in e\right\}$. Given a set of edges $F\subseteq E$, the corresponding split vertices are denoted by $X_F \defeq \{x_f \mid f\in F\}$.
	\end{definition}
    
	\begin{definition}[Indicator Measure]
	\label{def:vertex_measure_of_subset}
	   Let $G=(V,E)$ be a graph and let $U\subseteq V$ be a subset. We define the \emph{weighting induced by $U$} as $\mu_{U} : V \to \RR_{\ge 0}$, such that $\mu_U = \1_U$. In other words, for each $v\in V$, $\mu_{U}(v) = 1$ if $v\in U$, and $0$ otherwise.
	\end{definition}
    In almost all of our use cases, we will only use such indicator measures. In these cases, one should think of the measure $\mu$ as a way of distinguishing a subset of the vertices.

    \subsection{Expansion}

    The following definition extends the standard notion of conductance to the vertex-weighted setting, where some vertices are more important than others. Specifically, when $\mu$ is an indicator measure, it lets us focus  on a subset of the vertices (for example, those corresponding to flow terminals).

	\begin{definition}[Expansion with Vertex Weights]
	\label{def:expansion_with_mu}
    	Let $G=(V,E)$ and $\emptyset\neq S \subset V$. Let $\mu : V \to \RR_{\ge 0}$ be a weighting of the vertices. The expansion of the cut $(S, \bar{S})$ with respect to $\mu$ (sometimes called $\mu$-expansion), denoted by  $\Phi^{\mu}_{G}(S,\bar{S})$, is defined as
    	\[
            \Phi^{\mu}_{G}(S,\bar{S}) \defeq \frac{\capacity(S,\bar{S})}{\min(\mu(S), \mu(\bar{S}))},
    	\]
    	whenever $\min(\mu(S), \mu(\bar{S})) > 0$. A cut with low $\mu$-expansion is called \emph{$\mu$-sparse} (the exact bound on the expansion is given in the context). 
        
        The expansion of $G$ with respect to $\mu$ (sometimes called $\mu$-expansion) is defined to be
    	\[
            \Phi^{\mu}(G) \defeq \min_{\substack{S\subseteq V\\\min(\mu(S), \mu(\bar{S})) > 0}}\Phi^{\mu}_{G}(S,\bar{S})
    	\]
    \end{definition}
    
    The next definition introduces the corresponding notion of an expander, \ie, a graph that maintains sufficiently large expansion across all vertex subsets.
    
    \begin{definition}[Expander]
    \label{def:expander_near_expander}
    	Let $G = (V,E)$ and let $\mu$ be a weighting of the vertices. We say that $G$ is an $\alpha$-expander with respect to $\mu$ if $\Phi^{\mu}(G) \ge \alpha$. 
        
	\end{definition}
    So far, expansion has been defined as a \emph{property of a graph}. 
    In many of our later arguments, however, we talk about 
    \emph{individual subsets} $A \subseteq V$ that expand well 
    \emph{within} a given graph $G$ under some weighting. 
    The next definition captures this idea: it says that a subset $A$, \emph{$\alpha$-expands with respect to $\mu$}, if no cut of $V$ separates $A$ too unevenly, when the imbalance is measured according to the $\mu$-measure of the vertices. This notion generalizes  the $\alpha$-cut-linked notion of~\cite{chekuri2005multicommodity}.

    \begin{definition}[Subset of $V$ Expands]
    \label{def:expanding-set}
        Given a graph $G=(V,E)$, $\alpha>0$, a weighting $\mu : V \to \RR_{\ge 0}$ of the vertices, and a set of vertices $A\subseteq V$, we say that $A$ $\alpha$-expands 
        with respect to $\mu$ in $G$ if
        \[
            \forall \substack{S\subseteq V\\\min(\mu(A\cap S), \mu(A\setminus S)) > 0}: \frac{\capacity(S, V\setminus S)}{\min(\mu(A\cap S), \mu(A\setminus S))} \ge \alpha.
        \]
        In other words, if we let $\mu' = \mu\cdot\mu_A$, then $A\subseteq V$,  $\alpha$-expands in $G$ with respect to $\mu$, 
        if and only if $G$ is an $\alpha$-expander with respect to $\mu'$. 
        When $\mu\equiv 1$, we just say that $A$ $\alpha$-expands in $G$.
    \end{definition}

    Note that in case $\mu = \1_X$ is an indicator measure, some set $A$ is $\alpha$-expanding with respect to $\1_X$ is equivalent to saying that the set $A\cap X$ is $\alpha$-expanding.
    Definition \ref{def:expanding-set} is the cut-analogue of the set $A$ being $\alpha$-flow-linked, in case each vertex $v\in A$ injects $\mu(v)$ units of flow. This approach is useful for us when we want to choose a measure $\mu = \1_X$ in advance and \emph{then} find some set $A$ such that $A\cap X$ expands.

    The next remark bridges the two viewpoints introduced in Definitions~\ref{def:cut_respects_flow} and~\ref{def:expanding-set}: 
    the notion of a set $A$ that $\alpha$-expands with respect to a measure $\mu$, and the notion of a graph $G$ that $\alpha$-respects a corresponding demand matrix defined by $\mu$. 

    \begin{remark}
    \label{remark:mu_expansion_as_problem}
        Let $G = (V,E)$ and let $\mu$ be a weighting of the vertices. Let $A\subseteq V$ be a set of vertices with $\mu(A) > 0$ and let $\mu' = \mu\cdot \mu_A$. Consider the demand matrix $Q$ given by $Q(u, v) = 
        \frac{\mu'(u)\cdot\mu'(v)}{2\mu(A)}$ for each $u, v \in V$. Fix a cut $(S, \bar{S})$. It holds that
        $\dem_Q(S, \bar{S}) = \frac{\mu'(S)\cdot\mu'(\bar{S})}{\mu(A)}$. Assume $\mu'(S) \le \mu'(\bar{S})$. Then, $\frac{1}{2}\le\frac{\mu'(\bar{S})}{\mu(A)}\le 1$. Plugging in,
        \[
            \Phi^{\mu'}(S, \bar{S})\le\frac{\capacity(S, \bar{S})}{\dem_Q(S, \bar{S})}\le2\cdot\Phi^{\mu'}(S, \bar{S}).
        \]

        In other words, if $A$ $\alpha$-expands in $G$ with respect to $\muG$, then it certifies that $G$ $\alpha$-respects the demand matrix $Q$ 
        given above, and conversely if $G$ $\alpha$-respects $Q$, then $A$  $\frac{\alpha}{2}$-expands with respect to $\mu$ in $G$.

        In particular (by setting $\mu \equiv \1$), if $A$ $\alpha$-expands in $G$, then $G$ $\frac{\alpha}{2}$-respects the all-to-all flow problem on $A$ (recall Definition~\ref{def:all_to_all_flow}), and conversely if $G$ $\alpha$-respects the all-to-all flow problem on $A$ then $A$ $\alpha$-expands in $G$.

    \end{remark}

    The following fact shows that if $G[A]$ is an expander with respect to a measure, then $A$ expands in $G$ with respect to this measure. Note that the converse is not true. The flow analogue of this claim is that if we can route an all-to-all flow problem on the vertices of $A$ within $G[A]$ then we can do so in $G$, but not conversely.
    
    \begin{fact}
    \label{fact:expansion_implies_flow_expansion}
        If $G[A]$ is an $\alpha$-expander with respect to $\mu|_A$, then $A$ $\alpha$-expands with respect to $\mu$ in $G$. 
    \end{fact}

    \begin{proof}
        Indeed, let $S\subseteq V$ be such that $\min(\mu(A\cap S), \mu(A\setminus S)) > 0$. Then, 
        \[
            \frac{\capacity(S, V\setminus S)}{\min(\mu(A\cap S), \mu(A\setminus S))} \ge \frac{\capacity(A\cap S, A\setminus S)}{\min(\mu(A\cap S), \mu(A\setminus S))} \ge \alpha~,
        \]
        where the last inequality is due to $(A\cap S, A\setminus S)$ being a cut in $G[A]$.
    \end{proof}

    The following fact shows that expansion is a monotone property, in the sense that subsets of an expanding sets also expand.
    \begin{fact}
    \label{fact:subset_expands}
        Assume $A_1\subseteq A_2\subseteq V$ and $A_2$ $\alpha$-expands with respect to $\mu$. Then $A_1$ $\alpha$-expands with respect to $\mu$.
    \end{fact}
    \begin{proof}
        The denominator satisfies $\min(\mu(A_1\cap S), \mu(A_1\setminus S)) \le \min(\mu(A_2\cap S), \mu(A_2\setminus S))$, which gives the result.
    \end{proof}

    \subsection{Cut-Matching Game}
    \label{subsec:cut-matching}
    
    The following is a recent result of~\cite{ADK25}, which generalizes~\cite{agassy2023expander}. They show that in near-linear time we can get a balanced sparse cut, or certify that a large part of the graph is an expander, with respect to a general vertex measure.\footnote{Theorem \ref{theorem:balanced_cut_mu} requires that $\mu(v)$ will be in $\left[\frac{1}{\poly(n)}, \poly(n)\right]\cup\{0\}$, which will hold for our usages.
    }

    \paragraph{Convention (orientation of cuts).}
    In all balanced-cut results we use (Theorems~\ref{theorem:balanced_cut_mu} and~\ref{theorem:trimming_allows_routing}, and Corollary~\ref{cor:balanced-cut-refined}), we orient a cut $(A,\bar A)$ so that the left side $A$ is conceptually the \emph{larger} side with respect to the relevant measure (typically $\mu$).

    \begin{theorem}[\texorpdfstring{\cite[Theorem 5.2]{ADK25}}{[ADK25, Theorem 5.2]}]
    \label{theorem:balanced_cut_mu}
        Given a graph $G=(V,E)$ of $m$ edges, a parameter $\phi > 0$, and a weighting $\mu : V \to \RR_{\ge 0}$ of the vertices, there exists a randomized algorithm which takes $\tilde{O}(m)$ time and must end in one of the following three cases:
        \begin{enumerate}
            \item We certify that $G$ has $\mu$-expansion $\Phi^{\mu}(G)=\Omega(\phi)$, with high probability. 
            \item We find a cut $(A,\bar{A})$ 
            in $G$ of $\mu$-expansion $\Phi^{\mu}_G(A,\bar{A})=O(\phi\log n)$, and $\mu(A), \mu(\bar{A})$ are both $\Omega\left(\frac{\mu(V)}{\log n}\right)$, i.e, we find a relatively balanced $\mu$-sparse cut. 
            \item We find a cut $(A,\bar{A})$ with $0<\mu(\bar A)\le \frac{\mu(V)}{2}$, $\Phi^{\mu}_G(A,\bar{A})=O(\phi\log n)$, and with high probability, $G[A]$ is an $\Omega(\phi)$-expander with respect to $\mu$. 
        \end{enumerate}
    \end{theorem}

    \paragraph{Routing from a cut.}
        Let $G = (V, E)$, and let $(A, S \setminus A)$ be a cut of a cluster $S \subseteq V$.  
        When we say that we \emph{route flow in $G[A]$ from the cut} $(A, S \setminus A)$ to some destination subset $X \subseteq A$, 
        we mean that there exists a flow $f$ in the subgraph $G[A]$ such that each vertex $v \in A$ sends
            $\deg_{G[S]}(v) - \deg_{G[A]}(v)$
        units of mass, i.e., one unit for every incident edge of $v$ crossing the cut $(A, S \setminus A)$.  
        The specific requirements on the destination (\ie, the amount of mass that each vertex $x\in X$ receives),   
        and congestion of this flow are stated in the corresponding context.

    The authors of~\cite{ADK25} supply additional guarantees in Cases (2) and (3) of Theorem \ref{theorem:balanced_cut_mu} (see \cite[Section 4.3 and Corollaries A.3, A.4, A.5]{ADK25}). We summarize these guarantees in the following theorem.

    \begin{theorem}[Additional Guarantees for Theorem~\ref{theorem:balanced_cut_mu}]
    \label{theorem:trimming_allows_routing}
        The following are additional guarantees for Cases (2) and (3) of Theorem \ref{theorem:balanced_cut_mu}. First, in both cases $\mu(A) \ge \frac{1}{4}\mu(V)$.
        
        In Case (2), there exist constant $c>0$ and a sequence of cuts $S_t\subseteq A_t$, $0\le t < T$ 
        with $T=O(\log^2 n)$, such that $A_0 = V$, $A_{t+1} = A_t \setminus S_t$ (for $t \in \{0,\ldots T-1\}$). The returned cut is simply $(A,\bar{A})=(A_{T},\bigcup_{t=0}^{T-1}{S_t})$.
        The cuts apart from the last one satisfy $\mu\!\left(\bigcup_{t=0}^{T-2}{S_t}\right)=O\!\left(\frac{\mu(V)}{\log n}\right)$. Each $S_t$ is sparse: $\Phi^\mu_{G[A_t]}(A_{t+1}, S_t) \le c\phi \log n$. 
        Moreover, for every $t\in\{0,\ldots, T-1\}$,
        \begin{itemize}
            \item 
            There exists a flow routable with congestion $O(1)$ in $G[S_t]$, from the cut  $(A_t\setminus S_t, S_t)$ 
            to $S_t$, such that each vertex  $v\in S_t$ receives at most $O(\phi\log n\cdot\mu(v))$ units. 
            \item There exists a flow routable with congestion $O(1)$ in $G[A_t\setminus S_t]$, from the cut $(A_t\setminus S_t, S_t)$ to $A_t\setminus S_t$, such that each vertex $v\in A_t\setminus S_t$ receives at most $O(\phi\log n\cdot\mu(v))$ units. 
        \end{itemize}
        In other words, in each step we can route flow from the cut edges $E(A_t\setminus S_t, S_t)$ to either the side $S_t$ or $A_t\setminus S_t$. 
                    
        In Case (3), this sequence of cuts exists as well, but a ``trimming step" is run to get some larger cut $\bar{A}\supseteq \bigcup_{t=1}^{T-1} S_t$ which is returned. Moreover,
        \begin{itemize}
            \item There exists a flow routable with congestion $O(1)$ in $G[A]$, from the cut $(A,\bar{A})$ 
            to $A$, such that each vertex  $v\in A$ receives at most $O(\phi\cdot\mu(v))$ units. 
            \item There exists a flow routable with congestion $O(1)$ in $G[\bar{A}]$, from the cut $(A,\bar{A})$ to $\bar{A}$, such that each vertex $v\in \bar{A}$ receives at most $O(\phi\log n\cdot\mu(v))$ units.
        \end{itemize}
    \end{theorem}

    \begin{remark}
        In Case (3) of Theorem~\ref{theorem:balanced_cut_mu}, the fact that $G[A]$ is an $\Omega(\phi)$-expander with respect to $\mu$, implies that $A$ $\Omega(\phi)$-expands with respect to $\mu$ in $G$ by Fact~\ref{fact:expansion_implies_flow_expansion}. 
    \end{remark}

    \subsection{Demand States}

    When we describe the solution to a multi-commodity flow problem, it is useful to talk about \emph{intermediate} stages of the flow. Analogously, we wish to talk about intermediate states of the distribution of demands when showing that a graph respects a given demand matrix.
    Intuitively, we imagine the vertices of the graph holding some demand of each commodity, which can be positive or negative. These demands move from one vertex to another according to demand matrices (\ie, multi-commodity flow problems), incurring some congestion (or, in the cut-sparsifier case, \emph{charge}, as detailed in Section \ref{section:charging-scheme}) on the edges of the graph. When a positive demand meets a negative demand of the same commodity, they cancel out. This scheme is made formal in the definitions in this section.
    
    \begin{definition}[Demand State]
    \label{def:flow_state}
        Consider a set of commodities $\K$. 
        The \emph{demand state of $v$}, $v\in V$, is a vector, $P(v)$, 
        where $P(v,k)$, $k\in \K$, defines the amount of commodity $k\in \K$ which resides at $v\in V$ and needs to be shipped out. Note that 
        $P(v,k)$ may be either positive or negative, and in case it is negative, the commodity needs to be shipped into $v$. We denote by $\norm{P(v)}_1=\sum_{k\in\K}{|P(v,k)|}$ the \emph{load} on $v$.

        A \emph{demand state $P$ of $G$} (or simply a \emph{demand state} when $G$ is clear from context) is a collection of demand states, one for each $v\in V$, \ie, $P\in \RR^{V\times \K}$, where 
        the row of vertex $v$ is the demand state of $v$.

        For a set $A \subseteq V$, we say that a demand state $P$ on $G$
        is \emph{supported on $A$} if for all $k\in \K$, $P(v, k) = 0$ whenever $v\notin A$.
    \end{definition}

    As explained below in Remark \ref{remark:flow_problem_as_flow_state}, demand states generalize demand matrices.
    Unlike demand matrices (see Definition~\ref{def:flow_problem}), where each vertex serves as the source of its own commodity and source-sink pairs are explicitly encoded in the matrix, in a demand state the commodities are independent of the vertices, and multiple vertices may act as sources of the same commodity (similarly, the same vertex may act as a source of multiple commodities).

    Intuitively, we think of a demand state as the current state of mass residing at the vertices of the graph, where positive and negative demands ``await'' to be canceled-out, while a multi-commodity flow problem and its demand matrix describe an explicit \emph{movement} of mass between pairs of vertices.

    \begin{definition} [Valid Demand State]
    \label{def:valid_flow_state}
        We say that a demand state on $G$ is \emph{valid} if $\sum_{u\in V}{P(u, k)} = 0$ for each $k\in\K$. That is, in total, no mass
        of commodity $k$ is created or destroyed.

        The \emph{demand} of $P$ across a cut $(A, B=\bar{A})$ is 
        \[
            \dem_P(A, B)\defeq\sum_{k\in \K}{\left|\sum_{u\in A}{P(u, k)}\right|} = \sum_{k\in \K}{\left|\sum_{u\in B}{P(u, k)}\right|}.
        \]

        We extend the notion of a cut (resp., graph) respecting a demand matrix, to a cut (resp., graph) respecting valid demand state using this definition: A cut $(S, \bar{S})$ in $G$ $\alpha$-respects $P$ if $\capacity(S, \bar{S})\ge\alpha\cdot\dem_P(S, \bar{S})$. $G$ $\alpha$-respects $P$ if
        \[
            \min_{\substack{S\subseteq V\\\dem_P(S, \bar{S}) > 0}} \frac{\capacity(S, \bar{S})}{\dem_P(S, \bar{S})} \ge \alpha.
        \]
    \end{definition}

    Throughout the paper we manipulate many demand states that are not themselves valid. Such demand states typically arise by restricting a valid demand state to a subgraph (or cluster), that is, by retaining only the demand vectors associated with vertices in the cluster.
    Although these restricted demand states may violate the validity condition, this poses no issue: they should be viewed as a \emph{partial} or \emph{intermediate} collection of demand vectors. 

    \begin{remark}
    \label{remark:flow_problem_as_flow_state}
        Note that the notion of a valid demand state generalizes the notion of a demand matrix from Definition~\ref{def:flow_problem}. Indeed, a demand matrix is a $V\times V$ matrix $Q$ 
        where $Q(u, v)$ is the amount of $u$'s commodity that has to go to $v$. Such a specification corresponds to a valid demand state $P$ in which $\K = V$, the total amount shipped out 
        of $u$ of its own commodity is $P(u, u) = \sum_{v\in V\setminus \{u\}}{Q(u, v)}$, and the amount shipped to $v$ of $u$'s commodity is $P(v, u) = -Q(u, v)$, $v\neq u$ (note that the commodity is shipped \emph{into} $v$, so $P(v,u)$ is negative). The standard definitions of congestion and demand across a cut of $Q$ equal those of $P$. 

        On the other hand, given a valid demand state, we can define the corresponding demand matrix (according to Definition \ref{def:flow_problem}) by adding to the graph a source vertex and sink vertex corresponding to each commodity, with appropriate additional edges capacitated according to the demands in $P$. The latter reduction is not be used in this paper, so we omit the full details.
    \end{remark}
    
    \begin{definition}[Sum and Negation]
    \label{def:sum_of_valid_states}
        We define the \textit{sum} of two demand states $P_1$ and $P_2$ on the same set of commodities $\K$ to be  $(P_1+P_2)(u, k) = P_1(u, k) + P_2(u, k)$. Note that if $P_1$ and $P_2$ were valid demand states, then $P_1+P_2$ is also a valid demand state.     

        We define the \textit{negation} $-P$ of a demand state $P$ on the set of commodities $\K$ by $(-P)(u, k) = -P(u, k)$. Note that if $P$ was a valid demand state, then $-P$ is also a valid demand state.
    \end{definition}

    The following fact, which we use extensively, is a direct consequence of Definitions~\ref{def:valid_flow_state} and~\ref{def:sum_of_valid_states}: 
    \begin{fact}
    \label{fact:extended_flow_problems}
        Fix a cut $(S, \bar{S})$. 
        For two valid demand states $P_1, P_2$ on the commodities $\K$, we have $\dem_{P_1 + P_2}(S, \bar{S}) \le \dem_{P_1}(S, \bar{S}) + \dem_{P_2}(S, \bar{S})$.

        Additionally, if $G$ $\alpha_1$-respects $P_1$ and $\alpha_2$-respects $P_2$, then it $\frac{1}{\frac{1}{\alpha_1}+\frac{1}{\alpha_2}}$-respects $P_1+P_2$.
        
    \end{fact}

    Suppose we have a demand state $P$. Given a demand matrix $Q$ (recall Definition \ref{def:flow_problem}), 
    we can ``move'' the demand vectors in $P$ according to the requirements of $Q$. 
    Each vertex $u\in V$ will send $Q(u, v)$ units in total to $v\in V$, and we split the demands evenly according to the distribution of commodities in the demand vector $P(u)$. In other words, $\frac{P(u,k)}{\norm{P(u)}_1}\cdot Q(u,v)$ units of commodity $k\in \K$ will be shipped to $v$. If $P(u,k) < 0$ this corresponds to shipping units of commodity $k$ from $v$ to $u$.
    This results in a demand state $P^{\uparrow Q}$, called the updated demand state:

    \begin{definition}[Updated Demand State]
    \label{def:augment_state}
        Let $P$ be a (not necessarily valid) demand state on the set of commodities $\K$ and let $Q$ be a demand matrix (\ie, a multi-commodity flow problem). 
        We define the \emph{updated} demand state $P^{\uparrow Q}$ on $\K$ as follows:
        \[
            P^{\uparrow Q}(u, k) = P(u,k) - \frac{P(u,k)}{\norm{P(u)}_1}\cdot\sum_{v\in V}{Q(u,v)} + \sum_{v\in V}{\frac{P(v,k)}{\norm{P(v)}_1}\cdot Q(v,u)} .
        \]
    \end{definition}

    \begin{example}
    \label{example:updating-demands}
        Suppose there are three commodities and the demand states of three vertices $u,v,w$ are
        \[
        P(u)=(20,-70,10),\qquad
        P(v)=(6,2,-2),\qquad
        P(w)=(-4,12,-4).
        \]
        Consider a demand matrix $Q$ that sends $30$ units from $u$ to $v$ and $10$ units from $w$ to $u$.
        Sending $30$ units from $u$ to $v$ moves $30\cdot (0.2,-0.7,0.1)=(6,-21,3)$ units of demand from $u$ to $v$.
        Sending $10$ units from $w$ to $u$ moves $10\cdot(-0.2,0.6,-0.2)=(-2,6,-2)$ units of demand from $w$ to $u$.
        
        Applying both movements simultaneously, the updated demand vectors are
        \[
        P^{\uparrow Q}(u)=(12,-43,5),\qquad
        P^{\uparrow Q}(v)=(12,-19,1),\qquad
        P^{\uparrow Q}(w)=(-2,6,-2).
        \]
    \end{example}

    We  make frequent use of the following fact. It states that because we only move the mass around, no mass of any commodity is created or destroyed. This means that the total mass of each commodity is the same in $P$ and $P^{\uparrow Q}$, so it is zero in $P - P^{\uparrow Q}$. Therefore, $P - P^{\uparrow Q}$ is a valid demand state. The full proof of the fact is deferred to Appendix~\ref{appendix:omitted-proofs}.

    \begin{fact}
    \label{fact:basic_flow_state}
        In the notation of Definition~\ref{def:augment_state}: 
        \begin{enumerate}
            \item $P-P^{\uparrow Q}$ is a valid demand state.
            \item For any cut $(S, \bar{S})$ of $V$, $\dem_{(P-P^{\uparrow Q})}(S, \bar{S})\le\dem_Q(S, \bar{S})$.
        \end{enumerate}
    \end{fact}

    The next lemma presents two basic types of demand-state updates that will be used extensively throughout the analysis. The first corresponds to uniformly redistributing mass within a set, and the second to pushing all mass out of a set. Both updates arise frequently when we manipulate demand states.

    \begin{lemma}
    \label{lemma:special_flows}
        Let $P$ be a (not necessarily valid) demand state on the set of commodities $\K$. Then:
        \begin{enumerate}
            \item Suppose
            $Q$ is a \emph{scaled all-to-all flow problem} on some subset $A\subseteq V$, such that $Q(u, v) = \frac{\norm{P(u)}_1}{|A|}$ if $u, v\in A$, and  $Q(u,v) = 0$ otherwise. Then, $P^{\uparrow Q}(v) = \frac{1}{|A|}\sum_{u\in A}{P(u)}$ for all $v\in A$, and $P^{\uparrow Q}(v) = P(v)$ for $v\notin A$.
            That is, $P^{\uparrow Q}$ distributes the mass in $A$ uniformly. 
            \item Suppose $Q$ moves demands from a set $A\subseteq V$ to $V\setminus A$: for all $u\in V$ and $v\in A$, $Q(u, v) = 0$, and moreover, for all $u\in A$, $\sum_{v\in V\setminus A}{Q(u, v)} = \norm{P(u)}_1$. In this case, $P^{\uparrow Q}$ is supported on $V\setminus A$.
        \end{enumerate}
    \end{lemma}

    \begin{proof}
        ~\\
        \begin{enumerate}
            \item 
                Let $v\in V$. If $v\notin A$, then $Q(v,u)=Q(u,v)=0$ for all $u\in V$, and therefore $P^{\uparrow Q}(v) = P(v)$. Assume now that $v\in A$, then
                \begin{align*}
                    P^{\uparrow Q}(v) &= P(v) - \frac{P(v)}{\norm{P(v)}_1}\cdot\sum_{u\in V}{Q(v,u)} + \sum_{u\in V}{\frac{P(u)}{\norm{P(u)}_1}\cdot Q(u,v)} \\
                    &=
                    P(v) - \frac{P(v)}{\norm{P(v)}_1}\cdot\sum_{u\in A}{\frac{\norm{P(v)}_1}{|A|}} + \sum_{u\in A}{\frac{P(u)}{\norm{P(u)}_1}\cdot \frac{\norm{P(u)}_1}{|A|}} \\
                    &= P(v)- P(v) + \frac{1}{|A|}\sum_{u\in A}{P(u)}
                    \\&= \frac{1}{|A|}\sum_{u\in A}{P(u)}.
                \end{align*}
            \item
                Let $v\in A$, then
                \begin{align*}
                    P^{\uparrow Q}(v) &= P(v) - \frac{P(v)}{\norm{P(v)}_1}\cdot\sum_{u\in V}{Q(v,u)} + \sum_{u\in V}{\frac{P(u)}{\norm{P(u)}_1}\cdot Q(u,v)} \\
                    &=
                    P(v) - \frac{P(v)}{\norm{P(v)}_1}\cdot\sum_{u\in V\setminus A}{Q(v,u)}
                    =0.
                \end{align*}
        \end{enumerate}
    \end{proof}

	\section{The Merge Phase}
    \label{section:merge-phase}

    In this section we detail our algorithm for constructing the hierarchical decomposition for Theorem \ref{theorem:final}, whose main partitioning step will be an important tool for the proof of Theorem \ref{theorem:final-improved} as well.    
    
    Our construction closely follows that of \cite{racke2014computing}. For the readers familiar with \cite{racke2014computing}, we mention the main differences:
    \begin{enumerate}
        \item The ``first partition" step (Theorem~\ref{theorem:merge-phase_part1}) is translated to the cut-sparsifier case (which improves the quality of the cut-sparsifier by a factor of $\log n$), and its proof is significantly simplified, as the ``heavy algorithmic tools'' of a \emph{balanced cut or expander} theorem \cite{ADK25} and a \emph{fair cuts} computation \cite{LNPSsoda23} are used in a black box manner.
        \item The proof of the main lemma used in the second partition step (Lemma~\ref{lemma:fair-cut-separator}) is improved and simplified by using the results of \cite{LNPSsoda23}.
    \end{enumerate}

    As explained in Section \ref{section:techniques}, the hierarchical decomposition is constructed by recursively applying the ``merge phase'' partitioning procedure to input clusters $S\subseteq V$. Each input cluster $S\subseteq V$ is partitioned in two levels. 
    At the first level, we define the direct children of $S$ to be disjoint clusters $L$ and $R$\footnote{In may happen that $R=\emptyset$. In this case, no node corresponding to $R$ is added to the hierarchical tree $T$. Similarly, we do not add nodes corresponding to the empty set to the hierarchy.} such that $L\cup R = S$, and at the second level the children of $L$ and of $R$, respectively, are defined to be disjoint clusters $L_1, \ldots, L_\ell$ and $R_1, \ldots, R_r$ (satisfying $L_1\cup \cdots \cup L_\ell = L$ and $R_1 \cup \cdots \cup R_r = R$). 
    The size of the clusters $L_i$ and $R_i$ is bounded as $|L_i|\le\frac{2}{3}|S|$ and $|R_i|\le\frac{2}{3}|S|$.

    The two-level partitioning 
    is produced in two steps. First (Section~\ref{section:merge-phase_part1}), we define a partitioning of the cluster $S$ into disjoint sets $Z_1, \ldots, Z_z$, where each cluster satisfies $|Z_i| \le \tfrac{2}{3}|S|$.
    Next (Section~\ref{section:merge-phase-part2}), we partition $S$ into two disjoint sets $L, R$.
    The second-level clusters $L_1, \ldots, L_\ell$ and $R_1, \ldots, R_r$ are defined as $L_i = L\cap Z_i$ and $R_i = R\cap Z_i$, taking only non-empty intersections (see Figure~\ref{fig:merge_phase_high_level}).

    \begin{figure}[t]
        \begin{center}
            \includegraphics[scale=0.5]{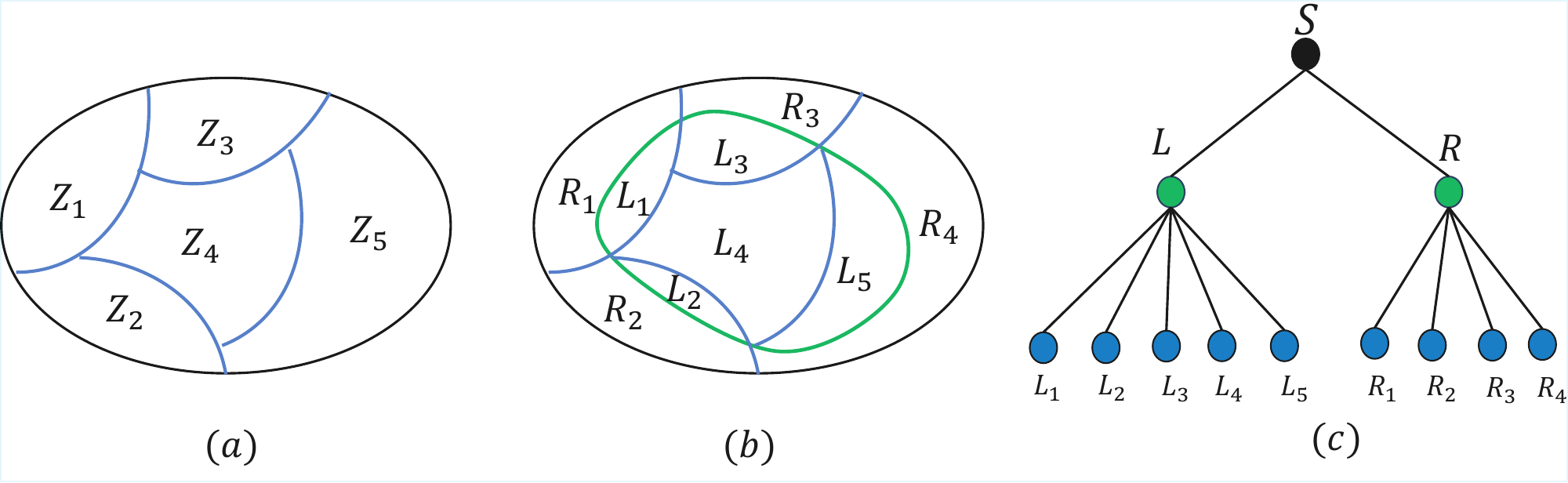}
        \end{center}
        \caption{
        (a) Illustration of the first partitioning step, where a cluster $S$ is divided into smaller clusters $Z_1, \ldots, Z_z$. 
        (b) Illustration of the second partitioning step, where the cluster $S$ is first split into two disjoint clusters $L$ and $R$ (shown in green), and then each of $L$ and $R$ is further subdivided according to the $Z$-partition into sub-clusters $L_1, \ldots, L_\ell $ and $R_1, \ldots, R_r$ (shown in blue). 
        (c) The corresponding hierarchical-tree view, where $S$ has two children $L$ and $R$, each of which splits into its second-level sub-clusters.}
        \label{fig:merge_phase_high_level}
    \end{figure}

    The intuition behind the two-level partition is as follows. As discussed in Section \ref{section:techniques}, we  want to partition $S$ into sub-clusters $\{S_1,\ldots,S_t\}$ of size $|S_i|\le c\cdot |S|, c < 1$, so that the inter-cluster edges $F$ and the boundary edges $B$ of $S$, satisfy that $F\cup B$ expands. However, this may not be possible for a cluster $S$ if $B$ does not expand in $G[S]$. Consequently, the first partition into $Z_1, \ldots,Z_z$ ignores $B$ and focuses on reducing the size of the clusters, while ensuring that $F$, the set of inter-cluster edges, expands. 
    In Section \ref{section:charging-scheme}, we describe how to move the demands that need to leave each cluster $S\in V_T$, to the cluster's boundary edges. Specifically, we wish to move demands that reside on $F$ (\ie, boundary edges of sub-clusters) to $B$. Therefore, we use a \emph{fair-cut} computation~\cite{LNPSsoda23} to find an approximate min-cut, $Y\subseteq E$, that blocks the flow problem of routing from $F$ to $B$. We then add this cut to the tree (\ie, find a partition $L\cup R$ induced by $Y$), thereby forcing the given demand matrix to respect this cut. This allows us to bound the amount of mass that has to cross the cut $Y$. See Section \ref{section:charging-scheme} for details on the usage of this partition.

    As mentioned above, the first partition into $Z_1, \ldots, Z_z$ ignores the boundary edges of the cluster and focuses only on the induced subgraph $G[S]$. Afterwards, the second partition\footnote{Note that the partition into $L\cup R$ corresponds to the first layer of the hierarchical tree (see Figure~\ref{fig:merge_phase_high_level}).} into $L\cup R$, takes into account the boundary edges of $S$ (\ie, edges of $E_G(S, V\setminus S)$), as well as the inter-cluster edges from the first partition.

    \subsection{Partition of \texorpdfstring{$G[S]$}{G[S]}}
    \label{section:merge-phase_part1}
    
    We first partition the cluster $S$ into clusters $Z_1, \ldots, Z_z$ 
    such that edges between the clusters $Z_i$ are ``well-connected'' in $G$, in the sense that the following theorem make precise. Furthermore, the clusters $Z_i$ are small (that is, each $Z_i$ satisfies $|Z_i|\le\frac{2}{3}|S|$). Formally,
    \begin{theorem}[Improved version of \texorpdfstring{\cite[Theorem 3.1]{racke2014computing}}{[RST14, Theorem 3.1]}]
    \label{theorem:merge-phase_part1}
        There is an algorithm \emph{merge-phase-1}
        that partitions the induced subgraph $G[S]$ of a cluster $S$ into disjoint connected components $Z_1, \ldots, Z_z$ such that:
        \begin{enumerate}
            \item $Z_1\cup\cdots\cup Z_z = S$,
            \item $\forall i\in \{1, \ldots, z\}: |Z_i|\le\frac{2}{3}|S|$, and
            \item with high probability, in the subdivision  graph $G[S]'$ (recall Definition~\ref{def:subdivision_graph}), the set of split vertices corresponding to the inter-cluster edges
            \[
                X_F=\left\{x_{(u, v)} \mid u\in Z_i, v\in Z_j, (u,v)\in E, i\neq j\right\}
            \]
             $\alpha$-expands 
             in $G[S]'$ 
             for $\alpha = \Omega\!\left(\frac{1}{\log n}\right)$.
             
        \end{enumerate}
        The algorithm runs in time $\tilde{O}(m)$ where $m$ denotes the number of edges of $G[S]$.
    \end{theorem}
    
    Note that Theorem \ref{theorem:merge-phase_part1} has $\alpha = \Omega\!\left(\frac{1}{\log n}\right)$, instead of $\alpha = \Omega\!\left(\frac{1}{\log^2 n}\right)$ in \cite[Theorem 3.1]{racke2014computing}, which leads (as shown below) to the improvement by a factor of $\log n$ in the quality of the tree cut-sparsifier. The partitioning algorithm in Theorem \ref{theorem:merge-phase_part1} is independent of the edges outside $G[S]$. Therefore, we can restrict our attention to the induced subgraph $G[S]$: in the rest of this section, $G=(V,E)$ refers to $G[S]$, and $n, m$ denote the number of vertices and edges, respectively, in $G[S]$. Following the notation of \cite{racke2014computing}, we say that a subset $F\subseteq E$ of edges induces a \emph{balanced clustering} in $G$, if deleting $F$ from $G$ leaves connected components with at most $\frac{2}{3}|V|$ vertices in each.

    To establish Theorem~\ref{theorem:merge-phase_part1}, we state the following lemma, which improves and unifies Lemmas 3.1 and 3.2 in \cite{racke2014computing}. We defer the proof of Lemma {\ref{lemma:separator}} to the end of the section. 
    \begin{lemma}[Improved version of \texorpdfstring{\cite[Lemmas 3.1 and 3.2]{racke2014computing}}{[RST14, Lemmas 3.1 and 3.2]}]
    \label{lemma:separator}
        Given a subset $F\subseteq E$ that induces a balanced clustering, we can find in time $\tilde{O}(m)$, either
        \begin{enumerate}
            \item a smaller edge set $\tilde{F}\subseteq E$ that induces a balanced clustering, with $|\tilde{F}|\le|F|\cdot\left(1 - \frac{c}{\log n}\right)$ for some constant $c > 0$, \textbf{or}
            \item an edge set $\tilde{F}=\tilde{A}\cup \tilde{C}$, 
            such that $\tilde{F}$ induces a balanced clustering, and with high probability $X_{\tilde{A}}$, $\Omega(1/\log n)$-expands in $G'$.
            Additionally there exists a flow in $G'$ from nodes in $X_{\tilde{C}}$ to $X_{\tilde{A}}$ that can be routed with congestion $O(1)$ and:
            \begin{itemize}
                \item every node $x_c\in X_{\tilde{C}}$ sends $1$ unit of flow;
                \item every node $x_a\in X_{\tilde{A}}$ receives at most $1$ unit of flow.
            \end{itemize}
        \end{enumerate}
    \end{lemma}

    The following lemma shows that in Case~(2) of Lemma~\ref{lemma:separator}, the set $X_{\tilde{F}}$, $\Omega(1/\log n)$-expands in~$G'$.

    \begin{lemma}
    \label{lemma:separator-general}
        Let $G$ be a graph, and let $A, C\subseteq V$. Assume that $A$ $\phi$-expands in $G$, and additionally that there exists a flow routable with congestion $c$, such that each vertex $v\in C$ sends $1$ unit of flow, and each vertex $u\in A$ receives at most $a$ units (other vertices receive no flow). 
        Then $A\cup C$  $\frac{\phi}{2\cdot(2+2a+c\phi)}$-expands in $G$.
    \end{lemma}

    Before proving Lemma \ref{lemma:separator-general}, we show how Lemmas~\ref{lemma:separator} and~\ref{lemma:separator-general} imply Theorem~\ref{theorem:merge-phase_part1}.
    
    \begin{proof}[Proof of Theorem~\ref{theorem:merge-phase_part1}]
        We begin with $F = E$ (which induces a balanced clustering) and repeatedly invoke Lemma~\ref{lemma:separator} on $F$ until it terminates with Case (2). Each time the lemma terminates with Case (1), $|F|$ shrinks by a factor of at least $1 - \frac{c}{\log n}$, and still induces a balanced clustering. Consequently, after at most $O(\log^2 n)$ iterations, the procedure must terminate with Case (2).

        Once Lemma \ref{lemma:separator} terminates with Case (2), $X_{\tilde{A}}$, $\Omega(1/\log n)$-expands in $G'$, and we can route in $G'$ with congestion $O(1)$ a flow such that each node in $X_{\tilde{C}}$ sends $1$ unit of flow and each node in $X_{\tilde{A}}$ receives at most $1$ unit of flow.
        By Lemma~\ref{lemma:separator-general}, 
        this implies that edges in $\tilde{F}=\tilde{A}\cup \tilde{C}$ satisfy that $X_{\tilde{F}}$,  $\Omega(1/\log n)$-expands in $G'$. Additionally, $\tilde{F}$ induces a balanced clustering whose connected components will be $\{Z_1,\ldots, Z_z\}$.\footnote{After computing $\tilde{F}$ and the resulting partition $\{Z_1,\ldots,Z_z\}$, we re-define $F$ to be the new set of inter-cluster edges $F \defeq \{(u,v) \mid u\in Z_i, v\in Z_j, (u,v)\in E, i\neq j\}$. This set satisfies $F\subseteq \tilde{F}$ and thus $\Omega(1/\log n)$-expands in $G[S]'$ (see Fact \ref{fact:subset_expands}).}
        This proves Theorem~\ref{theorem:merge-phase_part1}.
    \end{proof}

    Next, we prove the deferred lemmas, starting with Lemma~\ref{lemma:separator-general} and concluding with Lemma~\ref{lemma:separator}.
    We begin by giving intuition for Lemma~\ref{lemma:separator-general}. By Remark~\ref{remark:mu_expansion_as_problem}, it is enough to prove that $G$ $\frac{\phi}{2\cdot(2+2a+c\phi)}$-respects the all-to-all flow problem between the nodes of $A\cup C$ (see Definition~\ref{def:all_to_all_flow}).

    Using the fact that $A$ expands in $G$ (thus $G$ respects the all-to-all flow problem on $A$) and that there is a low-congestion flow $f$ from $C$ to $A$ (in particular $G$ respects the flow problem that moves mass from $C$ to $A$ according to $f$), we observe that the all-to-all demand matrix on $A \cup C$ can be realized by the following sequence of demand transitions, each of which is respected by $G$:

    \begin{enumerate}
        \item First, ``move"  the demands in $C$ to $A$ according to $f$.
        \item Then, spread the demands in $A$ uniformly (using the fact that $A$ expands in $G$).
        \item Finally, ``move'' part of the (uniform) demands in $A$ back to $C$ according $f$ (in reverse). 
    \end{enumerate}

    We now formalize this argument.
    
    \begin{proof}[Proof of Lemma~\ref{lemma:separator-general}]
        Consider the all-to-all demand matrix  $Q_{A\cup C}$ on $A\cup C$ in~$G$ (recall Definition~\ref{def:all_to_all_flow}). By Remark~\ref{remark:mu_expansion_as_problem}, it suffices to show that $G$ respects $Q_{A\cup C}$. Let $P_1\defeq P_{A\cup C}$ be the corresponding valid demand state (see Remark~\ref{remark:flow_problem_as_flow_state}). Note that for each $u\in A\cup C$, $\norm{{P_1}(u)}_1 = 2\cdot\left(1-\frac{1}{|A\cup C|}\right) \le 2$.
        
        We use the given flow from $C$ to $A$, which corresponds to a demand matrix denoted by $Q$, to transport the demand state vectors of the vertices in $C\setminus A$ to vertices in $A$:
        We scale the demands
        of $Q$ 
        such that each $u\in C\setminus A$ sends $2\cdot\left(1-\frac{1}{|A\cup C|}\right)$ units and each $u\in C\cap A$ sends $0$. That is, we multiply the rows of $Q$ corresponding to $C\setminus A$ by $2\cdot\left(1-\frac{1}{|A\cup C|}\right)$ and rows of $C\cap A$ are zeroed out. 
        Let $Q'$ be the scaled demand matrix. Because each of the rows was scaled by at most $2$, $Q'$ can be routed with congestion $2c$ and every vertex $v\in A$ receives at most $2a$ units. Next, consider the updated demand state $P_2\defeq P_1^{\uparrow Q'}$ (recall Definition \ref{def:augment_state}). By Lemma \ref{lemma:special_flows}(2), $P_2$ is supported only on $A$. Moreover, for every $u\in A$, the load satisfies $\norm{{P_2}(u)}_1\le 2+2a$. 
        Indeed, by triangle inequality,
        \begin{align*}
            \norm{P_2(u)}_1 &\le \norm{P_1(u) - \frac{P_1(u)}{\norm{P_1(u)}_1}\cdot\sum_{v\in V}{Q'(u,v)}}_1 + \norm{\sum_{v\in V}{\frac{P_1(v)}{\norm{P_1(v)}_1}\cdot Q'(v,u)}}_1 \\ &\le 
            \norm{P_1(u)}_1 + \sum_{v\in V}Q'(v,u) \le  \norm{P_1(u)}_1 +2a \le 2+2a.
        \end{align*}

        Because $Q'$ can be routed with congestion $2c$, it follows (see Fact \ref{fact:congestion_implies_respecting_cut}) that $G$, $\frac{1}{2c}$-respects $Q'$. Thus, by Fact \ref{fact:basic_flow_state}(2), for any cut $(U, W= V\setminus U)$ we have $\dem_{(P_1 - P_2)}(U,W) \le 2c\cdot \capacity_G(U,W)$. 
        
        Next, since $A$ $\phi$-expands in $G$, Remark~\ref{remark:mu_expansion_as_problem} implies that $G$ $\frac{\phi}{2}$-respects the all-to-all flow problem on $A$, whose demand matrix is denoted by $Q_A$.
        By scaling the rows of $Q_A$ by at most $2+2a$, we can get a demand matrix which routes $\norm{{P_2}(u)}_1$ units out of each $u\in A$, and spreads them uniformly on $A$. This scaled problem is $\frac{\phi}{4+4a}$-respected by $G$. 
        Updating $P_2$ using this scaled flow gives $P_3$, which by Lemma~\ref{lemma:special_flows}, has identical demand vectors for all $u\in A$. Because $P_1$ is a valid demand state, so are $P_2$ and $P_3$ (see Fact \ref{fact:basic_flow_state}(1)). Since all non-zero demand vectors in $P_3$ are equal, this means that $P_3\equiv 0$.

        Similar to before, we get that for any cut $(U,W)$ we have $\dem_{(P_2 - P_3)}(U,W) \le \frac{4+4a}{\phi}\cdot \capacity_G(U,W)$. Therefore, using Fact \ref{fact:extended_flow_problems}, for any cut $(U,W)$, we have
        \begin{align*}
            \dem_{Q_{A\cup C}}(U,W)&=\dem_{P_1}(U,W) \le \dem_{(P_1-P_2)}(U,W) + \dem_{(P_2-P_3)}(U,W) 
            \\
            &\le \left(2c+\frac{4+4a}{\phi}\right)\cdot\capacity_G(U,W) = \frac{2\cdot(2+2a+c\phi)}{\phi}\cdot\capacity_G(U,W) ,
        \end{align*}
        which gives the result using Remark \ref{remark:mu_expansion_as_problem}.

    \end{proof}

    \begin{figure}[t]
    \begin{center}
    \begin{tabular}{ccccc}
    \includegraphics[scale=0.45]{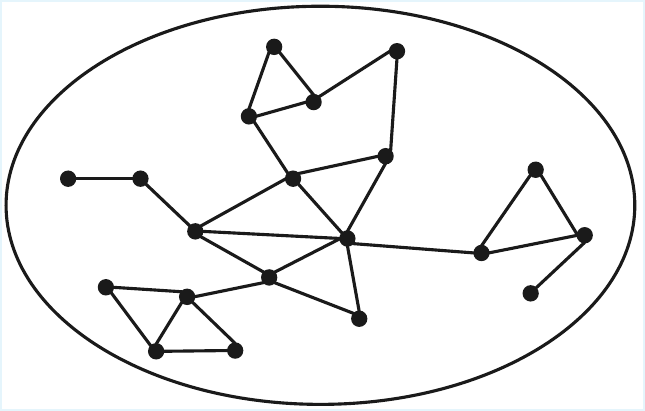} & \quad &
    \includegraphics[scale=0.45]{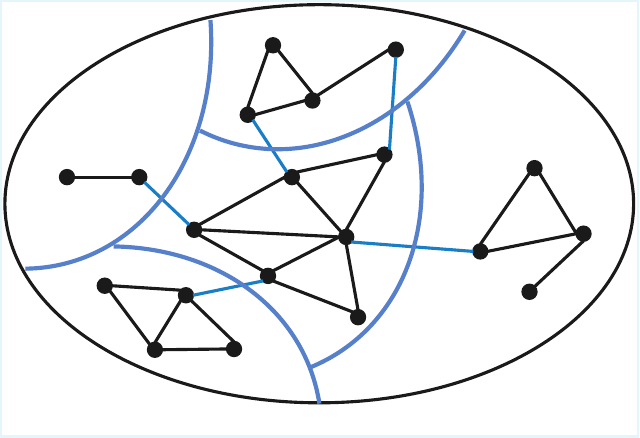} & \quad & 
    \includegraphics[scale=0.45]{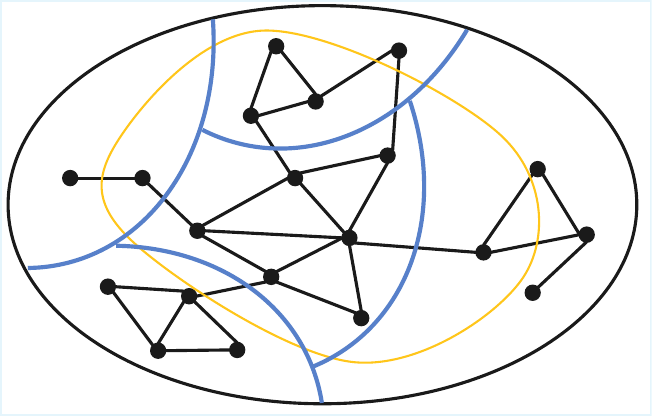}\\
    (a) & \quad & (b) & \quad & (c) \\
    \includegraphics[scale=0.45]{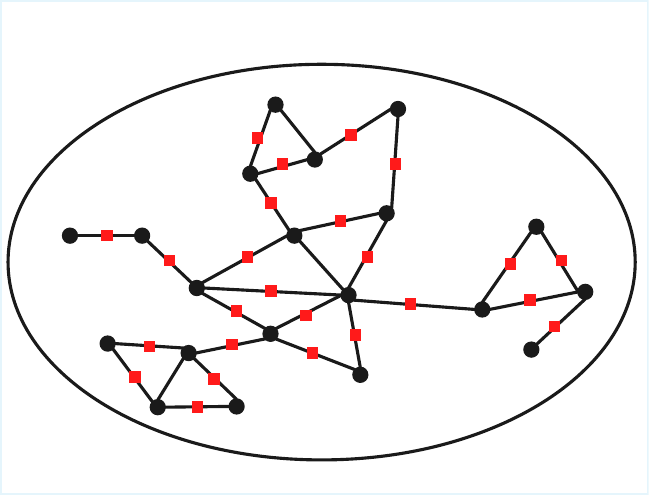} & \quad &
    \includegraphics[scale=0.45]{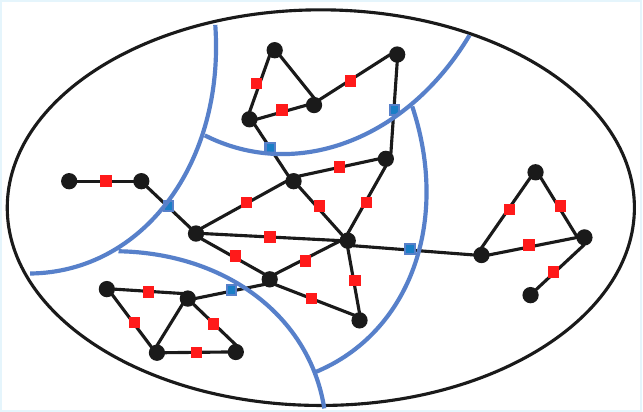} & \quad &
    \includegraphics[scale=0.45]{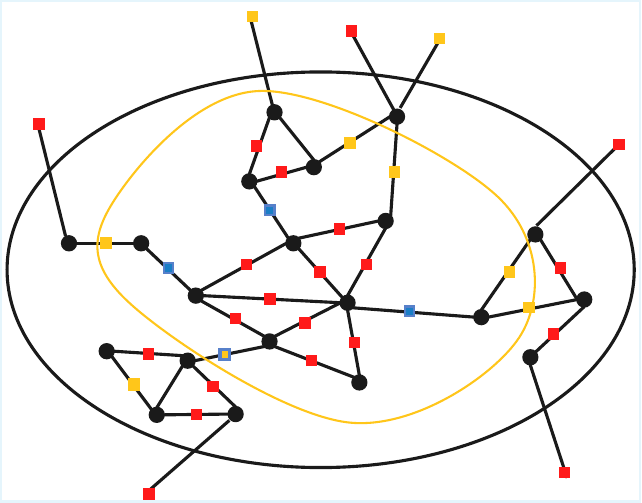}\\
    (d) & \quad & (e) & \quad & (f)
    \end{tabular}
    \end{center}
    \caption{Partitioning of cluster $S$. Figures (a)-(c) represent $G$, figures (d)-(f) represent the split graph. Black ovals depict vertices and squares depict split nodes. Figures $(b)$ and $(e)$ correspond to the first partition. Blue edges (and blue squares)
    depict inter-cluster edges of the first partition (that is, $F$ and $X_F$, respectively). 
    Figure $(f)$ depicts the second partition (which is computed in $G'[S']$) (recall $S'=S\cup \{x_e \mid e\cap S\ge 1\}$). The core $L$ is enveloped in gold. Split nodes in $X_Y$ are filled with gold. Split nodes in $X_{F\cap Y}$ are colored in blue and gold.}
    \label{fig:examples}
    \end{figure}

    To prove Lemma~\ref{lemma:separator} we need the following claim. It is a slightly more general rephrasing of Claim 1 of \cite{racke2014computing}. We provide its proof for completeness. 
    
    \begin{claim}[\texorpdfstring{\cite[Claim 1]{racke2014computing}}{[RST14, Claim 1]}]
    \label{claim:improve_balanced}
        Let $F\subseteq E$ be a subset that induces a balanced clustering. Let $A\subseteq F$ be arbitrary. Let $C\subseteq E$ be a set of edges that separates $A$ and $F\setminus A$. Then, either $A\cup C$ or $(F\setminus A)\cup C$ induces a balanced clustering.
    \end{claim}
    \begin{proof}
        Consider the graph $G \setminus C$ where we remove the edges in $C$. Denote by $B$ the set of vertices that can reach any edge of $A$ in this graph, and by $B'$ the set of vertices that can reach any edge of $F\setminus A$ in this graph. Because $C$ separates $A$ from $F\setminus A$, $B\cap B' = \emptyset$. Assume w.l.o.g that $|B|\le|B'|$. In particular, $|B|\le\frac{1}{2}|V|$. Note that each connected component of the graph $G\setminus ((F\setminus A)\cup C)$, is either a subset of a connected component of $G\setminus F$, or it is a subset of $B$. In particular all connected components have size at most $\frac{2}{3}|V|$.
    \end{proof}

    We are ready to prove Lemma~\ref{lemma:separator}.
    
    \begin{proof} [Proof of Lemma~\ref{lemma:separator}]
        Given the subset $F\subseteq E$ that induces a balanced clustering, we use Theorem~\ref{theorem:balanced_cut_mu} on $G'$ with $\phi \defeq \Theta\left(\frac{1}{\log n}\right)$ and $\mu \defeq \mu_{X_F}$. 
        It terminates with one of three cases:
        \begin{enumerate}
            \item With high probability, $\Phi^{X_F}(G') = \Omega(\frac{1}{\log n})$. In this case $X_F$ $\Omega(\frac{1}{\log n})$-expands in $G'$, and we have Case (2) of the lemma (with $\tilde{A} = F, \tilde{C} = \emptyset$). 
            
            \item We find a cut $(A', \bar{A'})$ in $G'$ of $\mu$-expansion $\Phi_{G'}^{X_F}(A', \bar{A'}) = O(\phi\log n)$, and $|X_F\cap A'|, |X_F\setminus A'|$ are both $\Omega(\frac{|F|}{\log n})$. We choose the hidden constant in the definition of $\phi = \Theta\left(\frac{1}{\log n}\right)$ such that 
            \begin{equation}
            \label{eq:separator1}
                \Phi_{G'}^{X_F}(A', \bar{A'}) \le \frac{1}{2}~,
            \end{equation}
            and both
            \begin{equation}
            \label{eq:separator2}
                |X_F\cap A'|, |X_F\setminus A'| \ge 2c_1\cdot\frac{|F|}{\log n}~,
            \end{equation}
            for some constant $c_1 > 0$. 
            Let $C'\subseteq E'$ be the set of edges that cross the cut $(A', \bar{A'})$ in $G'$. Let $C\subseteq E$ be the corresponding set in $G$: $C \defeq \{e\in E \mid \exists c\in C': x_e\in c \}$. 
            We have $|C| \le |C'|$. Inequality (\ref{eq:separator1}) is equivalent 
            to $|C'| \le \frac{1}{2}\min(|X_F\cap A'|, |X_F\setminus A'|)$. Let $A\subseteq E$ be the set of edges corresponding to $A'$: $A \defeq \{e\in E\mid x_e\in A'\}$. We have $|F\cap A|=|X_F\cap A'|$ and $|F\setminus A| = |X_F\setminus A'|$.
            Therefore,
            \begin{equation}
            \label{eq:separator3}
                |C| \le \frac{1}{2}\min(|F\cap A|, |F\setminus A|)~.
            \end{equation}
            
            Since the set of edges $C'$ separated $A'$ from $\bar{A'}$ in $G'$, the edges in $C$ separate the edges in $A$ from the edges in $E\setminus A$ in the graph $G$.
            Claim~\ref{claim:improve_balanced} gives that either $(F\cap A)\cup C$ or $(F\setminus A)\cup C$ induces a balanced clustering. W.l.o.g it is $(F\cap A)\cup C$.
            We have
            \begin{align*}
                |(F\cap A)\cup C| &\le |F| - |F\setminus A| + |C| \le |F| - |F\setminus A| + \frac{1}{2}\min(|F\cap A|, |F\setminus A|) 
                \\
                &\le |F| - \frac{1}{2}|F\setminus A| \le |F|\cdot\left(1 - \frac{c_1}{\log n}\right)
            \end{align*}
            where the second inequality follows from (\ref{eq:separator3}), and the last one follows from (\ref{eq:separator2}).
            So we get Case (1) of the lemma, for $\tilde{F} = (F\cap A)\cup C$. 
            
            \item We find a cut $(A',\bar{A'})$ in $G'$ 
            with $0<|X_F\setminus A'|\le \frac{|F|}{2}$, $\Phi_{G'}^{X_F}(A', \bar{A'}) = O(\phi\log n)$, and with high probability, 
            $A'$  $\Omega(\phi)$-expands with respect to $X_F$ in $G$. If $|X_F\setminus A'| > \frac{|F|}{\log n}$, then we complete the proof as in Case (2). Therefore, assume that $0<|X_F\setminus A'|\le \frac{|F|}{\log n}$.
            As $\phi = \Theta\left(\frac{1}{\log n}\right)$, denote by $c_2>0$ the constant such that $\Phi_{G'}^{X_F}(A', \bar{A'}) \le c_2$. 
            Like in the previous case, denote by $C'\subseteq E'$ the set of edges that cross the cut $(A', \bar{A'})$ in $G'$. Let $C\subseteq E$ be the corresponding set in $G$: $C \defeq \{e\in E \mid \exists c\in C': x_e\in c \}$. We have $|C| \le |C'|$. Because the cut is sparse with respect to $\mu_{X_F}$, $|C'| \le c_2\cdot|X_F\setminus A'|$.
            
            From Theorem~\ref{theorem:trimming_allows_routing}, 
        we get that there exists a flow in $G'$, routable with congestion $O(1)$, such that each split-node in $X_C$ sends $1$ unit of flow and each node $v\in A'$ receives at most $O(\phi\cdot\mu(v))$ units. Recall that $\phi = \Theta\left(\frac{1}{\log n}\right)$, so for large enough $n$, we may assume that each node receives at most $\mu(v)$ units of flow.
            In particular, only split nodes $v\in X_F$ can receive flow, and each receives at most one unit.
            
            As in the previous case, let $A \defeq \{e\in E\mid x_e\in A'\}$. We have $|F\cap A|=|X_F\cap A'|$ and $|F\setminus A| = |X_F\setminus A'|$.
            Denote $\tilde{A} = F\cap A$, $\tilde{R} = F\setminus A$. Then, $\tilde{A} \cup \tilde{R} = F$ induces a balanced clustering. 
            
            As in the previous case, observe that the set $C\subseteq E$ separates $\tilde{A}$ from $\tilde{R}$. By Claim~\ref{claim:improve_balanced}, either $\tilde{A}\cup C$ or $\tilde{R}\cup C$ induces a balanced clustering.
            \begin{enumerate}
                \item Assume $\tilde{R}\cup C$ induces a balanced clustering. In this case, we will set $\tilde{F} \defeq \tilde{R}\cup C$, and claim that we get Case (1) of the lemma. Indeed, we have $|\tilde{R}\cup C| \le |\tilde{R}| + |C| = |F\setminus A| + |C| \numle{1} (c_2 + 1)\cdot|X_F\setminus A'| \le (c_2 + 1)\cdot \frac{|F|}{\log n} < |F|\cdot (1 - \frac{1}{\log n})$, where Inequality~$(1)$ holds since $|C|\le |C'| \le c_2\cdot|X_F \setminus A'|$, and the last inequality holds for large enough $n$. 
                \item Assume $\tilde{A}\cup C$ induces a balanced clustering. We  set $\tilde{C} \defeq C$, $\tilde{F} = \tilde{A} \cup \tilde{C}$, and claim that we have Case (2) of the lemma. As mentioned above, the existence of the flow is guaranteed by Theorem~\ref{theorem:trimming_allows_routing}, and each node $v\in A'$ receives at most $\mu(v)$ units of flow.  
                Since $\mu = \mu_{X_F}$, each node $x_a\in X_{\tilde{A}} = X_F\cap A'$ 
                receives at most one unit of flow, and the rest receive zero. 
                
                The fact that $A'$ $\Omega(\phi)$-expands with respect to $\mu_{X_F}$ in $G'$ means that $G'$ is an $\Omega(\phi)$ expander with respect to $X_F \cap A'  = X_{\tilde{A}}$, which by definition means that $X_{\tilde{A}}$ $\Omega(\phi)$-expands in $G'$. As $\phi = \Theta(\frac{1}{\log n})$, this means that $X_{\tilde{A}}$ $\Omega(\frac{1}{\log n})$-expands in $G'$.
                This shows we have Case (2) of the Lemma. 

            \end{enumerate}
        \end{enumerate}
    \end{proof}
    
    \subsection{Finding a Bottleneck Cut}
    \label{section:merge-phase-part2}

    We denote by $Z_1, \ldots, Z_z$ the partition of $S$ that is returned by Theorem \ref{theorem:merge-phase_part1}. Recall from that theorem that $F$ denotes the set of inter-cluster edges (\ie, $F = \{(u, v)\in E \mid u\in Z_i, v\in Z_j, i\neq j\}$). 
    The partition of $S$ into $L\cup R$ is computed using both this initial clustering and the set of boundary edges $B = E(S, V\setminus S)$ of $S$. 

    As mentioned in Section \ref{section:techniques}, we want the set $F\cup B$ of inter-cluster edges and boundary edges to expand in $G[S]$. 
    However, the existence of bottleneck-cuts may make such a partition impossible (in the cut-sparsifier case, a bottleneck cut is a cut that certifies that the boundary edges are not expanding within the cluster). 
    Similar to the technique \cite{racke2014computing} used in the flow-sparsifier case, the approach we use in this case is the following: As Claim~\ref{claim:equiv-def-cut} shows, in order to certify that a tree is a cut-sparsifier, we should check that all demand matrices which can be routed in the tree are respected 
    by $G$. We therefore find a ``problematic'' cut $(L, R=S\setminus L)$ and add $L, R$ to the tree. This forces every demand matrix which is respected by $T$ to respect the cut $(L, V\setminus L)$.
    Specifically, we will add a cut which ``blocks'' the sending of flow from $X_F$ to $X_B$. As shown in Section~\ref{section:charging-scheme}, this will be enough.

    Formally, the second partition is parameterized by $0 < \uptau \le 1$. The parameter $\uptau$ balances two opposing effects in our analysis: a larger $\uptau$ reduces the cost incurred when routing flow \emph{within} a cluster, but simultaneously increases the $\ell_{1}$-blowup (that is, the increase in load) of the demand state vectors as we move from a cluster to its parent. Thus, $\uptau$ governs the trade-off between the internal routing cost and the load accumulation along the decomposition tree. See Theorem \ref{theorem:charge-in-S} for the exact details of this trade-off.
    
    We shift our focus from the graph $G[S]'$ to the graph $G'[S']$, where $S'\defeq S\cup \{x_e \mid e\cap S\ge 1\}$. That is, in addition to the vertices of $G[S]'$, we also include the nodes $x_e$ for every boundary edge $e\in B$ and connect $x_e$ to $v$ where $v\in e\cap S$.
    
    Next, we define the weighting $\mu_\uptau$ on the vertices of $G'[S']$: the boundary split nodes are assigned a weight of $\uptau$, and the rest of the vertices 
    have weight 1.
\footnote{For non split nodes, $v\in S$, $\mu_\uptau(v)$ does not play a role in the following discussion.}
    As explained above, we are looking for a cut which blocks the sending of flow from $X_F$ to $X_B$. We  therefore apply the following lemma on the sets $X_F$ and $X_B$. The set of edges $Y$ will be the cut edges that ``block'' the flow.

    \begin{lemma} [Improved version of \texorpdfstring{\cite[Lemma 3.9]{racke2014computing}}{[RST14, Lemma 3.9]}]
    \label{lemma:fair-cut-separator}
        Let $G = (V, E)$ be a graph, and let $\mu$ be a measure on $V'$. Consider two sets of split vertices $X_A$ and $X_B$ in $G'$. We can find a set $X_Y$ of split vertices such that the following holds:
        \begin{itemize}
            \item The set $X_Y$ separates $X_A$ from $X_B$. This means that every path between nodes $x_a\in X_A$ and $x_b\in X_B$ in $G'$ must contain a vertex $x_y\in X_Y$.
            \item There exists a multi-commodity flow in $G'$ from nodes $X_Y$ to nodes in $X_A$ that can be routed with congestion $2$, such that every node $x_y\in X_Y$ sends out $\mu(x_y)$ units of flow, and every node $x_a\in X_A$ receives at most $2\mu(x_a)$ units of flow.
            \item There exists a multi-commodity flow in $G'$ from nodes $X_Y$ to nodes in $X_B$ that can be routed with congestion $2$, such that every node $x_y\in X_Y$ sends out $\mu(x_y)$ units of flow, and every node $x_b\in X_B$ receives at most $2\mu(x_b)$ units of flow.
        \end{itemize}
        The set $X_Y$ can be computed in time $\tilde{O}(m)$.
    \end{lemma}
    \begin{remark}
    \label{remark:fair-cut-separator}
        In the lemma above, the separator $Y$ must satisfy $\mu(X_Y)\le 2\min(\mu(X_A), \mu(X_B))$ for such flows to exist, as the total source mass is $\mu(X_Y)$ and the total sink capacities are $2\mu(X_A), 2\mu(X_B)$.

    \end{remark}

    Consider the flow network where we add to $G'$ a source $s$ connected to $X_A$ (with edges $(s, x_a)$ of capacity $\mu(x_a)$) and a sink $t$ connected to $X_B$ (with edges $(x_b, t)$ of capacity $\mu(x_b)$). It can be seen that an exact min–cut $Y$ that separates $s$ from $t$ in this network would satisfy the requirements of Lemma~\ref{lemma:fair-cut-separator} with the constant~$1$ instead of~$2$. 
    However, computing a min–cut is too expensive for our purposes. Instead, we use fair–cuts \cite{LNPSsoda23}, which provide the same guarantees up to constant factors and can be computed in near-linear time.

    \begin{definition}[Fair Cut, \cite{LNPSsoda23}]
        Let $G=(V, E)$ be an undirected graph with edge capacities $c\in \RR^E_{>0}$. Let $s, t$ be two vertices in $V$. For any parameter $\alpha \ge 1$, we say that a cut $(S, T)$ is a \emph{$\alpha$-fair $(s, t)$-cut} if there exists a feasible $(s, t)$-flow $f$ such that $f(u, v)\ge \frac{1}{\alpha}\cdot c(u, v)$ for every $(u, v)\in E(S, T)$, where $u\in S$ and $v\in T$. In particular, $f$ is an $\alpha$-approximate max flow.
    \end{definition}

    Standard approximate max-flow algorithms have the unpleasant property that the flow provided can go backwards from $T$ to $S$ across the approximate min-cut $(S,T)$. Fair-cuts give a much stronger guarantee: a fair-cut requires that \emph{every} edge that crosses the cut is roughly saturated. Thus, the flow can only cross the cut from $S$ to $T$ and the capacity of every cut edge is almost fully utilized. 
    Intuitively, for our use cases, fair-cuts can be thought of as analogues to minimum cuts which can be computed in near linear time (they give the same guarantees as minimum cuts, up to constant factors).

    The following theorem, due to~\cite{li2025simple}, is the state of the art result for the computation of fair-cuts.

    \begin{theorem}[Fair Cut, \texorpdfstring{\cite[Theorem 1]{li2025simple}}{[LL25, Theorem 1]}]
    \label{theorem:fair_cuts}
        Given a graph $G=(V, E)$, two vertices $s, t\in V$ and $\epsilon\in (0, 1]$, we can compute with high probability a $(1+\epsilon)$-fair $(s, t)$-cut in $\tilde{O}(\frac{m}{\epsilon})$ time.
    \end{theorem}

    \begin{proof}[Proof of Lemma~\ref{lemma:fair-cut-separator}]
        Consider the following flow problem: start with $G'$ and set the edge capacities according to $\mu$. That is, set $\capacity(u,x_e)= \capacity(x_e, v) = \mu(x_e)$, for every $e=(u,v)\in E$. Next, 
        connect a source $s$ to all split nodes in $x_a \in X_A$ with capacity $\mu(x_a)$ and connect all split nodes $x_b \in X_B$ to a target $t$ with capacity $\mu(x_b)$. 
        Using Theorem~\ref{theorem:fair_cuts}, we get a $\frac{3}{2}$-fair $(s,t)$-cut $(S,\bar{S})$ in $G'$ corresponding to a $\frac{3}{2}$-approximate max flow $f$.
        We define $X_Y$ as follows: for every cut edge $(v, x_e) \in E'(S,\bar{S})$ we add $x_e$ to $X_Y$. This $X_Y$ separates $X_A$ from $X_B$. From the definition of a $\frac{3}{2}$-fair cut and since $X_Y$ consists of split nodes, it follows that each $x_e \in X_Y$ receives (and also sends)  
        at least $\frac{2}{3}\muG{}{x_e}$ flow. 

        The desired flows are obtained by the following modification of $f$. Consider the scaled flow $f' := \frac{3}{2}f$. For each $x_e \in X_Y$, the total flow (in $f'$) $x_e$ sends (or receives, depends on its side in the cut) through the cut $(S,\bar{S})$ is at least $\muG{}{x_e}$. 
        We decompose $f'$ into flow paths and scale them down such that, for $x_e \in X_Y$, the total flow $x_e \in X_Y$ sends (or receives) through $(S,\bar{S})$ is exactly $\muG{}{x_e}$. This scaling argument is possible since $(S,\bar{S})$ is a fair-cut and therefore each flow path of $f$ crosses the cut $(S,\bar{S})$ exactly once. Let $f''$ be the resulting flow.

        To see why $X_Y$ satisfies the constraints, decompose $f''$ into flow paths. 
        Thus, we can view $f''$ as two routings: one from $s$ (\ie, $X_A$) to $X_Y$, and one from $X_Y$ to $t$ (\ie, $X_B$). Reversing the former flow gives the result.
    \end{proof}
    
    We define the second partition ($L$, $R$) by computing the subset $X_Y$ of split nodes in $G'[S']$, as obtained from applying Lemma~\ref{lemma:fair-cut-separator} on $X_F$ and $X_B$ (see Figure~\ref{fig:examples}).
    The sets $L$ and $R$ are then defined as follows: $R$ consists of all vertices in $S$ that can reach a boundary node $x_b\in X_B\setminus X_Y$ in the graph $G'[S']\setminus X_Y$, while $L$ contains the remaining vertices. We adopt the notation of \cite{racke2014computing} and call $L$ the \emph{core} of $S$. Note that $X_F\setminus X_Y\subseteq L'$, where $L' \defeq L \cup \{x_e \mid e=(u,v), u,v\in L\}$. 

    By applying Lemma \ref{lemma:fair-cut-separator} with $\mu_\uptau$, we get the following theorem:
    \begin{theorem}
    \label{theorem:merge-phase-part2}
        Let $0 < \uptau \le 1$. Define the weighting $\mu_\uptau$ on the split vertices of $G'[S']$, such that boundary split nodes have weight $\uptau$ while the rest of the split nodes have weight $1$. There exists an algorithm \emph{merge-phase-2} which finds a partition $L\cup R = S$, induced by a set of edges $Y\subseteq E$ as described above, such that,
        \begin{enumerate}
            \item\label{prop-1} There is a flow in $G'[S']$ 
            from nodes in $X_Y$ to nodes in $X_B$, routable with congestion $2$, such that \emph{every} node $x_y \in X_Y$ sends out $\mu_\uptau(x_y)$ units of flow, 
            and each node $x_b \in X_B$ gets at most $2 \uptau$ units of flow.
        
            \item\label{prop-2} There is a flow in $G'[S']$ from nodes in $X_Y$ to nodes in $X_F$, routable with congestion $2$, such that \emph{every} node $x_y \in X_Y$ sends out $\mu_\uptau(x_y)$ units of flow, and each node $x_f \in X_F$ gets at most $2$ units of flow.
        
            \item The nodes in $X_F$ are separated from $X_B$ in $G'[S'] \setminus X_Y$. This means that every path in $G'[S']$ that starts at $x_b\in X_B$ and ends at $x_f\in X_F$ contains a node from $X_Y$ (which may be $x_b$ or $x_f$).
        
            \item The nodes in $L$ are separated from $X_B$ in $G'[S']\setminus X_Y$.
        \end{enumerate}

        The algorithm runs in time $\tilde{O}(m)$ where $m$ denotes the number of edges of $G'[S']$.
    \end{theorem}

    \begin{remark}
        In degenerate cases, it may occur that $F = \emptyset$ (\eg, when $G[S]$ has no edges) or $B=\emptyset$ (when $S = V$ or if $G$ is disconnected).
        In either case, Theorem~\ref{theorem:merge-phase-part2} returns $Y = \emptyset$.
    \end{remark}

    At this point, we can give some informal intuition for the role of the parameter $\uptau$ (used in Theorem~\ref{theorem:merge-phase-part2}) and the trade-off it induces. As we will see in the next section, to prove that our hierarchical decomposition $T$ is a tree cut-sparsifier, it suffices to show that at each cluster $S$ we can perform the following local operation: given demands supported on the boundaries of the sub-clusters of $S$ (that is, on $X_F\cup X_Y\cup X_B$), we are able to (i) sufficiently mix the demands inside $S$, so that some positive and negative demands on $X_F\cup X_Y$ cancel out (see Lemma~\ref{theorem:charge-in-S_part1} and Corollary~\ref{cor:obs_F}), 
    and (ii) route the remaining demands to the boundary split nodes $X_B$ of $S$, while ensuring that the load on each boundary split node does not increase too much, so that these demands can be handled at higher levels of the tree. 
    Theorem~\ref{theorem:merge-phase-part2} provides the mechanism need to carry these operations.

    For operation (i), in order to mix the demands on $X_F\cup X_Y$, we first move the demands residing on $X_Y$ to $X_F$, and then exploit the fact that $F$ $\Omega(\frac{1}{\log n})$-expands in $S$ in order to mix these demands inside $X_F$.
    We therefore apply Theorem~\ref{theorem:merge-phase-part2} and move the demands from $X_Y$ to $X_F$. 
    Since $X_Y$ may include boundary split nodes $x_y\in X_Y\cap X_B$ (which have constant load on them but a capacity of $\mu_\uptau(x_y) = \uptau$), the congestion incurred by routing the demands from $X_Y$ to $X_F$ is inversely proportional to $\uptau$.

    For operation (ii), we use the fact that the cut $(L,R)$, constructed using the separator $Y$, was added to the hierarchical decomposition. 
    This allows us to argue that after the mixing performed in operation~(i), the total demand remaining on $X_F\cup X_Y$ is bounded by the capacity of $Y$\footnote{Actually, the demand is bounded by the capacity of $\tilde{Y}\subseteq Y$ which is defined in the proof.} (see Corollary~\ref{cor:obs_F}). 
    Consequently, by again invoking the routing guarantees of Theorem~\ref{theorem:merge-phase-part2}, we can route the remaining demands to $X_Y$ such that the total demand on each vertex is constant.
    Finally, by applying Theorem~\ref{theorem:merge-phase-part2} once more, we move the demands from $X_Y$ to the boundary $X_B$, while increasing the load on the boundary split nodes by at most a factor of $1+\uptau$.
    
    \section{The Charging Scheme}
    \label{section:charging-scheme}
    
    Let $T$ be the tree corresponding to the hierarchical decomposition from Section~\ref{section:merge-phase}. In this section we show that $T$ is indeed a tree cut-sparsifier with quality $O(\log^3 n)$, proving Theorem~\ref{theorem:final}. We define the capacity in $T$ of an edge connecting a cluster $S$ to its parent cluster $P_T(S)$, as $\capacity_G(S, V\setminus S)$. In section \ref{section:techniques} we mentioned that the method used to show that a tree is a cut-sparsifier is to prove that every demand matrix that can be routed in the tree is respected by the graph. We begin by explaining this formally.
    
    Consider the following claim for tree flow-sparsifiers. 

    \begin{claim}
    \label{claim:equiv-def-flow}
        The following notions are equivalent for a (hierarchical) tree $T$ and a parameter $\alpha$:
        \begin{itemize}
            \item $T$ is a tree flow-sparsifier of $G$ with quality $\alpha$.
            \item For every demand matrix (recall Definition \ref{def:flow_problem}) $Q$ on $V$, if $Q$ can be routed in $T$ with congestion $1$, then it can be routed in $G$ with congestion $\alpha$.
        \end{itemize}
    \end{claim}
    \begin{proof}
        The implication from the first item to the second is immediate from the definition of a flow-sparsifier. We therefore prove the converse.
        
        Let $Q$ be a demand matrix on $V$. We have to show that $\congestion_{T}(Q)\le \congestion_{G}(Q) \le \alpha \cdot \congestion_{T}(Q)$. 
        The first inequality is clear: We route the flow in $T$ naturally (a unit from $u$ to $v$ moves along the unique path connecting them). 
        Therefore, there exists an edge  $(S, P_T(S))$ of $T$ (of capacity $\capacity_G(S, V\setminus S)$) such that $\dem_Q(S, V\setminus S) = \congestion_{T}(Q)\cdot\capacity_G(S, V\setminus S)$.
        Moreover, since $Q$ can be routed in $G$ with congestion $\congestion_{G}(Q)$, it follows from Fact~\ref{fact:congestion_implies_respecting_cut} that $\dem_Q(S, V\setminus S) \le \congestion_{G}(Q)\cdot\capacity_G(S, V\setminus S)$. Therefore, we conclude that $\congestion_{T}(Q) \le \congestion_{G}(Q)$. 
        
        For the other inequality, let $R$ denote the flow problem corresponding to $\frac{1}{\congestion_T(Q)}\cdot Q$. This flow can be routed in $T$ with congestion $1$. Hence,
        \[\alpha \ge \congestion_G\left(\frac{1}{\congestion_T(Q)}\cdot Q\right) = \frac{\congestion_G(Q)}{\congestion_T(Q)},\] 
        as required.
    \end{proof}
    
    The following is the corresponding claim for tree cut-sparsifiers. It is our main tool to show that $T$ is a tree cut-sparsifier. Recall the definition of ``respecting a cut" (Definition \ref{def:cut_respects_flow}).

    \begin{claim}
    \label{claim:equiv-def-cut}
        The following notions are equivalent for a (hierarchical) tree $T$ and a parameter $\alpha$:
        \begin{itemize}
            \item $T$ is a tree cut-sparsifier of $G$ with quality $\alpha$.
            \item For every demand matrix $Q$ on $V$, if $T$ $1$-respects $Q$, then $G$ $\frac{1}{\alpha}$-respects $Q$.
        \end{itemize}
    \end{claim}
    \begin{proof}
        The implication from the first item to the second is immediate from the definition of a cut-sparsifier.
        We therefore prove the converse.
        
        Fix a cut $(U,W)$ of $V$. 
        We have to show that $\capacity_G(U,W)\le\mincut_T(U,W)\le\alpha\cdot\capacity_G(U,W)$. 
        The first inequality is clear given the definition of the capacities in $T$. Indeed, consider any cut $(U',W')$ in $T$, where $U\subseteq U', W\subseteq W'$. For each $u\in U, w\in W$, the path from the leaf $\{u\}$ to $\{w\}$ in $T$ must contain an edge $(S,P_T(S))$  that crosses the cut $(U',W')$ of $T$. Assume without loss of generality that $S\in U'$ and $P_T(S)\in W'$. In particular, $u\in S$.
        Then, $\capacity_G(u,w)$ will be counted in the capacity $\capacity_T(S,P_T(S))=\capacity_G(S,V\setminus S)$.
        
        Next, we will show the other inequality. Consider the flow network $T' = T\cup\{s,t\}$, where $s$ is connected to all nodes in $W$ and $t$ is connected to all nodes in $U$ with edges of infinite capacity.
        Due to the max-flow min-cut theorem,
        the maximum $(s,t)$-flow $f$ in this network has value $\mincut_T(U,W)$. Separate this flow into paths. Each path begins in a vertex of $U$ and ends in a vertex of $W$.
        Denote by $R$ the demand matrix corresponding to the flow paths (each demand has a distinct commodity, whose source is a node $u\in U$ and whose target is a node $w\in W$). Clearly, the congestion of routing $R$ in $T$ is $1$, which means (see Fact \ref{fact:congestion_implies_respecting_cut}) that $R$ is $1$-respected by $T$. Therefore, we have $\frac{\capacity_{G}(U,W)}{\dem_{R}(U,W)} \ge \frac{1}{\alpha}$. Note that $\dem_R(U,W)$, the total demand that crosses the cut, is exactly the value of the flow $|f| = \mincut_T(U,W)$. Therefore, we have $\mincut_T(U,W)\le\alpha\cdot\capacity_G(U,W)$, as we wanted.
    \end{proof}

    We prefer to work with demand states instead of demand matrices.  To this end, we use the following Corollary:
    \begin{corollary}
    \label{cor:equiv-def-cut2}
        Assume every valid demand state $P$ (recall Definition \ref{def:valid_flow_state}) on $V$ (with arbitrary commodities $\K$),
        that is $1$-respected by $T$, is $\frac{1}{\alpha}$-respected by $G$. In other words, we have for any cut $(U,W)$, 
        \[
            \frac{\capacity_{G}(U,W)}{\dem_{P}(U,W)} \ge \frac{1}{\alpha} ~.
        \]
        Then $T$ is a tree cut-sparsifier with quality $\alpha$.
    \end{corollary}
    \begin{proof}
        Follows from Claim \ref{claim:equiv-def-cut} and Remark \ref{remark:flow_problem_as_flow_state}.
    \end{proof}

    Therefore, we have to show that every valid demand state that is respected by $T$ is $\Omega\!\left(\frac{1}{\log^3 n}\right)$-respected by $G$. 
    To this end, we use $T$ as a \emph{blueprint} that guides how demands are moved (in $G$) and charged to the capacity of cuts.
    Specifically, we update the demand state on $G$ iteratively as we go up along the tree structure, such that after processing a sub-cluster $S$, the demand state is no longer supported on any vertex of $S$, apart from it's boundary split nodes.\footnote{As will be explained soon, we work in the subdivision graph.}
    As mentioned in Section \ref{section:techniques}, we will ``pay'' not for the congestion of routing the demands, but rather according to how much the graph respects the corresponding demand transitions.

    The construction below follows~\cite{racke2014computing}, with the required modifications for the case of cuts instead of flows. 
    In contrast to~\cite{racke2014computing}, our presentation makes use of demand states, which provides a convenient and precise framework for describing the dynamic movements of the demands as we go up along the decomposition tree and for analyzing how the graph respects them.

    Fix a cut $(U \subseteq V, W = V\setminus U)$ in $G$, and fix a valid demand state $P$ on the vertex set $V$ (with arbitrary commodity set $\K$), 
    which is $1$-respected by $T$. 
    Our goal is to show that $\dem_P(U,W)\le O(\log^3 n)\cdot\capacity_G(U,W)$. We do this by charging the demand crossing the cut to the cut edges, such that the total charges dominate the demand, and each edge is charged $O(\log^3 n)$. 

    We work in the subdivision graph $G' = (V'=V\cup X_E, E')$. Recall that for a cluster $S\subseteq V$, we define the corresponding cluster in the subdivision graph as $S' \defeq S\cup\{x_e \mid e\cap S \ge 1\}$, which also includes the split nodes of the boundary edges of $S$. Recall that the graph $G'[S']$ is the subdivision graph restricted to the set $S'$.
    
    Let $\left(U' = U\cup\{x_e \mid e\cap U\ge 1\}, W' = V' \setminus U'\right)$ be the corresponding cut in $G'$ (see Figure~\ref{fig:U-W-cut}).\footnote{Note that we arbitrarily chose to include the split vertices of the edges crossing $(U,W)$ in $U'$.}
    Intuitively, we want to convert the demand state $P$ into a demand state $P'$ on $G'$, and then charge the edges crossing $(U',W')$ to cover the demand $\dem_{P'}(U',W')$, which would be bounded from below by $\dem_P(U,W)$.

    \begin{figure}[t]
        \begin{center}
            \begin{tabular}{c}
                \includegraphics[scale=0.35]{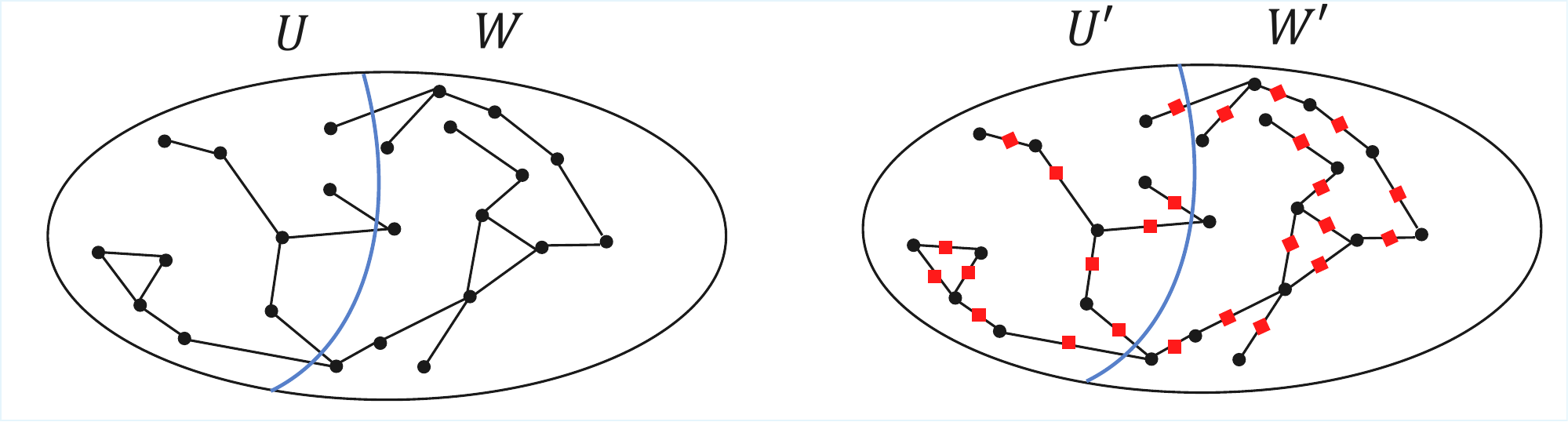} \\
                (a) \quad \quad \quad \quad \quad \quad \quad \quad \quad \quad \quad \quad \quad \quad \;  (b)
            \end{tabular}
        \end{center}
        \caption{Figure (a) - a cut $(U,W)$ in $G$. Figure (b) - the corresponding cut $(U',W')$ in $G'$.}
        \label{fig:U-W-cut}
    \end{figure}

    \cite{racke2014computing} essentially used Claim~\ref{claim:equiv-def-flow} to prove that $T$ is a tree flow-sparsifier, and to this end they showed how to route the corresponding valid demand state in $G$ with low congestion. Though they did not explicitly use the notation of demand states, their key method (rephrased in this notation) was akin to maintaining a collection of demand states, each supported on some sub-cluster $S'$ of the subdivision graph, corresponding to a vertex of $T$. The original state $P$ is converted to a set of demand states $\{P_{\{v\}'}\}_{v\in V}$, where for each vertex $v\in V$, $P_{\{v\}'}$ is supported on $\{v\}'$\footnote{The notation $\{v\}'$ refers to the cluster $S'$ induced by the single vertex set $S=\{v\}$. That is, this includes the vertex $v$ as well as it's adjacent split nodes.} (the vertices $V$ are the leaves in the decomposition tree $T$).
    Then they show how to merge the demand states according to the structure of the decomposition tree. 

    We proceed similarly. We combine the demand states when we go up the tree $T$, by constructing a series of flow problems (\ie, demand matrices) to update one demand state to the next (recall Definition \ref{def:augment_state}). Instead of showing that these matrices can be routed with low congestion, we charge the edges of the cut $(U',W')$ to cover the demands that these demand matrices induce across the cut. That is, in each such update step, we assign to every cut edge an amount of charge, so that the total assigned charge is at least the demand that the corresponding demand matrix induces across the cut.
    This is made explicit below.

    Initially, each edge of $(U',W')$ is charged $0$. 
    We convert the given valid demand state $P$ to multiple demand states $P_{\{v\}'}$, one for each vertex $v\in V$. For each $v\in V$,
    its demand vector $P(v)$ is  evenly distributed among the split nodes $x_e$, incident to $v$. Formally, each split node $x_e$ receives the demand vector $P_{\{v\}'}(x_e) = \frac{1}{\deg(v)}P(v)$.
    Observe that $\norm{P(v)}_1 = \dem_P{(\{v\}, V\setminus \{v\})} \le \capacity_{G}(\{v\}, V\setminus \{v\}) = \deg_G(v)$, where the inequality is due to the fact that $\{v\}$ is a cluster in $T$ and that $P$ is $1$-respected by $T$.
    In particular, each split node $x_e$ has a total load of at most $1$: $\norm{P_{\{v\}'}(x_e)}_1=\frac{1}{\deg(v)}\norm{P(v)}_1
    \le 1$. Each non-split vertex of $V$ has load of $0$.

    We combine demand states as we go up the decomposition tree $T$. We consider each cluster $S'$ in~$T$, and replace the demand states $P_{S_1'}, \ldots, P_{S_r'}$ that are supported on its sub-clusters $S_1', \ldots, S_r'$ with a new demand state that is supported on $S'$. This is done in Theorem~\ref{theorem:charge-in-S} below. Note that our demand states are not necessarily valid demand states, since $P_{S'_i}$ only stores the portion of the global demand that currently resides in $S'_i$. The matching positive and negative parts of a commodity may lie in different clusters, so complete cancellation cannot occur within a single cluster.  However, when summing over all demand states, these contributions do cancel out, and the total demand state remains valid. The demands states will, however, satisfy the following invariant with respect to $P$, the initial valid demand state we have fixed:

    \begin{invariant}
    \label{def:invariant}
        For a cluster $S$ and $\alpha \ge 1$, we say that a demand state $P_{S'}$  
        \emph{satisfies the invariant} with load $\alpha$ 
        if
        \begin{enumerate}[label=(\arabic*), ref=\arabic*]
            \item \label{inv-1} It is supported on the boundary of $S$:
            \[
                \forall k\in \K, \forall v\notin X_B: P_{S'}(v, k) = 0~,
            \]
            where $B$ are the boundary edges of $S$.  
            \item \label{inv-2} No demands were created or destroyed (relative to $P$): 
            \[
                \forall k\in \K: \sum_{v\in S'}{P_{S'}(v, k)} = \sum_{v\in S}{P(v, k)} ~.
            \]
            \item \label{inv-3} The load on each vertex is at most $\alpha$:
            \[
                \forall v\in S', \norm{P_{S'}(v)}_1 \le \alpha~.
            \]
        \end{enumerate}
    \end{invariant}
    It can be seen that $P_{\{v\}'}$ satisfies the invariant with $\alpha = 1$ for the cluster $\{v\}$. 
    
    Recall from Section \ref{section:merge-phase} that the partition of a cluster $S$ is comprised of two levels. In the first level, $S$ is cut into two sets $L$ and $R$, and in the second level $L$ and $R$ are partitioned into $L_1, \ldots, L_\ell$ and $R_1, \ldots, R_r$, respectively. Further recall the notation from Section \ref{section:merge-phase}: The set $F$ denotes the \textit{inter-cluster edges} from Theorem \ref{theorem:merge-phase_part1}, the set $Y$ denotes the \textit{cut edges} from Theorem \ref{theorem:merge-phase-part2} and the set $B$ denotes the \textit{boundary edges} $E(S, V\setminus S)$. These sets may intersect. In Section~\ref{section:charging-in-S}, we describe a procedure that is the core component of our argument, and ``combines'' demand states satisfying Invariant~\ref{def:invariant} for sub-clusters of a cluster $S$ to a demand state that satisfies it for $S$.
    
    The properties of this process are summarized in the following theorem:
    \begin{theorem}
    \label{theorem:charge-in-S}
        Let $S$ be a cluster, and let $S_1, \ldots, S_t$ be its sub-clusters (\ie, $\{S_1,\ldots S_t\} = \{L_1,\ldots,L_\ell\}\cup\{R_1,\ldots,R_r\}$), computed by Theorems~\ref{theorem:merge-phase_part1} and~\ref{theorem:merge-phase-part2} with parameter $0 < \uptau \le 1$. 
        For each $1\le i\le t$, assume that $P_{S'_i}$ is a demand state satisfying Invariant~\ref{def:invariant} for $S_i$ with load $\alpha \ge 1$. Let $P^{\mathrm{before}}_{S'}=\sum_{i=1}^t P_{S'_i}$ be their sum.
        Then, we can replace $P^{\mathrm{before}}_{S'}$ with a demand state $P^{\mathrm{after}}_{S'}$ which satisfies the following:
        
        \begin{enumerate}
            \item $P^{\mathrm{after}}_{S'}$ satisfies Invariant~\ref{def:invariant} for $S$ with load 
            $\alpha' = \alpha\cdot(1 + c\cdot\uptau)$ for some constant $c > 0$.
            \item $G'[S']$ $\Omega\!\left(\frac{\uptau}{\alpha \log n}\right)$-respects $P^{\mathrm{before}}_{S'} - P^{\mathrm{after}}_{S'}$. Moreover, $P^{\mathrm{before}}_{S'} - P^{\mathrm{after}}_{S'}$ is a valid demand state.
        \end{enumerate}
    \end{theorem}

    \begin{remark}
    \label{remark:respects_means_charge}
        Result (2) of Theorem \ref{theorem:charge-in-S} means that if we  charge the edges $e\in E(U',W')\cap G'[S']$ to cover the demand of $P^{\mathrm{before}}_{S'} - P^{\mathrm{after}}_{S'}$ across the cut $(U',W')$ by
        \[
            \charge_S(e) \defeq \frac{\dem_{\left(P^{\mathrm{before}}_{S'} - P^{\mathrm{after}}_{S'}\right)}(U',W')}{\capacity\!\left(E(U',W')\cap G'[S']\right)},
        \]
        then  $\charge_S(e)=O\!\left(\frac{\alpha\log n}{\uptau}\right)$.
        Note that the charged edges are in the subdivision graph.
    \end{remark}

    In other words, Theorem~\ref{theorem:charge-in-S} shows that the demand states supported on the sub-clusters $S_1, \ldots, S_t$ can be merged into a single demand state on $S$, while only slightly increasing the load (by a factor of $1+O(\uptau)$) on the boundary edges of $S$. The cost of this operation is that we add $O\left(\frac{\alpha\log n}{\uptau}\right)$ to the charge of each edge $e\in E(U',W')\cap G'[S']$.
    The proof of this theorem is given in Section~\ref{section:charging-in-S}.
    
    Next, we use Theorem~\ref{theorem:charge-in-S} to show Theorem \ref{theorem:final}. 

    \begin{proof} [Proof of Theorem \ref{theorem:final}]
        We prove that the conditions of Corollary~\ref{cor:equiv-def-cut2} are satisfied. Let $P$, $(U,W)$, and $(U',W')$ be as in the discussion following  Corollary~\ref{cor:equiv-def-cut2}.
        Set $\uptau = \frac{1}{\log n}$.
        We process $P$ in a bottom-up manner: First we convert it to a set of demand states $P_{\{v\}'}$ as described above. Then, for each cluster in $T$, we apply Theorem~\ref{theorem:charge-in-S} and Remark \ref{remark:respects_means_charge}, accumulating the charges incurred at each application of the theorem.
        Let $\charge_S(e)$
        denote the charge assigned to an edge $e \in E(U',W')\cap G'[S']$ when Theorem~\ref{theorem:charge-in-S} is applied to cluster $S$ as in Remark~\ref{remark:respects_means_charge}.
        
        After applying the theorem for the final cluster $V$, the resulting demand state $P_{V'}$ is actually the zero demand state on $V'$. This follows from Invariant~\ref{def:invariant}\eqref{inv-1}, since the cluster $V$ has no boundary edges. We now explain why the accumulated charges are enough to cover $\dem_P(U,W)$.

        Denote by $\Q$ the set of current ``active'' clusters, such that initially $\Q=\Q_0\defeq\{\{v\} \mid v\in V\}$.  Then in each iteration we remove $S_1, \ldots, S_t$ from $\Q$ and add $S$. Let $P_\Q$ be the sum of all demand states $P_{S'}$, for $S\in \Q$. 

        \begin{claim}
\label{claim:total_active_demand_is_a_flow_problem}
            In each step, $P_\Q$ is a valid demand state on $V'$. That is, $\sum_{v\in V'}{P_\Q(v, k)} = 0$, for every $k\in \K$.
        \end{claim}
        \begin{proof}
            By induction on the applications of Theorem~\ref{theorem:charge-in-S}, for each $S\in \Q$, $P_{S'}$ satisfies Invariant~\ref{def:invariant}. 
            Therefore, for each $k\in \K$,
            \begin{align*}
                \sum_{v\in V'}{P_\Q(v, k)} &= \sum_{v\in V'}{\sum_{S\in \Q}{P_{S'}(v, k)}} = \sum_{S\in \Q}{\sum_{v\in V'}{P_{S'}(v, k)}} 
                \\
                &\numeq{3} \sum_{S\in \Q}{\sum_{v\in S'}{P_{S'}(v, k)}} \numeq{4} \sum_{S\in \Q}{\sum_{v\in S}{P(v, k)}} = \sum_{v\in V}{P(v, k)} = 0,
            \end{align*}
            where Equality~(3) holds since $P_{S'}$ is supported on $S'$ and
            Equality~(4) uses Invariant~\ref{def:invariant}\eqref{inv-2}.
        \end{proof}
    
        The following claim, whose proof is deferred to Appendix~\ref{appendix:omitted-proofs}, 
        shows that $\dem_P(U,W)$ and $\dem_{P_{\Q_0}}(U',W')$ are the same up to some minor error that arises due to split nodes on the cut $(U',W')$ of $G'$.
    
        \begin{claim}
        \label{claim:initial_flow_state}
            $\dem_P(U,W)\le\dem_{P_{\Q_0}}(U',W') + \capacity_G(U,W)$. 
        \end{claim}

        Consider a replacement of $\Q$ by $\tilde{\Q} = \Q \cup \{S\} \setminus \{S_1, S_2, \ldots, S_t\}$ when we apply Theorem~\ref{theorem:charge-in-S}. 
        Recall that $P^{\mathrm{before}}_{S'}=\sum_{i=1}^t P_{S'_i}$ is replaced by $P^{\mathrm{after}}_{S'}$. Let $P^{\mathrm{before}}_{\tilde{\Q}} \defeq P_\Q$ be the sum of active demand states before the change and let $P^{\mathrm{after}}_{\tilde{\Q}} \defeq P_\Q - P^{\mathrm{before}}_{S'} + P^{\mathrm{after}}_{S'}$ be the sum after the change.
        We have $\dem_{\left(P^{\mathrm{before}}_{\tilde{\Q}} - P^{\mathrm{after}}_{\tilde{\Q}}\right)}(U',W') = \dem_{\left(P^{\mathrm{before}}_{S'} - P^{\mathrm{after}}_{S'}\right)}(U',W')$. This means that after accumulating all the charges in the applications of Theorem~\ref{theorem:charge-in-S} on all clusters $S \in T$, the total charge is  
        \begin{align*}
            \sum_{e\in E(U',W')}{\charge(e)} &= \sum_{e\in E(U',W')}\sum_{S\in T}{\charge_S(e)} = \sum_{S\in T}\sum_{e\in E(U',W')\cap G'[S']}{\charge_S(e)} 
            \\
            &\numge{1} \sum_{S\in T}{\dem_{\left(P^{\mathrm{before}}_{S'} - P^{\mathrm{after}}_{S'}\right)}(U',W')} 
            =
             \sum_{S\in T}{\dem_{\left(P^{\mathrm{before}}_{\tilde{\Q}} - P^{\mathrm{after}}_{\tilde{\Q}}\right)}(U',W') }
        \\ &\numge{2}\dem_{\left(P_{\Q_0} - P_{V'}\right)}(U',W') = \dem_{P_{\Q_0}}(U',W') 
            \ge \dem_P(U,W) - \capacity_G(U,W)~,
        \end{align*}
        where $\charge(e)$ is the total charge incurred on $e$, Inequality~$(1)$ follows from Remark~\ref{remark:respects_means_charge} and Inequality~(2)
        uses Fact~\ref{fact:extended_flow_problems} on a telescoping sum.
        Recall that $P_{V'}$ is the final demand state which is zero.
        By charging an additional unit to each cut edge, we get $\dem_P(U,W)\le\sum_{e\in E(U',W')}{\charge(e)}$. 
        
        To conclude, we bound the total charge on each edge as follows. The height of the decomposition tree is $O(\log n)$ and $\uptau=\frac{1}{\log n}$, which means that $\alpha$ is a constant for each cluster. Therefore, by Theorem~\ref{theorem:charge-in-S} and Remark \ref{remark:respects_means_charge}, each edge is charged $O(\log^2 n)$ for each cluster containing it, for a total charge of $O(\log^3 n)$ per edge, which gives $\dem_P(U,W) = O(\log^3 n)\cdot\capacity_{G'}(U',W') = O(\log^3 n)\cdot\capacity_{G}(U,W)$, as we wanted. Thus, the conditions of Corollary~\ref{cor:equiv-def-cut2} are satisfied, and Theorem~\ref{theorem:final} follows.
    \end{proof}

    \subsection{Proof of Theorem \texorpdfstring{\ref{theorem:charge-in-S}}{5.5}}
    \label{section:charging-in-S}

    The goal of this section is to prove Theorem~\ref{theorem:charge-in-S}. Recall the notations of this theorem: $S\in T$ is a cluster,
    and let $S_1, \ldots, S_t$ be its sub-clusters (\ie, $\{S_1,\ldots S_t\} = \{L_1,\ldots,L_\ell\}\cup\{R_1,\ldots,R_r\}$), computed by Theorems~\ref{theorem:merge-phase_part1} and~\ref{theorem:merge-phase-part2} with parameter $0 < \uptau \le 1$.
    For $i=1\ldots t$, $P_{S'_i}$ is a demand state satisfying Invariant~\ref{def:invariant} for $S_i$ with a given load $\alpha \ge 1$.

    Let $P^{\mathrm{before}}_{S'} \defeq \sum_{i=1}^{t}P_{S'_i}$ be the sum of the demand states on the sub-clusters.
    By Invariant~\ref{def:invariant}\eqref{inv-1}, all the non-zero demands in $P^{\mathrm{before}}_{S'}$ are residing at the nodes in $X_F\cup X_Y\cup X_B$ of $S'$ (these are the boundary edges of $\{S_i\}_{i=1}^t$). 
    Given these demand vectors, we wish to ``solve'' some of these demands by updating the demand states such that a positive demand and a negative demand of the same commodity are moved to the same vertex, and cancel out. Then we want to ``send'' the remaining demands to $X_B$ (the boundary of $S$).
    In other words, we replace $P^{\mathrm{before}}_{S'}$ by a demand state $P^{\mathrm{after}}_{S'}$, where the demands of $P^{\mathrm{after}}_{S'}$ lie on $X_B$ (\ie, satisfying Invariant~\ref{def:invariant} for $S$), and we show that $G'[S']$ $\Omega(\frac{\uptau}{\alpha\log n})$-respects $P^{\mathrm{before}}_{S'} - P^{\mathrm{after}}_{S'}$.

    We give a sketch of the argument via the following sequence of demand state updates:
    \begin{enumerate}
        \item Denote $P^{(1)}_{S'} \defeq P^{\mathrm{before}}_{S'}$. First, we deal with demands that originated inside $L$, \ie, with the demand states $P^{(1)}_{L'}=\sum_{i}{P_{L'_i}}$, where $L=\bigcup_i{L_i}$. We move the demands of $P^{(1)}_{L'}$ that reside on $X_Y\cap L'$ to $X_F$, using the flow guaranteed by Property~\ref{prop-2} of Theorem~\ref{theorem:merge-phase-part2}. This flow is routable in the graph with congestion $O(\frac{\alpha}{\uptau})$ and in particular, by Fact~\ref{fact:congestion_implies_respecting_cut}, is $\Omega(\frac{\uptau}{\alpha})$-respected by the graph. At this point, nodes in $X_F$ have $O(\frac{\alpha}{\uptau})$ load.
        \item Then, we use a scaled all-to-all demand matrix on $X_F$ to update the demand state and spread the demands uniformly on $X_F$. By Theorem \ref{theorem:merge-phase_part1} we know that this problem is $\Omega(\frac{\uptau}{\alpha\log n})$-respected by the graph. 
        This means that positive and negative demands of the same commodity cancel out, and we are only left with demand that has to enter or leave the cluster $L$. Because $L$ is a node in the hierarchical decomposition tree $T$ and $P$ is respected by $T$, the total amount of leftover demand is at most $\capacity_G(L, V\setminus L)$.
        \item Then, we apply the reverse demand transformations from step (1) to move the demands back to $X_Y$.\footnote{In fact, the demands are moved to only some edges of $Y$ (see the proof for details).}
        \item Finally, we take all the demands that are currently on $X_Y$ (either from $P_{R'} \defeq P^{(1)}_{S'} - P^{(1)}_{L'}$ or the resulting updated demands), and use the flow ensured by Property~\ref{prop-1} of Theorem~\ref{theorem:merge-phase-part2} to route them to $X_B$. This flow is routable in the graph with congestion $O(\alpha)$ and, by Fact~\ref{fact:congestion_implies_respecting_cut}, is $\Omega(\frac{1}{\alpha})$-respected by the graph.
    \end{enumerate}
    In all of these steps, we show that $G'[S']$ respects the updates.
     We do this in two different ways: whenever the update is implemented by an explicit flow (as in steps~$1$,$3$ and~$4$), we use the fact that the flow is routable with low congestion, and therefore, by Fact~\ref{fact:congestion_implies_respecting_cut}, the update is respected by $G'[S']$. In the step that relies on uniform spreading on $X_F$ (step~$2$), we instead use the expansion of $X_F$ to claim that the update is respected by $G'[S']$.  We formalize these arguments in the lemmas below.

     Let $P^{(1)}_{L'} \defeq \sum_{i=1}^{\ell} P_{L'_i}$ be the sum of the demand states on the sub-clusters within  $L$. The following lemma captures steps $1$ and $2$ in the above process.

    \begin{lemma}\label{theorem:charge-in-S_part1}
        We can replace $P^{(1)}_{L'}$ by a demand state $P^{(3)}_{L'}$ which satisfies the following:
        \begin{enumerate}
            \item 
            $P^{(3)}_{L'}$
            is supported and spread uniformly on $X_F$. Formally, 
            $P^{(3)}_{L'}(x_e) = \frac{1}{|X_F|}\sum_{v \in L'}P^{(1)}_{L'}(v)$,
            for $x_e\in X_F$, and $P^{(3)}_{L'}(v) = 0$ for $v\notin X_F$. Note that $P^{(3)}_{L'}$ is no longer supported on $L'$. 
            \item $G'[S']$ $\Omega\!\left(\frac{\uptau}{\alpha\log n}\right)$-respects $P^{(1)}_{L'} - P^{(3)}_{L'}$. Moreover, $P^{(1)}_{L'} - P^{(3)}_{L'}$ is a valid demand state.
        \end{enumerate}
    \end{lemma}

    \begin{proof}
        Initially, Invariant~\ref{def:invariant}\eqref{inv-3} gives that every boundary split-node $x_b\in X_B$ in the graph $G'[S']$ has load at most $\alpha$ in $P^{(1)}_{S'}$, while vertices of $(X_F\cup X_Y)\setminus X_B$ may have load at most $2\alpha$ (each of these vertices is a boundary node of two of the sub-clusters $\{S_i\}_{i=1}^t$, so two sub-clusters may have placed a load of $\alpha$ on it). That is, $\sum_{k\in \K}{|P^{(1)}_{S'}(x_b, k)|}\le \alpha$ for boundary nodes $x_b\in X_B$, and $2\alpha$ for nodes in $(X_F\cup X_Y)\setminus X_B$.  Due to Invariant~\ref{def:invariant}\eqref{inv-1}, $P^{(1)}_{L'}$ is supported only on $L'$.

        The first step is to move the demands of $P^{(1)}_{L'}$ from $X_Y\cap L'$ to $X_F$. Specifically, each vertex $x_y\in X_Y\cap L'$ 
        will send  $|P^{(1)}_{L'}(x_y, k)| = |\sum_{i=1}^\ell{P_{L'_i}(x_y, k)}|$ units (these units may be negative) of commodity $k$ to $X_F$, for every $k\in \K$.
        As mentioned above, this is at most $2\alpha$ for each vertex.  \
        
        We use Property~\ref{prop-2} of Theorem~\ref{theorem:merge-phase-part2} for the set $Y$, and ``move'' the demands that reside at these vertices to vertices of $X_F$. Property~\ref{prop-2} says that there exists a flow which can be routed with congestion $O(1)$ in $G'[S']$ such that every node $x_y\in X_Y$ sends $\mu_\uptau(x_y)$ units of mass,
        which is at least $\uptau$ (if $y\in B$, then this is tight, otherwise $\mu_\uptau(x_y)=1 \ge \uptau$), and every node in $X_F$ receives at most $2$ units of mass. Every such $x_y\in X_Y$ has at most $2\alpha$ units of mass that it should send. Therefore, by scaling the given flow we can route this mass to nodes in $X_F$ with congestion $O(\frac{\alpha}{\uptau})$. 
        By Fact \ref{fact:congestion_implies_respecting_cut} we get that this routing is $\Omega(\frac{\uptau}{\alpha})$-respected by $G'[S']$. In particular, we can replace the demand state $P^{(1)}_{L'}$ with the updated demand state $P^{(2)}_{L'}$ after this routing, such that $G'[S']$ $\Omega(\frac{\uptau}{\alpha})$-respects $P^{(1)}_{L'} - P^{(2)}_{L'}$. 
        To be precise, we use the fact that $P^{(1)}_{L'}-P^{(2)}_{L'}$ is a valid demand state (recall Fact~\ref{fact:basic_flow_state}). 

        After this step, any node in $X_F$ has load of at most $2\cdot\frac{2\alpha}{\uptau} + 2\alpha \le \frac{6\alpha}{\uptau}$ units in $P^{(2)}_{L'}$ (as we scaled the flow by $\frac{2\alpha}{\uptau}$ and there was a load of at most $2\alpha$ at the beginning). 
        By Lemma~\ref{lemma:special_flows}(2), $P^{(2)}_{L'}$ is supported only on the vertices $X_F$ (since $P^{(1)}_{L'}$ was supported on $X_F\cup X_Y$ and we moved the demands on $X_Y\cap L'$).
        Note that $P^{(2)}_{L'}$ is not supported only on $L'$, as some vertices of $X_F$ may lie outside of $L'$.\footnote{This can occur when there exists $e=(u,v)\in F\cap Y$ with $u,v\in R$.}

        Next, we update $P^{(2)}_{L'}$ by uniformly
        spreading the demand vectors of $P^{(2)}_{L'}$  on $X_F$. For this we use the \emph{scaled all-to-all flow problem} of Lemma~\ref{lemma:special_flows}(1),
        which is defined with respect to the current demand state: each $u\in X_F$ sends $\frac{\norm{P^{(2)}_{L'}(u)}_1}{|X_F|}$ units of its own commodity to every $v\in X_F$. Let $Q_F$ be the corresponding demand matrix and let $N_F \defeq \frac{1}{|X_F|}\sum_{v\in X_F}{P^{(2)}_{L'}(v)}$ be the average demand vector that Lemma~\ref{lemma:special_flows}(1) yields, that is, $(P^{(2)}_{L'})^{\uparrow Q_F}(v) = N_F$, for all $v\in X_F$.
        Let $P^{(3)}_{L'} \defeq (P^{(2)}_{L'})^{\uparrow Q_F}$.
         
        Let $\tilde{Q}_F$ be the all-to-all demand matrix on $X_F$ as in  Definition~\ref{def:all_to_all_flow}. That is, $\tilde{Q}_F(u,v) = \frac{1}{|X_F|}$ for $u,v\in X_F$ and $0$ otherwise.
        By Theorem~\ref{theorem:merge-phase_part1}, the set $X_F$ $\Omega(\frac{1}{\log n})$-expands in $G'[S']$. Thus, by Remark~\ref{remark:mu_expansion_as_problem}, $\tilde{Q}_F$
        is $\Omega(\frac{1}{\log n})$-respected by $G'[S']$. 
        Since the rows of $Q_F$ differ from the rows of $\tilde{Q}_F$ by a factor of
        $O(\frac{\alpha}{\uptau})$, it follows, using Fact \ref{fact:basic_flow_state}(2), that the valid demand state $P^{(2)}_{L'} - P^{(3)}_{L'}$ is  $\Omega(\frac{\uptau}{\alpha \log n})$-respected by $G'[S']$. Using Fact \ref{fact:extended_flow_problems}, we get that $P^{(1)}_{L'} - P^{(3)}_{L'}$ is also $\Omega(\frac{\uptau}{\alpha \log n})$-respected by $G'[S']$.

        Finally, due to Fact~\ref{fact:basic_flow_state}(1), $P^{(1)}_{L'} - P^{(2)}_{L'}$ is a valid demand state and therefore:

        \[
            N_F = 
            \frac{1}{|X_F|}\sum_{v\in X_F}{P^{(2)}_{L'}(v)}
            \numeq{2}
            \frac{1}{|X_F|}\sum_{v\in V'}{P^{(2)}_{L'}(v)}
            =
            \frac{1}{|X_F|}\sum_{v\in V'}{P^{(1)}_{L'}(v)}
            \numeq{4}
            \frac{1}{|X_F|}\sum_{v\in L'}{P^{(1)}_{L'}(v)},
        \]
        where Equality $(2)$ holds since $P^{(2)}_{L'}$ is supported on $X_F$, and Equality $(4)$ holds since $P^{(1)}_{L'}$ is supported on $L'$.
        \end{proof}

        Let $P^{(3)}_{L'}$ be the demand state from Lemma~\ref{theorem:charge-in-S_part1}.
        All demands in $P^{(3)}_{L'}$ are distributed uniformly among vertices in $X_F$. Hence, positive demand units and negative demand units of the same commodity are ``paired-up'' and cancel out.
        
        Denote by $B^R\subseteq E(S, V\setminus S)$ the set of boundary edges that are incident to a node in $R$. Note that $B^R$ may not be all of $B$, as some boundary edges can be incident to a node in $L$, as long as these boundary edges are also in $Y$. Let $\tilde{Y} \defeq Y \setminus B^R$ (see Figure~\ref{fig:BR_Y_tildeY_explanation}).

     \begin{figure}[t]
        \begin{center}
            \includegraphics[scale=0.6]{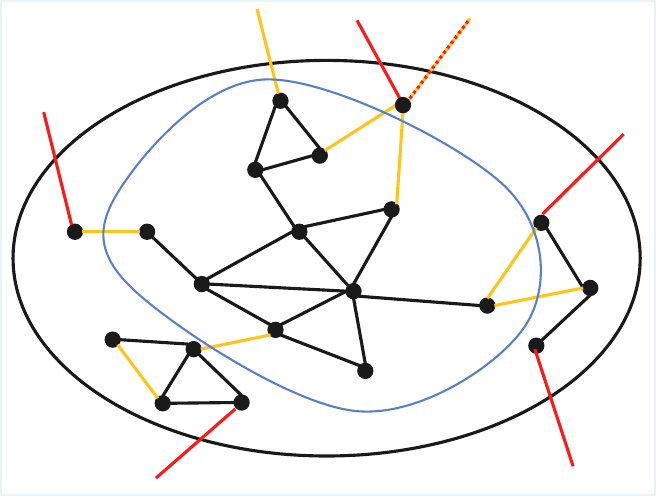}
        \end{center}
        \caption{Visualization of the sets $L,R,B^R, Y,$ and $\tilde{Y}$ in a cluster $S$. The core $L$ is the region outlined in blue. Red edges correspond to $B^R$, and golden edges correspond to $Y$. The golden-red edge lies in $B^R\cap Y$. The set $\tilde{Y}$ is defined as $\tilde{Y} \defeq Y \setminus B^R$. }
        \label{fig:BR_Y_tildeY_explanation}
    \end{figure}

        A key observation is that since $L$ is node in the hierarchical decomposition tree $T$, and the original demand state $P$ is respected by $T$, the amount of mass that has to cross the cut $(L, V\setminus L)$ is bounded by the capacity of this cut. This is the reason we included $L$ in the hierarchical decomposition tree. We formalize this observation in the following claim: 
        \begin{claim}
        \label{claim:opt_constraint}
            Let $X\subseteq V'$ and let $\tilde{P}$ be a demand state that is supported on $X$ and distributed uniformly on it. Let $P$ be a demand state that is $1$-respected by $T$. Let $L\in V_T$ be a cluster in the hierarchical decomposition and assume that $\sum_{v\in X}{\tilde{P}(v)} = \sum_{v\in L}{P(v)}$. Then, $\sum_{v\in X}{\norm{\tilde{P}(v)}}_1 \le \capacity(L, V\setminus L)$.
        \end{claim}
        \begin{proof}
            Because $\tilde{P}$ is spread uniformly on $X$, we have for each commodity $k\in \K$:
            \[
                \left|\sum_{v\in X}{\tilde{P}(v, k)} \right| = \sum_{v\in X}{\left| \tilde{P}(v, k) \right|}.
            \]
            Therefore,
            \begin{align*}
                \sum_{v\in X}{\norm{\tilde{P}(v)}_1} &=
                \sum_{v\in X}{\sum_{k\in \K}{\left| \tilde{P}(v, k) \right|}} = \sum_{k\in \K}{\sum_{v\in X}{\left| \tilde{P}(v, k) \right|}} 
                \\
                &= \sum_{k\in \K}{\left|\sum_{v\in X}{\tilde{P}(v, k)}\right|} 
                = \sum_{k\in \K}{\left|\sum_{v\in L}{P(v, k)}\right|} = \dem_P(L, V\setminus L).
            \end{align*}
            $L$ corresponds to a node in $T$, and $P$ is $1$-respected by $T$. Therefore, we must have $\dem_P(L, V\setminus L) \le \capacity_G(L, V\setminus L)$, which gives the result.
        \end{proof}
                
        We can now prove the following corollary:
        \begin{corollary}
        \label{cor:obs_F}
            The total load on the nodes $X_F$ in $P^{(3)}_{L'}$ satisfies
            \[
                \sum_{v\in X_F}{\norm{P^{(3)}_{L'}(v)}_1} \le \capacity_G(\tilde{Y}).
            \]
        \end{corollary}
        \begin{proof}
            For each commodity $k\in \K$ and $i\in\{1,\ldots,\ell\}$, using Invariant~\ref{def:invariant}\eqref{inv-2}, we have:
            \[
                \sum_{v\in L'_i}{P_{L'_i}(v, k)} = \sum_{v\in L_i}{P(v, k)}.
            \]
            Therefore, 
            \[
                \sum_{v\in V'}{P^{(1)}_{L'}(v, k)}
                =
                \sum_{v\in L'}{P^{(1)}_{L'}(v, k)}
                = \sum_{v\in L}{P(v, k)},
            \]
            where the first equality holds since $P^{(1)}$ is supported on $L'$. Since the updates that moved the demand vectors from $X_Y$ to $X_F$ and the all-to-all routing do not change the total demand (\ie, the difference $P^{(3)}_{L'} - P^{(1)}_{L'}$ is a valid demand state), we also get for each commodity $k\in \K$ that
            \[
                \sum_{v\in X_F}{P^{(3)}_{L'}(v, k)} = \sum_{v\in V'}{P^{(3)}_{L'}(v, k)} =
                \sum_{v\in V'}{P^{(1)}_{L'}(v, k)}=
                \sum_{v\in L}{P(v, k)},
            \]
            where the first equality is because $P^{(3)}_{L'}$ is supported on $X_F$. Therefore, we get that the demand state vectors satisfy $\sum_{v\in X_F}{P^{(3)}_{L'}(v)} = \sum_{v\in L}{P(v)}$.

            Applying Claim \ref{claim:opt_constraint} with $X \defeq X_F, \tilde{P}\defeq P^{(3)}_{L'}$ gives that $\sum_{v\in X_F}{\norm{P^{(3)}_{L'}(v)}_1} \le \capacity_G(L, V\setminus L)$.
            Since $\tilde{Y}$ is a super-set of the edges that leave the core $L$ in $G$ (see Figure~\ref{fig:BR_Y_tildeY_explanation}), we must have $\capacity_G(L, V\setminus L) \le \capacity_G(\tilde{Y})$,  which gives the result.
        \end{proof}

    Next, we distribute the remaining demands in $P^{(3)}_{L'}$ of nodes in $X_F$ to nodes in $X_{\tilde{Y}}$.

    \begin{lemma}\label{theorem:charge-in-S_part2}
        We can replace $P^{(3)}_{L'}$ by a demand state $P^{(4)}_{L'}$ which satisfies the following:
        \begin{enumerate}
            \item 
            $P^{(4)}_{L'}$
            is supported on $X_{\tilde{Y}}$. Moreover, 
            $\norm{P^{(4)}_{L'}(x_e)}_1 \le 1$, for every $x_e \in  X_{\tilde{Y}}$.
            \item $G'[S']$ $\Omega(\frac{\uptau}{\log n})$-respects $P^{(3)}_{L'} - P^{(4)}_{L'}$. Moreover, $P^{(3)}_{L'} - P^{(4)}_{L'}$ is a valid demand state.
        \end{enumerate}
    \end{lemma}

    \begin{proof}
        By Property~\ref{prop-2} of Theorem~\ref{theorem:merge-phase-part2} and since $X_F$ $\Omega(\frac{1}{\log n})$-expands in $G'[S']$, the demand matrix which ``sends'' one unit of flow from every $x_y\in X_{\tilde{Y}}\subseteq X_Y$  to a uniform distribution over $X_F$ is $\Omega(\frac{\uptau}{\log n})$-respected by $G'[S']$. 
        Indeed, as done in Lemma~\ref{theorem:charge-in-S_part1}, we can first route the demands from $X_{\tilde{Y}}$ to $X_F$ and then spread the demands uniformly. By reversing these demands movements, we can distribute the demands of $P^{(3)}_{L'}$ remaining on $X_F$ uniformly on the nodes in $X_{\tilde{Y}}$. 
        By Corollary~\ref{cor:obs_F}, each node $x_e\in X_{\tilde{Y}}$ receives at most one unit of demand. The new demand state $P^{(4)}_{L'}$ satisfies that $P^{(3)}_{L'} - P^{(4)}_{L'}$ is $\Omega(\frac{\uptau}{\log n})$-respected by $G'[S']$. 
    \end{proof}

    We are ready to prove Theorem~\ref{theorem:charge-in-S}.

    \begin{proof}[Proof of Theorem~\ref{theorem:charge-in-S}]
        Denote $P^{(1)}_{R'} \defeq \sum_{i=1}^r{P_{R'_i}}$.
        Observe that $P^{(1)}_{S'} = P^{(1)}_{L'} + P^{(1)}_{R'}$.
        We apply Lemmas~\ref{theorem:charge-in-S_part1}, and~\ref{theorem:charge-in-S_part2} consecutively.
        Let $P^{(4)}_{S'} = P^{(4)}_{L'} + P^{(1)}_{R'}$ be the resulting demand state. Observe that $P^{(4)}_{S'}$ is not supported on any node in $X_F\setminus X_Y$. Additionally, $P^{(4)}_{L'}$ is supported only on $X_{\tilde{Y}}$.
        
        To summarize the state so far, note that only vertices in $X_B\cup X_Y$ have non-zero demands. The vertices of $X_B\setminus X_Y \subseteq X_{B^R}$
        have load at most $\alpha$ (this load comes from $P^{(1)}_{R'}$). The vertices of $X_Y\setminus X_B$ have load at most $2\alpha + 1$: at most $2\alpha$ from $P^{(1)}_{R'}$ and at most $1$ unit from $P^{(4)}_{L'}$. Crucially, the vertices in $X_B\cap X_Y$ also have at most $\alpha$ units of demand in $P^{(4)}_{S'}$. Indeed, if $x_e\in X_B \cap X_Y$ is incident to a node $v\in R$ (i.e., $x_e\in X_{B^R}$), then $x_e$ carries at most $\alpha$ units from $P^{(1)}_{R'}$ and no additional units from $P^{(4)}_{L'}$. Otherwise, $x_e$ is incident to a node $v\in L$, so $x_e\in X_{\tilde{Y}}$. Therefore, $x_e$ does not have any demand from $P^{(1)}_{R'}$, and has at most $1\le \alpha$ units from $P^{(4)}_{L'}$.
        
        \begin{claim}
        \label{claim:Y-to-B}
            We can route a flow in $G'[S']$ which moves the demands in $P^{(4)}_{S'}$ from $X_Y\setminus X_B$ to the boundary nodes $X_B$, such that each node $x_b\in X_B$ receives at most $6\alpha\uptau$ units of demand (in addition to at most $\alpha$ units of existing load already on these vertices in $P^{(4)}_{S'}$). This flow can be routed with congestion $O(\alpha)$.
        \end{claim}
        \begin{proof}
            Property~\ref{prop-1} of Theorem~\ref{theorem:merge-phase-part2} promises a flow routable in $G'[S']$ with congestion $O(1)$, which sends $1$ unit of flow from every node $x_y\in X_Y\setminus X_B$, such that every node $x_b\in X_B$ receives at most $2\uptau$ units of flow. To get the flow for Claim \ref{claim:Y-to-B}, we scale this flow by $2\alpha + 1 \le 3\alpha$.
        \end{proof}
        We can therefore update $P^{(4)}_{S'}$ with the routing of Claim~\ref{claim:Y-to-B} to get the new demand state $P^{(5)}_{S'}$, such that $P^{(4)}_{S'}-P^{(5)}_{S'}$ is $\Omega(\frac{1}{\alpha})$-respected by $G'[S']$.
        
        We can now wrap up the proof of Theorem~\ref{theorem:charge-in-S}. The resulting demand state $P^{\mathrm{after}}_{S'}$ in Theorem~\ref{theorem:charge-in-S} is $P^{(5)}_{S'}$. It satisfies Invariant~\ref{def:invariant} with $\alpha' = \alpha\cdot(1 + 6\uptau)$. Indeed, by Claim~\ref{claim:Y-to-B}, the demands of $P^{(5)}_{S'}$ only reside on nodes of $X_B$ and the demand on each node is bounded by $\alpha'$.
        
        Finally, by Fact \ref{fact:extended_flow_problems}, we can see that $G'[S']$ $\gamma$-respects $P^{(1)}_{S'}-P^{(5)}_{S'}$, where
        \[
            \gamma = \frac{1}{O\!\left(\frac{\alpha\log n}{\uptau}\right) +
            O\!\left(\frac{\log n}{\uptau}\right) + 
            O\!\left(\alpha\right)} =
            \Omega\!\left(\frac{\uptau}{\alpha \log n}\right).
        \]
        This completes the proof of Theorem~\ref{theorem:charge-in-S}.
    \end{proof}

    \begin{remark}
        In the context of oblivious routing schemes, the construction of \cite{racke2014computing} was oblivious in the sense that the flow paths they construct between each pair of vertices $s, t$ were independent of the demands.

        Similarly, in our case, the charges on the edges of $(U,W)$ may be defined for each pair $s, t$, independently of the demands of $P$. The construction guarantees that for any valid demand state $P$, which is $\Omega(1)$-respected by $G$ (and thus by $T$), if we add up the ``oblivious'' charges incurred by each $s, t$ pair, scaled up by the demands of $P$, we will incur a charge of $O(\log^3 n)$ on each edge.        
    \end{remark}

    \section{Improved Tree Cut-Sparsifier}
    \label{section:improvement}
    In this section, we prove Theorem~\ref{theorem:final-improved}, establishing an improved tree cut-sparsifier of quality $O(\log^2 n \log \log n)$. 
    
    We begin by summarizing the results from Sections~\ref{section:merge-phase} and~\ref{section:charging-scheme}:
    \begin{theorem}
    \label{theorem:merge_phase}
        There is an algorithm ``merge-phase'' that given a cluster $S$ and $0<\uptau\le 1$, runs in $\tilde{O}(m)$ time and partitions $S$ into sub-clusters $S_1 ,\ldots,S_k$, such that:
        \begin{enumerate}
            \item $|S_i| \le \frac{2}{3} |S|$ for each $i$, and
            \item For every collection of demand states $\{P_{S'_i} \}_{i=1}^k$ that satisfy Invariant~\ref{def:invariant} with load $\alpha$, we can replace $P^{\mathrm{before}}_{S'} = \sum_{i=1}^{k} P_{S'_i}$ by a demand state $P^{\mathrm{after}}_{S'}$ that satisfies Invariant~\ref{def:invariant} with load 
            $\alpha' = \alpha\cdot(1 + c\cdot\uptau)$ for some constant $c > 0$, and
            \item $G'[S']$ $\Omega(\frac{\uptau}{\alpha\log n})$-respects $P^{\mathrm{before}}_{S'} - P^{\mathrm{after}}_{S'}$. Moreover, $P^{\mathrm{before}}_{S'} - P^{\mathrm{after}}_{S'}$ is a valid demand state.
        \end{enumerate}
    \end{theorem}

    In this section we show how to improve the quality of the resulting tree cut (and flow) sparsifier by a factor of $\Theta\left(\frac{\log n}{\log \log n}\right)$, thus proving Theorem \ref{theorem:final-improved}. 

    Recall from Section \ref{section:merge-phase} that we paid a factor of $\frac{1}{\uptau}$ in the quality of the tree cut-sparsifier, in order to accumulate a factor of just $1 + O(\uptau)$ to the load of each boundary edge as we move up the tree. This forced us to choose $\uptau=O(\frac{1}{\log n})$ to make the total load $O(1)$ in each step.
    We wish to avoid paying this additional factor of $\frac{1}{\uptau}=\Omega(\log n)$. Recall that even though each sub-cluster $S$ is a node in the hierarchical decomposition tree (and therefore the total amount of demand that has to leave it is at most the capacity of its boundary, because $P$ is respected by $T$), the distribution of demands on the boundary of $S$ is not uniform, so positive and negative demands of the same commodity do not cancel out.    
    
    We  amend this using a variation on the ``refinement phase'' technique of~\cite{bienkowski2003practical, harrelson2003polynomial, RackeS14},\footnote{\cite{harrelson2003polynomial} called it the ``bandwidth phase'' while \cite{bienkowski2003practical} called it the ``assure-precondition'' subroutine.} 
    to ensure that the all-to-all flow problem (recall Definition~\ref{def:all_to_all_flow})
    between the boundary edges of the sub-clusters we recur to, is respected by the sub-clusters, meaning that the boundary edges expand in these sub-clusters (see Remark~\ref{remark:mu_expansion_as_problem}).\footnote{This requirement is weaker than what was ensured in \cite{bienkowski2003practical, harrelson2003polynomial, RackeS14}. First, we only require the all-to-all flow problem to be respected by the sub-cluster, not routable as in \cite{bienkowski2003practical,harrelson2003polynomial}. Additionally, they ensured that a slightly stronger flow problem is routable, which implies the all-to-all problem on the boundary edges is routable (but not conversely). However, the boundary edges being expanding will be enough for our purposes.
    }
    
    The resulting scheme is a two-level decomposition: given an input cluster $S$, we  first run one level of the merge phase decomposition described in Theorem \ref{theorem:merge_phase}, with $\uptau = 1$. This partitions $S$ to sub-clusters $\{S_1, \ldots, S_k\}$, such that $|S_i| \le \frac{2}{3}|S|$. This merge phase 
    partition of $S$ corresponds to $2$ layers in the hierarchical decomposition (see Figure~\ref{fig:refine-then-merge}). On each $S_i$ we then run the refinement phase described below. This partitions $S_i$ into $\{S_{i, 1}, \ldots, S_{i, k_i}\}$, such that the boundary edges of $S_{i, j}$ (\ie, $E(S_{i, j}, V\setminus S_{i, j})$) expand within $G'[S_{i, j}']$. Formally, we denote the boundary edges of $S_{i, j}$ as $B_{i, j} \defeq E(S_{i, j}, V\setminus S_{i, j})$ and prove that $X_{B_{i, j}}$ $\Omega\left(\frac{1}{f(S_{i, j})\cdot \log n}\right)$-expands in $G'[S_{i, j}']$, where $f$ is defined in Theorem~\ref{theorem:refinement_phase}.
    The reason we need the merge phase is that we have no guarantee that the size of $S_{i, j}$ is smaller than $S_i$ by a constant factor. 
    Therefore, we repeat by performing an alternation between the merge phase  on each $S_{i, j}$ and then the refinement phase on each produced sub-cluster, etc.

    Even though $\uptau = 1$, in each sub-cluster $S_{i, j}$ created by the refinement phase, we can send the demands to a uniform distribution on the boundary edges of $S_{i, j}$, so positive and negative demands cancel out and we are left with at most a load of $1$ on each boundary edge. This effectively ``resets'' the load on each boundary edge so it does not accumulate when we go up along the hierarchical tree. This scheme allows us to improve the charge incurred, in total over all merge phase sub-clusters, from $O(\log^3 n)$ (by setting $\uptau = \frac{1}{\log n}$) to $O(\log^2 n)$ ($\uptau =1$) per cut edge. As for the sub-clusters of the refinement phase, we  show that they incur a total of $O(\log^2 n \log\log n)$ charge per cut edge.

    The main difference from \cite{bienkowski2003practical, harrelson2003polynomial, RackeS14} is that we are able to run the refinement phase in near-linear time. Additionally, we avoid dealing with signaled events from recursive sub-calls, leading to a simpler algorithm.
    Explicitly, we prove the following theorem:

    \begin{theorem}
    \label{theorem:refinement_phase}
        There is an algorithm ``refinement-phase'' that given a cluster $S$ and a parameter $\sigma \ge \frac{3}{2}|S|$, 
        runs in $\tilde{O}(m)$ time and partitions $S$ into sub-clusters $S_1 ,\ldots,S_k$,  such that:
        \begin{enumerate}
            \item With high probability, for each $i$, the sub-cluster $G'[S_i']$, $\Omega\!\left(\frac{1}{f(S_i)\cdot \log n}\right)$-respects the all-to-all flow problem\footnote{
            See Definition~\ref{def:all_to_all_flow}. Additionally, recall the definition of a graph respecting a flow problem in Definition~\ref{def:cut_respects_flow}.}
            on $X_{B_i}$, where $B_i = E(S_i, V\setminus S_i)$. 
            Here $f(S_i) \defeq c_f\log\left(\frac{\sigma}{|S_i|}\right)\cdot \log \log n$, where $c_f\ge 1$ is a constant which will be defined below.
            \item We can route in $G'[S']$ a demand matrix where each inter-cluster split node $x_f \in X_F$ (where $F = \{(u, v)\in E\mid u\in S_i, v\in S_j, i\neq j\}$) sends one unit of flow, and each boundary split node in $X_B$ (where $B = E(S, V\setminus S)$) receives at most $O(1)$ units of flow. This can be routed with congestion $O(\log n)$.
        \end{enumerate}
    \end{theorem}

    \paragraph{Intuition for $f$ and $\sigma$.} The quantity $f(S_i) \defeq c_f\log\left(\frac{\sigma}{|S_i|}\right)\cdot \log \log n$ plays two distinct roles in the analysis. 
    The parameter $\sigma$ is a \emph{global scale parameter} that controls how the costs incurred by the refinement phase aggregate across the levels of the hierarchical decomposition, allowing for a telescoping summation.
    At the same time, the quantity $f(S_i)$ provides a \emph{local guarantee} whose strength depends on the size of the sub-cluster $S_i$. Larger sub-clusters admit stronger all-to-all respect guarantees. This is crucial for the routing guarantee in Item~(2) of Theorem~\ref{theorem:refinement_phase}, since the routing will proceed in a bottom-up manner with respect to the refinement-induced hierarchy on $S$.
    Intuitively, if a refinement step significantly reduces the size of a sub-cluster, then the expansion requirement can be relaxed, as the smaller cluster contributes less to congestion (as they tend to be lower at the hierarchy). Conversely, if the cluster remains large, the analysis forces a stronger expansion guarantee. Thus, the refinement phase always makes progress either by enforcing strong expansion or by substantially shrinking the cluster size.

    We prove this theorem in Section~\ref{section:refinement_phase}. In the rest of this section we show how to improve the cut-sparsifier result.

    Our construction, see Algorithm~\ref{algo:alternating_construction}, begins with the merge phase (Theorem~\ref{theorem:merge_phase}) on $S = V$ with $\uptau = 1$. Then, we run the refinement phase (Theorem~\ref{theorem:refinement_phase}) on each of the resulting sub-clusters $S_i$, with $\sigma = |S|$. We proceed with the merge phase with $\uptau = 1$ on each resulting sub-cluster $S_{i,j}$ (i.e., the sub-clusters created from $S_i$ in the previous refinement step), 
    and then with the refinement phase on $S_{i, j, k}$, etc. In each subsequent run of the refinement phase, we choose $\sigma$ to be the size of the previous cluster created by a refinement phase, \ie{} the parent of the current sub-cluster (\eg, when running on $S_{i,j,k}$, we take $\sigma = |S_{i,j}|$, not $|S_{i, j, k}|$.). 
    Because $S_{i}$ is a sub-cluster created in a merge phase, it will satisfy $|S_{i}| \le \frac{2}{3}\cdot |V|$, ensuring the required condition on $\sigma$. We proceed recursively in this manner, running both stages on successively smaller clusters $S$, stopping when $|S| = 1$. 
    Let $T$ be the resulting hierarchical decomposition. Note that because the merge phase reduces the size of the sub-clusters by a constant fraction, the depth of $T$ is logarithmic in $n$.

    \begin{figure}[t]
        \begin{center}
            \includegraphics[scale=0.45]{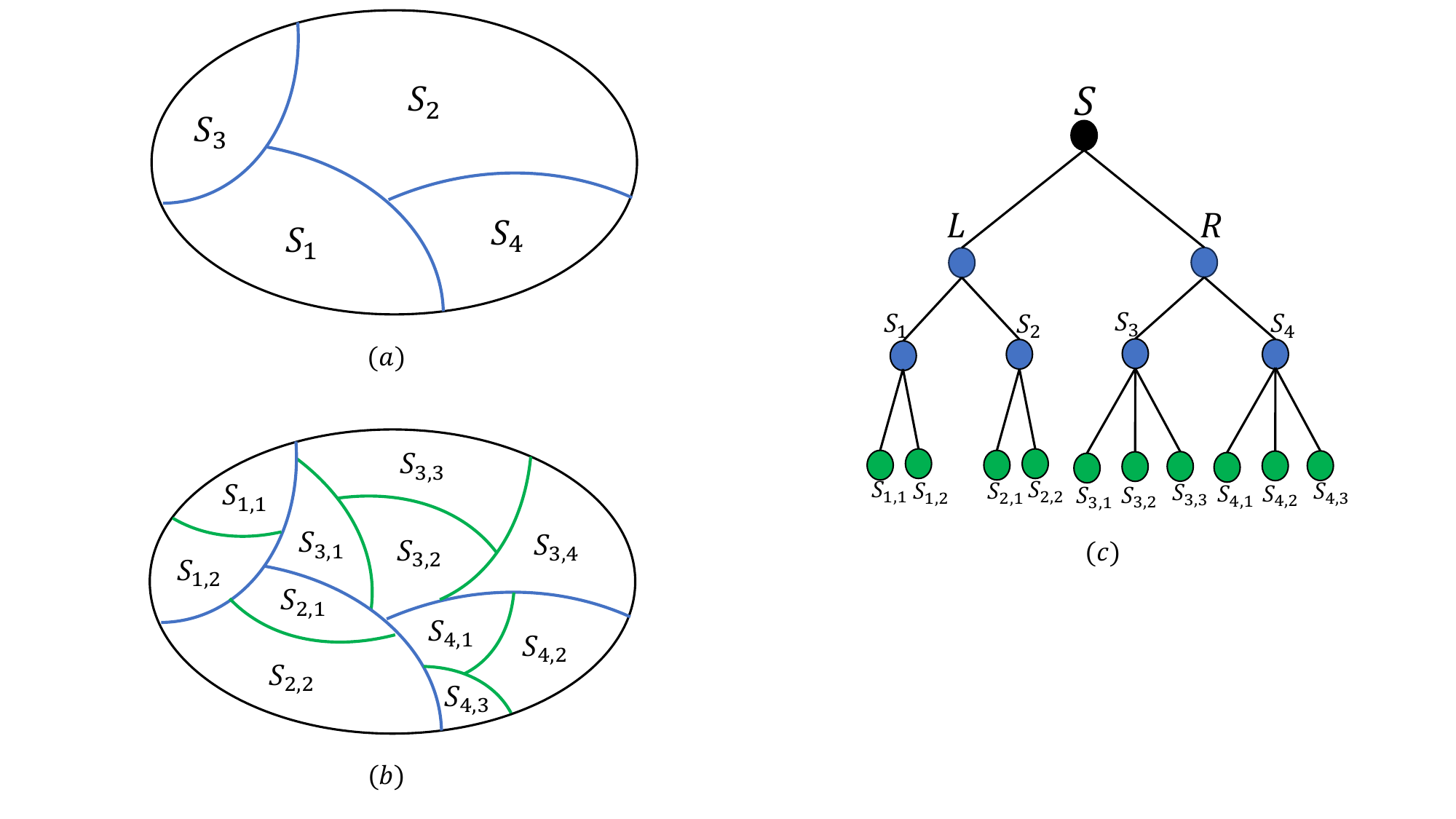}
        \end{center}
        \caption{ The two-step partition of a cluster $S$. The merge phase (Theorem~\ref{theorem:merge_phase}) is shown in Figure~(a). The refinement phase (Theorem~\ref{theorem:refinement_phase}) is then applied to the sub-clusters, as shown in Figure~(b). The corresponding hierarchical tree is given in Figure~(c). Note that the merge phase contributes two levels to the hierarchical tree due to the partition into $L$ (the core) and $R$. 
        }
        \label{fig:refine-then-merge}
    \end{figure}

    We begin by re-introducing the relevant notation from Section~\ref{section:charging-scheme}.
    As per Corollary~\ref{cor:equiv-def-cut2}, we fix a cut $(U\subseteq V, W = V\setminus U)$ in $G$ and a valid demand state $P$ on $V$, which is $1$-respected by $T$. We work in the subdivision graph $G'=(V\cup X_E, E')$.
    We define for a cluster $S\subseteq V$, the corresponding cluster in the subdivision graph as $S'\defeq S\cup\{x_e\mid e\cap S\ge 1\}$ and let $G'[S']$ be the subdivision graph restricted to the set $S'$. We let $(U' = U\cup \{x_e \mid e\cap U \ge 1\}, W' = V'\setminus U')$ be the cut corresponding to $(U,W)$ in $G'$. We will charge the edges $E(U',W')$ in $G'$ to cover the demand $\dem_{P'}(U',W')$, where $P'$ will be the appropriate demand state on $G'$.

      \begin{algorithm}[hbt!]
    \caption{Construction of the tree Cut-Sparsifier}
    \label{algo:alternating_construction}
        \begin{algorithmic}[1]
        
            \Function{BuildTree}{$G=(V,E)$}
              \State Create a tree $T$ with root labeled by the cluster $V$.

              \State \Call{AltMergeRefine}{$G,V$}

              \State \Return $T$
            \EndFunction

            \\
            
            \Function{AltMergeRefine}{$G,S$}
              \If{$|S|\le 1$}
                \State \Return
              \EndIf
            
              \State $\sigma \gets |S|$ \Comment{Previous refinement phase cluster, or the root $V$}
            
              \State $(L,R,\{L_1,\ldots,L_\ell\},\{R_1,\ldots,R_r\})
                     \gets \Call{MergePhase}{G,S,\tau = 1}$.
              \State Make $L$ and $R$ children of $S$ in $T$.
            
              \ForAll{$C \in \{L_1,\ldots,L_\ell\} \cup \{R_1,\ldots,R_r\}$}
                \State Make $C$ a child of its corresponding parent ($L$ or $R$).
                \State $\{C_1,\ldots,C_k\} \gets \Call{Refine}{G,C,\sigma}$ \Comment{partition of $C$}
                \ForAll{$C_i \in \{C_1,\ldots,C_k\}$}
                  \State Make $C_i$ a child of $C$.
                  \State \Call{AltMergeRefine}{$G,C_i$}.
                \EndFor
              \EndFor
            \EndFunction
        
        \end{algorithmic}
    \end{algorithm}

    Like in Section~\ref{section:charging-scheme}, we convert the valid demand state $P$ to demand states $P_{\{v\}'}$
    for each vertex $v\in V$, by distributing the demand vector of $v$ to the split nodes neighboring $v$ evenly. These demand states satisfy Invariant~\ref{def:invariant} with $\alpha=1$. If we denote the sum of these initial demand states as $P_{\Q_0}$, Claim~\ref{claim:initial_flow_state} shows that $\dem_P(U,W)\le\dem_{P_{\Q_0}}(U',W') + \capacity_G(U,W)$.

    We will first prove the following improvement of Theorem~\ref{theorem:charge-in-S}.
    \begin{theorem}
    \label{theorem:charge-in-S_improved}
        Let $S$ be a cluster, and let $S_1, \ldots, S_t$ be its merge phase sub-clusters, computed with $\uptau = 1$. For each $1\le i\le t$, let $S_{i,1}, \ldots, S_{i,t_i}$ be the refinement phase sub-clusters of $S_i$, computed with $\sigma = |S|$. For each $1\le i\le t, 1\le j\le t_i$, assume that $P_{S'_{i,j}}$ is a demand state satisfying Invariant~\ref{def:invariant} for $S_{i,j}$ with a given $\alpha \ge 1$. Let $P^{\mathrm{before}}_{S'}=\sum_{i=1}^t \sum_{j=1}^{t_i} P_{S'_{i,j}}$ be their sum.
        We can replace $P^{\mathrm{before}}_{S'}$ with a demand state $P^{\mathrm{after}}_{S'}$ which satisfies Invariant~\ref{def:invariant} for $S$ with $\alpha' = O(1)$ (which is independent of $\alpha$),
        while charging each edge of the cut $E(U',W')\cap G'[S_{i, j}']$ (for each $i, j$) at most $O(\alpha\cdot\log n\cdot f(S_{i, j}))$ to cover the difference in demands.\footnote{$f(S_{i,j})$ is defined in Theorem~\ref{theorem:refinement_phase}. Additionally, note that each edge in $G'[S']$ is part of exactly one sub-cluster $G'[S'_{i,j}]$.}
        That is, we define $\charge_S(e)$ 
        for each $e\in E(U',W')\cap G'[S']$  such that $\sum_{e\in E(U',W')\cap G'[S']}{\charge_S(e)}\ge \dem_{\left(P^{\mathrm{before}}_{S'}-P^{\mathrm{after}}_{S'}\right)}(U',W')$ 
        and each $e\in E(U',W')\cap G'[S_{i, j}']$ is charged at most $O(\alpha\cdot\log n\cdot f(S_{i, j}))$. 
        Note that the charged edges are in the subdivision graph. 

    \end{theorem}

    \begin{proof}
        Let $P^{(1)}_{S'_{i,j}} \defeq P_{S'_{i,j}}$.
        We apply the first guarantee of Theorem~\ref{theorem:refinement_phase} on each of the clusters $S_{i,j}$ ($1\le i\le t, 1\le j\le t_i$), to distribute the demand states $P^{(1)}_{S'_{i,j}}$ uniformly on the boundary edges of the cluster. Formally, define the average demand vector $N_{i,j} \defeq \frac{1}{|X_{B_{i,j}}|}\sum_{v\in X_{B_{i,j}}}{P^{(1)}_{S'_{i,j}}(v, \cdot)}$, where $B_{i,j} = E(S_{i,j}, V\setminus S_{i,j})$ are the boundary edges of the cluster $S_{i,j}$. Consider the demand state $P^{(2)}_{S'_{i,j}}$ which is supported on $X_{B_{i,j}}$, where each demand vector is $N_{i,j}$. By the construction of $P^{(2)}_{S'_{i,j}}$, 
        it satisfies Invariant~\ref{def:invariant} for $S_{i,j}$ with the same parameter $\alpha$ as $P^{(1)}_{S'_{i,j}}$.  
        In fact, $P^{(2)}_{S'_{i,j}}$ is obtained by updating $P^{(1)}_{S'_{i,j}}$ by the all-to-all demand matrix on $X_{B_{i, j}}$ scaled by at most $\alpha$ (see Lemma~\ref{lemma:special_flows}).
 Using the first guarantee of Theorem~\ref{theorem:refinement_phase}, and by Facts~\ref{fact:basic_flow_state}(2) and~\ref{fact:congestion_implies_respecting_cut}, the valid demand state $P^{(1)}_{S'_{i,j}} - P^{(2)}_{S'_{i,j}}$ is $\Omega\left(\frac{1}{\alpha\cdot f(S_{i,j})\cdot \log n}\right)$-respected by $G'[S'_{i,j}]$. 
        Therefore, by charging $O(\alpha\cdot f(S_{i,j})\cdot \log n)$ to each edge of $E(U',W')\cap G'[S_{i,j}']$, we are able to cover the demand $\dem_{\left(P^{(1)}_{S'_{i,j}}-P^{(2)}_{S'_{i,j}}\right)}(U',W')$.

        Now, all demands in $P^{(2)}_{S'_{i,j}}$ are distributed uniformly among vertices in $X_{B_{i,j}}$. Hence, positive demand units and negative demand units of the same commodity are ``paired-up'' and cancel out.

        Next, we use the fact that $S_{i,j}$ is a sub-cluster in the hierarchical decomposition tree $T$, which means that because the original $P$ is respected by $T$, the amount of demand $P$ induces across the cut $(S_{i,j}, V\setminus S_{i,j})$ is at most the capacity of this cut. This is made explicit in the following claim:
        \begin{claim}
        \label{claim:refinement_phase_boundary}
            The total load in $P^{(2)}_{S'_{i,j}}$ on the nodes $X_{B_{i,j}}$ satisfies
            \[
                \sum_{v\in X_{B_{i,j}}}\norm{{ P^{(2)}_{S'_{i,j}}(v) }}_1 \le \capacity_G(B_{i,j}).
            \]
        \end{claim}
        \begin{proof}
            For each commodity $k\in\K$, by Invariant~\ref{def:invariant}\eqref{inv-2}, 
            \[
                \sum_{v\in S'_{i,j}}{P^{(1)}_{S'_{i,j}}(v, k)} = \sum_{v\in S_{i,j}}{P(v, k)}.
            \]
            The update by the scaled all-to-all flow problem does not change the total demand (\ie, $P^{(1)}_{S'_{i,j}} - P^{(2)}_{S'_{i,j}}$ is a valid demand state), so we get 
            \[
                \sum_{v\in X_{B_{i,j}}}{P^{(2)}_{S'_{i,j}}(v, k)}=
                \sum_{v\in S'_{i,j}}{P^{(2)}_{S'_{i,j}}(v, k)} = 
                \sum_{v\in S_{i,j}}{P(v, k)}~,
            \]
            where the first equality holds because $P^{(2)}_{S'_{i,j}}$ is supported on $X_{B_{i,j}}$. The result now follows by applying Claim \ref{claim:opt_constraint} with $X \defeq X_{B_{i,j}}, L \defeq S_{i,j}$ and $\tilde{P}\defeq P^{(2)}_{S'_{i,j}}$.
                
        \end{proof}
        Using Claim \ref{claim:refinement_phase_boundary}, we have that the total load $P^{(2)}_{S'_{i,j}}$ puts on $X_{B_{i,j}}$ is at most $\capacity_G(B_{i,j})$, and because the demand vectors are the same,  we get that the load on each vertex is at most $1$.

        For every $1\le i \le t$, let $P^{(2)}_{S'_i}\defeq \sum_{j=1}^{t_i}{P^{(2)}_{S'_{i,j}}}$. Denote the inter-cluster edges inside $S'_i$ by $F_i\defeq\{(u,v)\in E\mid u\in S_{i,j_1}, v\in S_{i,j_2}, j_1\neq j_2\}$ and the boundary edges of the cluster by $B_i \defeq E(S_i, V\setminus S_i)$. We have that $P^{(2)}_{S'_i}$ is supported on $X_{F_i} \cup X_{B_i}$, such that the load on each of $x_f\in X_{F_i}$ is at most $2$, and the load on each $x_b\in X_{B_i}$ is at most $1$.
        
        Next, we use the second guarantee of Theorem~\ref{theorem:refinement_phase}. This allows us to route in $G'[S'_i]$, one unit of mass from each $x_f\in X_{F_i}$ to $X_{B_i}$, such that each vertex $x_b\in X_{B_i}$ receives at most $O(1)$ units of mass. Scaling this flow by $2$ gives that we can route two units from each inter-cluster edge to the boundary. This flow can be routed with congestion $O(\log n)$. Hence, by Fact~\ref{fact:congestion_implies_respecting_cut}, we can replace $P^{(2)}_{S'_{i}}$ with the updated demand state obtained after this routing, denoted $P^{(3)}_{S'_i}$, while incurring a charge of $O(\log n)$ on the edges $E(U',W')\cap G'[S'_i]$. Consequently, $P^{(3)}_{S'_i}$ satisfies Invariant~\ref{def:invariant} for $S_i$ with $\alpha = O(1)$, as the demand state is supported on nodes of $X_{B_i}$ and the load on each node is constant.

        Denote $P^{(3)}_{S'} \defeq \sum_{i=1}^{t}{P^{(3)}_{S'_i}}$. We now use Theorem~\ref{theorem:merge_phase} on the merge phase partition $\{S_1, \ldots, S_t\}$. It allows us to replace $P^{(3)}_{S'}$ with the updated $P^{(4)}_{S'}$ which satisfies Invariant~\ref{def:invariant} for $S$ with $\alpha = O(1)$, while charging each edge of $E(U',W')\cap G'[S']$ at most $O(\log n)$ to cover the difference in demands.

        Thus, $P^{\mathrm{after}}_{S'} \defeq P^{(4)}_{S'}$ is the demand state required by Theorem \ref{theorem:charge-in-S_improved}, and the total charge on each edge of $E(U',W')\cap G'[S']$ is dominated by $O(\alpha\cdot f(S_{i,j})\cdot \log n)$. 
    \end{proof}

    Using Theorem~\ref{theorem:charge-in-S_improved}, Theorem \ref{theorem:final-improved} follows exactly like Theorem \ref{theorem:final} in Section~\ref{section:charging-scheme}. We give the full details for completeness. 

    \begin{proof}[Proof of Theorem \ref{theorem:final-improved}]
        We deal with $P$ in a bottom-up manner: first we convert it to a set of demand states $P_{\{v\}'}$ as shown in Section~\ref{section:charging-scheme}. Then, we use Theorem~\ref{theorem:charge-in-S_improved} for each cluster in $T$. We accumulate the charges incurred during each application of the theorem.
    
        Denote by $\Q$ the set of current ``active'' clusters, such that initially $\Q=\Q_0\defeq\{\{v\} \mid v\in V\}$. Then in each iteration we remove $\{S_{i,j} \mid 1\le i \le t, 1\le j \le t_i\}$ from $\Q$ and add $S$ to $\Q$, where $S_{i,j}$ are the refinement phase sub-clusters as in Theorem~\ref{theorem:charge-in-S_improved}. 
        Let $P_\Q$ be the sum of all demand states $P_{S'}$, for $S\in \Q$. Using Claim~\ref{claim:total_active_demand_is_a_flow_problem}, we have that $P_\Q$ is always a valid demand state on $V'$. Claim~\ref{claim:initial_flow_state} shows that $\dem_P(U,W)\le\dem_{P_{\Q_0}}(U',W') + \capacity_G(U,W)$.
    
        Consider a replacement of $\Q$ by $\tilde{\Q} = \Q \cup \{S\} \setminus \{S_{i,j} \mid 1\le i \le t, 1\le j \le t_i\}$ when we apply Theorem~\ref{theorem:charge-in-S_improved}. Recall that $P^{\mathrm{before}}_{S'}=\sum_{i=1}^t \sum_{j=1}^{t_i} P_{S'_{i,j}}$ is replaced by $P^{\mathrm{after}}_{S'}$. Let $P^{\mathrm{before}}_{\tilde{\Q}} \defeq P_\Q$ be the sum of active demand states before the change and $P^{\mathrm{after}}_{\tilde{\Q}} \defeq P_\Q - P^{\mathrm{before}}_{S'} + P^{\mathrm{after}}_{S'}$ be the sum after the change.
        We have $\dem_{\left(P^{\mathrm{before}}_{\tilde{\Q}} - P^{\mathrm{after}}_{\tilde{\Q}}\right)}(U',W') = \dem_{\left(P^{\mathrm{before}}_{S'} - P^{\mathrm{after}}_{S'}\right)}(U',W')$. This means that after accumulating all the charges in the applications of Theorem~\ref{theorem:charge-in-S_improved} 
        on all clusters $S \in T$, the total charge satisfies  
        \begin{align*}
            \sum_{e\in E(U',W')}{\charge(e)} &= \sum_{e\in E(U',W')}\sum_{S\in T}{\charge_S(e)} = \sum_{S\in T}\sum_{e\in E(U',W')\cap G'[S']}{\charge_S(e)} 
            \\
            &\ge \sum_{S\in T}{\dem_{\left(P^{\mathrm{before}}_{S'} - P^{\mathrm{after}}_{S'}\right)}(U',W')} 
            =
            \sum_{S\in T}{\dem_{\left(P^{\mathrm{before}}_{\tilde{\Q}} - P^{\mathrm{after}}_{\tilde{\Q}}\right)}(U',W')}
            \\ &\ge\dem_{\left(P_{\Q_0} - P_{V'}\right)}(U',W').
        \end{align*}
        where $\charge_S(e)$ is the charge incurred on $e$ when applying Theorem~\ref{theorem:charge-in-S_improved} 
        on the cluster $S$, $\charge(e)$ is the total charge incurred on $e$, and the last inequality uses Fact~\ref{fact:extended_flow_problems}. 

      Since $P_{V'}$ satisfies Invariant~\ref{def:invariant}\eqref{inv-1} for $V$, it follows that $P_{V'} = 0$. Thus, we have
        \begin{align*}
            \sum_{e\in E(U',W')}{\charge(e)} \ge \dem_{P_{\Q_0}}(U',W') \ge \dem_P(U,W) - \capacity_G(U,W),
        \end{align*}
        where the second inequality follows from Claim~\ref{claim:initial_flow_state}. By charging an additional unit to each cut edge, we get $\dem_P(U,W)\le\sum_{e\in E(U',W')}{\charge(e)}$. 

        It remains to bound the total charge $\charge(e)$ on each edge $e\in E(U',W')$.
        The load $\alpha$ is $O(1)$ for each cluster, as it starts as $\alpha = 1$ and is constant after each application of Theorem~\ref{theorem:charge-in-S_improved}. Therefore, each edge (in the subdivision graph) $e\in E(U',W')$, is charged by Theorem \ref{theorem:charge-in-S_improved}, $O(\log n \cdot f(S))$ for each refinement phase\footnote{Note that the charging in Theorem~\ref{theorem:charge-in-S_improved} only charges refinement phase clusters.}
        of a cluster $S$ that contains it (\ie, $e\in G'[S']$).
        In particular, an edge $e\in E(U',W')$ is only charged by clusters $S$ which lie on a root-to-leaf path in $T$.

        Consider a path in the decomposition hierarchy $T$, $S_1=V,\ldots, S_k$, where $S_{i+1}$ is a refinement phase sub-cluster of $S_i$ (that is, we consider only refinement phase sub-cluster between them since these are the only charged sub-clusters in Theorem~\ref{theorem:charge-in-S_improved}). 
        Note that for all $1\le i < k$, when $S_{i+1}$ was created (as a refinement of some sub-cluster of $S_i$), we used $\sigma = |S_i|$ in Theorem~\ref{theorem:refinement_phase}. So, $f(S_{i+1}) = c_f\log\left(\frac{|S_{i}|}{|S_{i+1}|} \right)\cdot \log \log n$. 
        Therefore, it holds that 
        \begin{align*}
            \sum_{i=1}^{k}{f(S_i)} &= f(S_1) + \sum_{i=1}^{k-1}{f(S_{i+1})} = c_f\log\log n\cdot\left( \log\left(\frac{2n}{|S_{1}|}\right) + \sum_{i=1}^{k-1}{\log\left(\frac{|S_{i}|}{|S_{i+1}|}\right)} \right) 
            \\
            &= c_f\log \log n \cdot \log\left(\frac{2n}{|S_{k}|}\right) = O(\log n\cdot \log \log n)
        \end{align*}
        Therefore, the total charge is $O(\log^2 n \cdot \log\log n)$ per edge, which gives $\dem_P(U,W) = O(\log^2 n \cdot \log\log n)\cdot\capacity_{G'}(U',W') = O(\log^2 n \cdot \log\log n)\cdot\capacity_{G}(U,W)$, as we wanted. Theorem \ref{theorem:final-improved} then follows from Corollary~\ref{cor:equiv-def-cut2}.
    \end{proof}

    As a direct corollary of Theorem~\ref{theorem:final-improved}, we get Theorem \ref{theorem:final-improved-flow}.

    \subsection{Refinement Phase}
    \label{section:refinement_phase}
    The goal of this section is to prove Theorem~\ref{theorem:refinement_phase}. Before delving into the proof in Section \ref{section:refinement_phase_proof}, in Section \ref{section:refinement_phase_overview} we provide a high level overview of the techniques.
    
    \subsubsection{Overview}
    \label{section:refinement_phase_overview}
    Given a cluster $S\subseteq V$, we would like to check if its boundary edges $B_S = E(S, V\setminus S)$  expand within $G[S]$.\footnote{The formal argument is performed in the subdivision graph (\ie\  proving that $X_{B_S}$  expands in $G'[S']$). See Section \ref{section:refinement_phase_proof} for details. For this overview, we omit these details for simplicity.}
    If they are not, we would like to find a sparse cut $(A, S\setminus A)$ with respect to the corresponding weighting $\mu_S = \1_{B_S}$, that certifies that the boundary edges do not expand. We then recur on both $A$ and $S\setminus A$ (with the new weightings $\mu_A = \1_{B_A}$ and similarly $\mu_{S\setminus A}=\1_{B_{S\setminus A}}$). We stop when we find a sub-cluster in which the boundary edges expand (equivalently, the sub-cluster expands with respect to the weighting given by its boundary edges). This can be facilitated using Theorem~\ref{theorem:balanced_cut_mu}. 
    Note that when we recur into a sub-cluster (\eg, $A\subseteq S$), the new cut edges $E(S,A\setminus S)$ are now considered boundary edges of the sub-clusters $A$ and $S\setminus A$. However, because the cut was sparse, the capacity of these cut edges is small compared to the boundary capacity already present in $S$. Consequently, for each sub-cluster, the total capacity of its boundary edges decreases relative to that of its parent cluster, leading to a polylogarithmic recursion depth. 
    Note that we must carefully adjust the target expansion $\phi$ when we invoke Theorem \ref{theorem:balanced_cut_mu} on the different sub-clusters, to ensure that the final sub-clusters $S_i\subseteq S$, $\Omega\!\left(\frac{1}{f(S_i) \cdot\log n}\right)$-expand with respect to 
    their boundary edges.

    We can imagine the sequence of cuts and recursions as a binary tree whose root is $S$, and where in each step, the children of the current node are the two sides of the cut. Because the cut is balanced (by Theorem~\ref{theorem:balanced_cut_mu}),
    we can ensure that the depth of the tree is poly-logarithmic, so the total running time is $\tilde{O}(m)$.

    As for the second result of Theorem \ref{theorem:refinement_phase}, we would like to use the flows that are promised by Theorem~\ref{theorem:trimming_allows_routing} to route demands from the cut edges to the boundary edges. Specifically, in each iteration we route demands from the cut edges $E(A, D\setminus A)$, where $D\subseteq S$ denotes the current sub-cluster of $S$, to the boundary edges of $A$ or of $D\setminus A$. 
    In this manner, we could route the flow from the inter-cluster edges $F$ to the boundary edges $B_S$ by going over the binary tree from the leaves upwards: each time we encounter a cut $(A, D\setminus A)$  
    of the current sub-cluster $D\subseteq S$, we route the mass currently residing on the cut edges to the boundary edges of either $A$ or $D\setminus A$. 
    Note that the demands on the boundary edges accumulate as we perform this step. In particular, after accounting for all relevant details, we obtain that routing a unit of flow from each edge of $E(A, D\setminus A)$ to $A$ adds $O\left(\frac{1}{f(D)}\right)$ units\footnote{Recall that $f(D) \defeq c_f\log\!\left(\frac{\sigma}{|D|}\right)\!\cdot\log \log n$.} of load to the boundary edges of~$A$ (and similarly if we route to the boundary of $D\setminus A$).  
    Assuming the load on the boundary edges of $A$ and on the cut edges $E(A, D\setminus A)$ (accumulated in the lower levels of the binary tree) was $\alpha$, this means that when we go upwards in the binary tree, the load on the boundary of $D$ will now be $\alpha' = \alpha \cdot \left(1+\frac{O(1)}{f(D)}\right)$. 
    
    If we could always route to the side with the smaller number of vertices (\ie, route to $A$ if $|A|\le |D\setminus A|$  and to $|D\setminus A|$ otherwise), this approach would work. The reasoning is as follows. Since $f(D)=c_f \cdot\log\left(\frac{\sigma}{|D|}\right)\log\log n$ and $\sigma \ge \frac{3}{2}|S|$, it follows that if $|D|\le |S|\cdot 2^{-j}$ for some $j\ge 1$, then $f(D) = \Omega(j \log\log n)$. 
    This ensures that for each edge $e$, if we consider the sequence of clusters $D_i$ (along the binary tree) in which the load on $e$ increases (\ie, when $e$ is a boundary edge on the side with fewer vertices), the product $\prod_{i}{\left(1+\frac{1}{f(D_i)}\right)}$ behaves like $\prod_{j=1}^{\log n}{\left(1+\frac{1}{j\log\log n}\right)}$, so it remains bounded by a constant. Hence, the total load accumulated on each edge, determined by this product, is constant (see Claim~\ref{claim:product-converges}).

    However, there is a problem with this approach. Specifically, as detailed in Theorem~\ref{theorem:trimming_allows_routing}, due to the algorithms implying Theorem \ref{theorem:balanced_cut_mu}, the returned sparse cut $(A, D\setminus A)$ of a sub-cluster $D\subseteq S$ 
    is created from a sequence of cuts $\{S_t\}_{t = 0}^{T-1}$, 
    where if we define $A_0 = D$, $A_{t+1} = A_t\setminus S_t$ ($t\in \{0,\ldots,T-1\}$), we have that $S_t\subseteq A_t$ is the cut we remove from the leftover cluster at step $t$. 
    All of these cuts are sparse (\ie, $(A_t\setminus S_t, S_t)$ is a sparse cut of $A_t$), and additionally we can route from the cut edges to both sides ($S_t$ and $A_t\setminus S_t$).  Ultimately, when $A = D\setminus \left(\bigcup_t{S_t}\right)$, we can route flow from the cut edges $(A, D\setminus A)$ to the side $D\setminus A$ (see Figure~\ref{fig:cor-flow-routings}$(2a)$). However, if we attempt to route instead towards $A$, we may incur congestion on the same edges multiple times (specifically, $T$ times),\footnote{This happens since $\{S_t\}_{t = 0}^{T-1}$ are disjoint while $\{A_t\setminus S_t\}_{t = 0}^{T-1}$ are not. Therefore the routings from the cuts $(A_t\setminus S_t, S_t)$ to $A_t\setminus S_t$ may overlap and accumulate congestion.} 
    preventing a routing with sufficiently low congestion. We emphasize that this is not an issue in $D\setminus A$, because each cut is removed from the graph after it is processed: we route from $E(A_t\setminus S_t, S_t)$ to $S_t$ \textbf{within} $G[S_t]$, and then $S_t$ is removed from the remaining graph. As a result, the corresponding flow routings remain disjoint (see Figure~\ref{fig:cor-flow-routings}$(2a)$) and the congestion does not accumulate.

    This routing is good for us if $|D\setminus A|\le|A|$, but we have no such guarantee from Theorem \ref{theorem:balanced_cut_mu} (as the theorem only bounds $\mu(D\setminus A)$, which is the number of boundary edges). Instead, we choose to do something slightly more nuanced: we stop the sequence of cuts $\{S_t\}_{t=0}^{T-1}$ early at some $T_0 < T$, as soon as its size $\left|\bigcup_{t=0}^{T_0}{S_t}\right|$ is a constant fraction of $|D|$. In other words, we find the first $T_0$ such that $\left|\bigcup_{t=0}^{T_0}{S_t}\right| = \Theta(|D|)$. Denote $A = D\setminus \left(\bigcup_{t=0}^{T_0}{S_t}\right)$. Assuming that at this step, both sides of the cut are of size $\Theta(|D|)$, this is enough to ensure that the recursion depth is poly-logarithmic (even though $\mu(D\setminus A) = \mu\left(\bigcup_{t=0}^{T_0}{S_t}\right)$ may be arbitrarily small), 
    and allows us to bound $|D\setminus A| = \left|\bigcup_{t=0}^{T_0}{S_t}\right| = \Theta(|D|)$ (see Figure~\ref{fig:cor-flow-routings}$(2b)$).

    The final obstacle with this approach is that it may be that while $\left|\bigcup_{t=0}^{T_0-1}{S_t}\right|$ is small (say $o(|D|)$), it's possible $\left|\bigcup_{t=0}^{T_0}{S_t}\right|$ is already too large (say $(1-o(1))\cdot|D|$), so $\left|\bigcup_{t=0}^{T_1}{S_t}\right|$ is not $\Theta(|D|)$ for any $T_1$. In this case we perform two steps: we denote $A_1 \defeq D \setminus \left(\bigcup_{t=0}^{T_0-1}{S_t}\right)$, which is a cut of $D$, and then further cut $A_1$ into $(S_{T_0}, A_2 \defeq A_1 \setminus S_{T_0})$. Routing from $E(A_1, D\setminus A_1)$ to boundary of the side $D\setminus A_1$ is done like before (as mentioned, routing to the side $D\setminus A_1$ does not incur multiple congestion on the same edges, because in $D\setminus A_1 = \bigcup_{t=0}^{T_0-1}{S_t}$, the cuts are removed at each step). Routing from $E( D\setminus A_2, A_2)$ to the boundary of $A_2$ is possible because $S_{T_0}$ now represents a single internal cut rather than a sequence of cuts. Consequently, we do not incur the congestion multiple times (see Figure~\ref{fig:cor-flow-routings}$(2c)$). 

    This provides a high-level overview of the partitioning algorithm when Theorem \ref{theorem:balanced_cut_mu} terminated with Case~(2). A similar, though slightly simpler, argument applies when the theorem terminates with Case~(3).
    
    In total, this partitioning scheme guarantees that we always route from the cut edges to the boundary edges of a side whose size is at most $c\cdot |D|$. This suffices to show that the total accumulated load on each boundary edge remains bounded by a constant.

    The full details are provided in Section~\ref{section:refinement_phase_proof}. It begins with Corollary~\ref{cor:balanced-cut-refined}, a refined version of Theorem~\ref{theorem:balanced_cut_mu}, which bounds the number of vertices on the side of the cut to which flow from the cut edges can be routed. Using this corollary, we then present the details of the aforementioned binary tree and prove Theorem~\ref{theorem:refinement_phase}.

    \subsubsection{Proof of Theorem \texorpdfstring{\ref{theorem:refinement_phase}}{6.2}}
    \label{section:refinement_phase_proof}
    We are going to apply Theorem~\ref{theorem:balanced_cut_mu} on the cluster $S\subseteq V$. However, as discussed in Section~\ref{section:refinement_phase_overview}, 
    Theorem~\ref{theorem:balanced_cut_mu} accounts only for the vertex measure~$\muG$. In our setting, we also need to control the number of vertices in the relevant sets, and therefore require a refined variant that incorporates both $\muG$ and cardinality. The following corollary of Theorems~\ref{theorem:balanced_cut_mu} and~\ref{theorem:trimming_allows_routing} 
    provides this. For generality, we state it with an arbitrary measure $\nu:V\rightarrow \RR_{\ge 0}$, even though in our application $\nu$ will always be the counting measure $\nu(A) = |A|$.\footnote{As will be described below, we use Corollary \ref{cor:balanced-cut-refined} on the subdivision graph and $\nu(A)$ will actually be the number of ``non-split vertices'' in $A$, \ie, $\nu(A)=|A\cap V|$.} 
    Recall the notion of routing from a cut to a set from Section~\ref{subsec:cut-matching}. See Figure~\ref{fig:cor-flow-routings} for an illustration of the cases of the corollary.

    \begin{figure}
        \centering
        \includegraphics[width=0.9\linewidth]{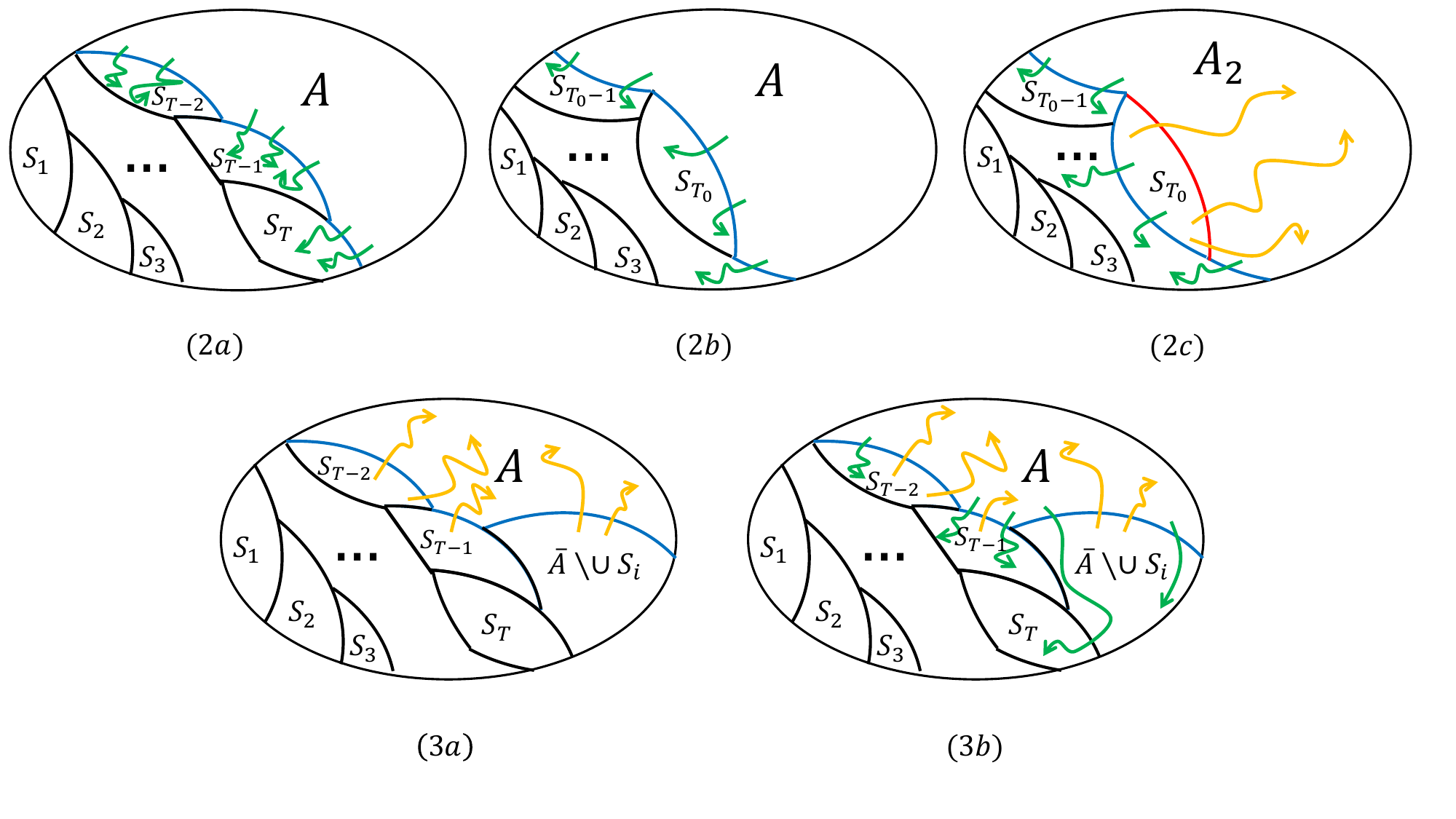}
        \caption{The 5 cases of Corollary~\ref{cor:balanced-cut-refined}. In all cases except case $(2c)$, the left side of the blue cut corresponds to $V\setminus A$, where in case $(2c)$ it corresponds to $V\setminus A_1$. The green flow depicts the guaranteed flow from Corollary~\ref{cor:balanced-cut-refined}. In cases $(2c)$ and $(3b)$ the corollary guarantees an additional flow, depicted in gold. The red cut in case $(2c)$ corresponds to the cut $(A_1\setminus A_2, A_2)$. 
        }
        \label{fig:cor-flow-routings}
    \end{figure}

    \begin{corollary}
    \label{cor:balanced-cut-refined}
        Given a graph $G=(V,E)$ of $m$ edges, a parameter $\phi > 0$, and two weight functions $\mu,\nu : V \to \RR_{\ge 0}$ on the vertices, there exists a randomized algorithm which takes $\tilde{O}(m)$ time and must end in one of the following cases:
        \begin{enumerate}
            \item We certify that $G$ has $\mu$-expansion $\Phi^{\mu}(G)=\Omega(\phi)$, with high probability. 
            \item[2a.] We find a cut $(A,V\setminus A)$ of $V$ with $\mu$-expansion $\Phi^{\mu}_G(A,V\setminus A)=O(\phi\log n)$, $\mu(V\setminus A)=\Omega\left(\frac{\mu(V)}{\log n}\right)$, and $\nu(V\setminus A)\le\frac{1}{4}\nu(V)$. Additionally, there exists a flow routable with congestion $O(1)$ in $G[V\setminus A]$, such that each cut edge of $E(A,V\setminus A)$ is a source of one unit of flow\footnote{Formally, there is a flow in $G[V\setminus A]$ routable with congestion $O(1)$, such that each $v\in V\setminus A$ sends $\deg_G(v)-\deg_{G[V\setminus A]}(v)$ units of flow and receives at most $O(\phi\log n\cdot\mu(v))$ units. The rest of the flows in this corollary are defined similarly.} and each $v\in V\setminus A$ receives at most $O(\phi\log n\cdot\mu(v))$ units. 
            \item[2b.] We find a cut $(A,V\setminus A)$ of $V$ with $\mu$-expansion $\Phi^{\mu}_G(A,V\setminus A)=O(\phi\log n)$, and $\frac{1}{4}\nu(V)\le\nu(V\setminus A)\le\frac{1}{2}\nu(V)$. 
            Additionally, there exists a flow routable with congestion $O(1)$ in $G[V\setminus A]$, such that each cut edge of $E(A,V\setminus A)$ is a source of one unit of flow and each $v\in V\setminus A$ receives at most $O(\phi\log n\cdot\mu(v))$ units.
            \item[2c.] We find a cut $(A_1,V\setminus A_1)$ of $V$ and a cut $(A_2, A_1\setminus A_2)$ of $A_1$, with $\mu$-expansion $\Phi^{\mu}_G(A_1,V\setminus A_1), \Phi^{\mu}_G(A_2,A_1\setminus A_2)=O(\phi\log n)$. Additionally, $\mu(A_2) \ge\frac{1}{4}\mu(V)$ and $\nu(A_1\setminus A_2)\ge \frac{1}{4}\nu(V)$.\footnote{In particular, the sets $V\setminus A_1, A_2,$ and $A_1\setminus A_2$ are all small with respect to either $\mu$ or $\nu$.} 
            There exists a flow routable with congestion $O(1)$ in $G[V\setminus A_1]$, such that each cut edge of $E(A_1,V\setminus A_1)$ is a source of one unit of flow and each $v\in V\setminus A_1$ receives at most $O(\phi\log n\cdot\mu(v))$ units. There also exists a flow routable with congestion $O(1)$ in $G[A_2]$, such that each cut edge of $E(A_2,A_1\setminus A_2)$ is a source of one unit of flow and each  $v\in A_2$ receives at most $O(\phi\log n\cdot\mu(v))$ units.
            \item[3a.] We find a cut $(A,V\setminus A)$ with $\mu$-expansion $\Phi^{\mu}_G(A,V\setminus A)=O(\phi\log n)$, $0<\mu(V\setminus A)\le \frac{\mu(V)}{2}$, and $\nu(V\setminus A)\ge\frac{1}{2}\nu(V)$. Additionally, there exists a flow routable with congestion $O(1)$ in $G[A]$, such that each cut edge of $E(A, V\setminus A)$ is a source of one unit of flow and each $v\in A$ receives at most $O(\phi\cdot\mu(v))$ units.
            \item[3b.] We find a cut $(A,V\setminus A)$ with $\mu$-expansion $\Phi^{\mu}_G(A,V\setminus A)=O(\phi\log n)$, with high probability, $G[A]$ is an $\Omega(\phi)$-expander with respect to $\mu$, and  $\nu(V\setminus A)\le\frac{1}{2}\nu(V)$. Additionally, there exists a flow routable with congestion $O(1)$ in $G[V\setminus A]$, such that each cut edge of $E(A, V\setminus A)$ is a source of one unit of flow and each $v\in V\setminus A$ receives at most $O(\phi\log n\cdot\mu(v))$ units. Furthermore, there exists a flow routable with congestion $O(1)$ in $G[A]$, such that each cut edge of $E(A, V\setminus A)$ is a source of one unit of flow and each $v\in A$ receives at most $O(\phi\cdot\mu(v))$ units.
        \end{enumerate}
    \end{corollary}

    \begin{proof}
        We apply Theorem~\ref{theorem:balanced_cut_mu} on $G$. If it terminated with Case~(1), we are done, since this exactly matches Case~(1) of the statement. 
        Otherwise, first assume it terminated with Case (2). Recall from Theorem~\ref{theorem:trimming_allows_routing} that $\bar{A}$ is given as $\bar{A}=\bigcup_{t=0}^{T-1}{S_t}$, where if we look at the sequence of vertex sets $A_0=V, A_{t+1}=A_t\setminus S_t$, each $(A_{t+1}, S_t)$ is a sparse cut. We split into cases according to $\nu(\bar{A})$. If $\nu(\bar{A})\le\frac{1}{4}\nu(V)$, we terminate with case (2a) (see Figure~\ref{fig:cor-flow-routings}). 
        
        So assume that $\nu(\bar{A})>\frac{1}{4}\nu(V)$. We consider the first $0\le T_0 < T$ such that $\nu\left(\bigcup_{t=0}^{T_0}{S_t}\right)>\frac{1}{4}\nu(V)$ (\ie, $\nu\left(\bigcup_{t=0}^{T_0-1}{S_t}\right)\le\frac{1}{4}\nu(V)$). If $\nu\left(\bigcup_{t=0}^{T_0}{S_t}\right)\le\frac{1}{2}\nu(V)$ we terminate with case (2b), returning the cut $\left(V\setminus \bigcup_{t=0}^{T_0}{S_t}, \bigcup_{t=0}^{T_0}{S_t}\right)$. 

        Next, assume that $\nu\left(\bigcup_{t=0}^{T_0}{S_t}\right)>\frac{1}{2}\nu(V)$. This means that $\nu\left(S_{T_0}\right)>\frac{1}{4}\nu(V)$. We set $A_1 = V\setminus(\bigcup_{t=0}^{T_0-1}{S_t})$ and $A_2 = A_1\setminus S_{T_0}$ and terminate in case (2c). 
        The guarantee $\mu(A_2)\ge\frac{1}{4}\mu(V)$ is given at the beginning of Theorem~\ref{theorem:trimming_allows_routing}. 

        This concludes case~(2) of Theorem~\ref{theorem:balanced_cut_mu}. 
        
        Now assume that Theorem~\ref{theorem:balanced_cut_mu} terminated with Case~(3). This means that with high probability $G[A]$ is an $\Omega(\phi)$-expander with respect to $\mu$. Additionally, the cut $(A, V\setminus A)$ is a result of the trimming step (see Theorem~\ref{theorem:trimming_allows_routing}) and satisfies the properties mentioned in Theorem~\ref{theorem:trimming_allows_routing}. In particular, we have $0 < \mu(V\setminus A) \le \frac{\mu(V)}{2}$.
        We split into cases according to $\nu(V\setminus A)$. If $\nu(V\setminus A)\ge\frac{1}{2}\nu(V)$, we terminate with case (3a). 
        Otherwise, we terminate with case (3b). 

        In all cases, the required flows are given directly by Theorem~\ref{theorem:trimming_allows_routing}.
    \end{proof}

    We now proceed to prove Theorem~\ref{theorem:refinement_phase}. The proof follows the overview presented in Section \ref{section:refinement_phase_overview}.

    \begin{proof}[Proof of Theorem \ref{theorem:refinement_phase}]
        We are given a cluster $S \subseteq V$ and a parameter $\sigma \ge \frac{3}{2}|S|$. 
        In the following, we apply Corollary~\ref{cor:balanced-cut-refined} with respect to a weighting $\mu$ supported only on the boundary split nodes of the current cluster. To this end, rather than applying Corollary~\ref{cor:balanced-cut-refined} on the subdivision graph $G'[S']$, we define an auxiliary graph $G_{\tilde{S}}$ that contains split nodes only for boundary edges, without the internal split nodes. This avoid technical edge-cases concerning the placement of split nodes on the ``correct" side of the cut.
        Formally, we denote the boundary edges of $S$ by $B_S \defeq E(S, V\setminus S)$, and the edges in $G[S]$ by $E_S \defeq \{e\in E\mid e\subseteq S\}$. We define $X_{B_S} \defeq \{x_b\mid b\in B_S\}$, and $\tilde{S} = S\cup X_{B_S}$. Then, the graph $G_{\tilde{S}}$ is given by $G_{\tilde{S}} \defeq (\tilde{S}, E_{\tilde{S}})$, where $E_{\tilde{S}} \defeq E_S\cup \{(v, x_b) \mid v\in S, b\in B_S, v\in b\}$.
        We define the weighting $\mu_S$ on $G_{\tilde{S}}$ as $\mu_S = \1_{X_{B_S}}$. 
        We always use $\nu = \1_{V}$, \ie, $\nu$ counts the number of non-split vertices.
        
        We then apply Corollary~\ref{cor:balanced-cut-refined} on $G_{\tilde{S}}$ with $\phi_S = \Theta\left(\frac{1}{\log n\cdot f(S)}\right)$.\footnote{The reason we prefer to apply Corollary \ref{cor:balanced-cut-refined} on $G_{\tilde{S}}$ instead of $G'[S']$ is that applying it on the latter may create some \emph{edge-cases} which complicate the argument. For example, it may turn out that an edge $e=(v, u)\in E_S$ has both $v$ and $u$ on the same side of the cut while $x_e$ is on the other.}
        The exact constant $c_\phi$ in $\phi_S = \frac{c_\phi}{\log n\cdot f(S)}$ will be  $c_\phi = \min \{ 1, \frac{1}{16c_\star}\}$, where $c_\star$ is the constant governing the sparsity bounds in Corollary~\ref{cor:balanced-cut-refined} (see below for details).
        We split into cases according to the result of the corollary (see Figure~\ref{fig:cor-flow-routings}). If Corollary~\ref{cor:balanced-cut-refined} terminates with Case (1), we terminate.
    
        Otherwise, Corollary~\ref{cor:balanced-cut-refined} returns a cut $(\tilde{A}, \tilde{S}\setminus \tilde{A})$ which satisfies $\Phi_{G_{\tilde{S}}}^{\mu_S}(\tilde{A}, \tilde{S}\setminus \tilde{A}) \le c_\star \cdot \phi_S\log n$ (for $c_\star \ge 0$), or (in case (2c)) two cuts $(\tilde{A}_1, \tilde{S}\setminus \tilde{A}_1), (\tilde{A}_2, \tilde{A}_1\setminus \tilde{A}_2)$ which satisfy this sparsity condition. Denote $A \defeq \tilde{A}\cap S$ (or $A_i \defeq \tilde{A}_i \cap S$). Unless Corollary~\ref{cor:balanced-cut-refined} terminated with case (3b), the refinement phase recurs on all sub-clusters: $A$ and $S \setminus A$ (in case (2c), $S\setminus A_1, A_1\setminus A_2, A_2$). In case (3b), we only recur on $S\setminus A$.\footnote{In the following, we assume that the cuts are \emph{well-behaved}, in the sense that if $(v, x_b)$ is an edge in $G_{\tilde{S}}$ with $x_b\in X_{B_S}$, then $v$ and $x_b$ are in the same side of the cut that Corollary \ref{cor:balanced-cut-refined} returns (\ie, $x_b$ is not an isolated vertex in the other side of the cut).
        If this does not hold, we may simply move $x_b$ to the side where $v$ resides. The presented arguments still hold, sometimes with slightly different constants.
        }
        Note that in the recursive step on a sub-cluster $R\subseteq  S$, we define $\mu_R$ on $G_{\tilde{R}}$ according to the boundary edges of this cluster, \ie, $E(R, V\setminus R)$, which includes some of the previous boundary edges and newly cut edges. Additionally, we proceed with $\phi_R = \Theta\left(\frac{1}{\log n \cdot f(R)}\right)$.

        Note that we do not recur into sub-clusters $S$ with $\mu(S) = 0$. For such sub-clusters, we simply return the partition $\{S\}$, as it trivially satisfies both requirement of Theorem \ref{theorem:refinement_phase}.
    
        The algorithm is summarized in Algorithm \ref{algo:refinement_phase}.

        \begin{algorithm}[hbt!]
        \caption{Refinement Phase}
        \label{algo:refinement_phase}
        \begin{algorithmic}[1]
            \Function{Refine}{$G$, $S$, $\sigma$}
                \State $f(S) \leftarrow c_f \log \left( \frac{\sigma}{|S|} \right)\cdot \log \log n$.
                \State Call Corollary~\ref{cor:balanced-cut-refined} on $G_{\tilde{S}}$ with $\phi_S = \Theta\left(\frac{1}{\log n\cdot f(S)}\right)$ and $\mu_S = \1_{X_{B_S}}$.
                \If {we get Case (1)}
                    \Return $\{S\}$.
                \ElsIf {we get Cases $(2a),(2b)$ or $(3a)$}  
                    \State Let $(\tilde{A}, \tilde{S}\setminus \tilde{A})$ be the returned cut.
                    \State $A \leftarrow \tilde{A} \cap S$.
                    \State \Return $\mathrm{Refine}(G,A,\sigma) \cup \mathrm{Refine}(G,S\setminus A,\sigma)$.
                \ElsIf {we get Case $(2c)$}  
                    \State Let $(\tilde{A}_1, \tilde{S}\setminus \tilde{A}_1)$ and $(\tilde{A}_2, \tilde{A}_1\setminus \tilde{A}_2)$ be the returned cuts.
                    \State $A_1 \leftarrow \tilde{A}_1 \cap S$, $A_2 \leftarrow \tilde{A}_2 \cap S$.
                    \State \Return $\mathrm{Refine}(G,S\setminus A_1,\sigma) \cup \mathrm{Refine}(G,A_1\setminus A_2,\sigma) \cup \mathrm{Refine}(G,A_2,\sigma)$.
                \ElsIf {we get Case $(3b)$}  
                    \State Let $(\tilde{A}, \tilde{S}\setminus \tilde{A})$ be the returned cut.
                    \State $A \leftarrow \tilde{A} \cap S$.
                    \State \Return $ \{A\}\cup \mathrm{Refine}(G,S\setminus A,\sigma)$.
                \EndIf
            \EndFunction
        \end{algorithmic}
    \end{algorithm}
    
        First, to show that the running time is $\tilde{O}(m)$, we note that whenever we recur on a sub-cluster $R\subseteq S$, it holds that either:
        \begin{itemize}
            \item $|R|\le \frac{3|S|}{4}$ (when we have $\nu(\tilde{R})\le\frac{3\nu(\tilde{S})}{4}$), or 
            \item $\mu_S(\tilde{S}\setminus \tilde{R}) = \Omega\left(\frac{\mu_S(\tilde{S})}{\log n}\right)$ and $\Phi^{\mu_S}_{G_{\tilde{S}}}(\tilde{R}, \tilde{S}\setminus \tilde{R}) = O(\phi_S\log n)$. We show that in this case $\mu_R(\tilde{R})=O(\mu_S(\tilde{S})\cdot(1-\frac{1}{\log n}))$. It holds that the new weighting $\mu_R$ satisfies $\mu_R(\tilde{R}) \le \mu_S(\tilde{R}) + \capacity(\tilde{R}, \tilde{S}\setminus \tilde{R})$ (see Figure~\ref{fig:refinment_change_in_boundary}).
            We have that $\Phi_{G_{\tilde{S}}}^{\mu_S}(\tilde{R}, \tilde{S}\setminus \tilde{R}) \le c_\star \cdot \phi_S \log n$. By the choice of the constant~$c_\phi\le\frac{1}{16c_\star}$ in $\phi_S = \frac{c_\phi}{\log n\cdot f(S)}$, we get that $\Phi_{G_{\tilde{S}}}^{\mu_S}(\tilde{R}, \tilde{S}\setminus \tilde{R}) \le \frac{1}{2 f(S)} \le \frac{1}{2}$, which in particular means that $\capacity_{G_{\tilde{S}}}(\tilde{R}, \tilde{S}\setminus \tilde{R}) \le \frac{1}{2}\mu_S(\tilde{S}\setminus \tilde{R})$.
            Thus, we have $\mu_R(\tilde{R})\le \mu_S(\tilde{R})+\frac{1}{2}\mu_S(\tilde{S}\setminus \tilde{R}) = \mu_S(\tilde{S})-\frac{1}{2}\mu_S(\tilde{S}\setminus \tilde{R}) = O(\mu_S(\tilde{S})\cdot(1-\frac{1}{\log n}))$. 
            \item We have case (2c), for the sub-cluster $R = A_1\setminus A_2$. In this case it may be that the cut $(\tilde{R}, \tilde{S}\setminus \tilde{R})$ is not sparse, if there are many edges in $E(\tilde{A_1}\setminus \tilde{A_2}, \tilde{S}\setminus \tilde{A_1})$. However, note that it does hold that $\mu_S(\tilde{R}) \le \frac{3}{4}\mu_S(\tilde{S})$ and additionally $\Phi^{\mu_S}_{G_{\tilde{S}}}(\tilde{S}\setminus \tilde{A_1}, \tilde{A_1}), \Phi^{\mu_S}_{G_{\tilde{S}}}(\tilde{A_1}\setminus \tilde{A_2}, \tilde{A_2})$ are both at most $c_\star\cdot\phi_S\log n$. By the choice of the constant $c_\phi\le\frac{1}{16c_\star}$ in $\phi_S = \frac{c_\phi}{\log n\cdot f(S)}$, both sparsities are at most $ \frac{1}{16 f(S)} \le \frac{1}{16}$, giving that $\capacity(\tilde{R}, \tilde{S}\setminus \tilde{R}) \le \capacity(\tilde{S}\setminus \tilde{A_1}, \tilde{A_1}) + \capacity(\tilde{A_1}\setminus \tilde{A_2}, \tilde{A_2}) \le 2\cdot \frac{1}{16}\mu_S(\tilde{S}) = \frac{1}{8}\mu_S(\tilde{S})$. Therefore, $\mu_R(\tilde{R})\le \mu_S(\tilde{R})+\capacity(\tilde{R}, \tilde{S}\setminus \tilde{R}) \le \frac{3}{4}\mu_S(\tilde{S}) + \frac{1}{8}\mu_S(\tilde{S})=\frac{7}{8}\mu_S(\tilde{S})$.
        \end{itemize}
        
        Therefore, the recursion depth is poly-logarithmic, so the running time is $\tilde{O}(m)$.

        \begin{figure}
            \centering
            \includegraphics[width=0.7\linewidth]{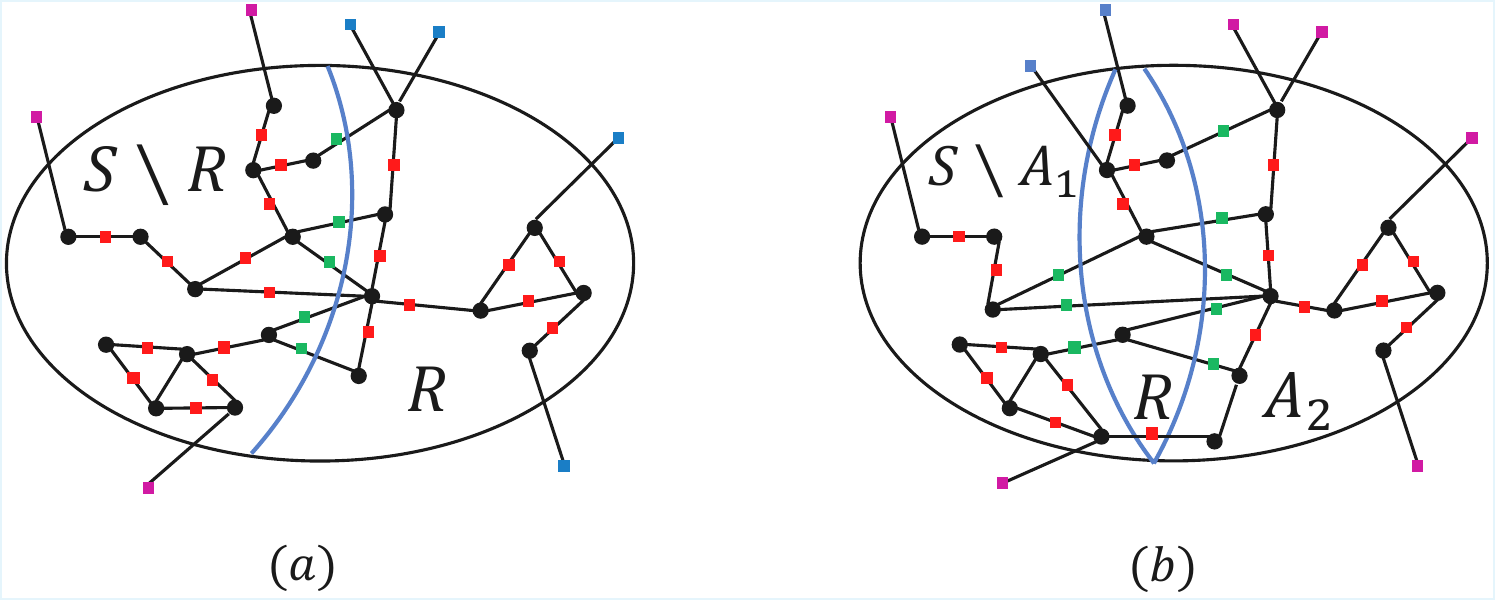}
            \caption{Visualizing the quantities $\mu_R(R'), \mu_S(R'),$ and $ \capacity(R', S'\setminus R')$. The quantity $\mu_R(R')$ corresponds to the blue and the green split nodes, while $\mu_S(R')$ corresponds only to the blue split nodes. Finally $\capacity_{G'[S']}(R', S'\setminus R') = \capacity_{G[S]}(R, S\setminus R)$ corresponds to the green nodes. Figure $(b)$ corresponds to case $(2c)$ of Corollary~\ref{cor:balanced-cut-refined}, and Figure $(a)$ corresponds to all other cases.
            }
            \label{fig:refinment_change_in_boundary}
        \end{figure}

        \begin{figure}
            \centering
            \includegraphics[width=0.9\linewidth]{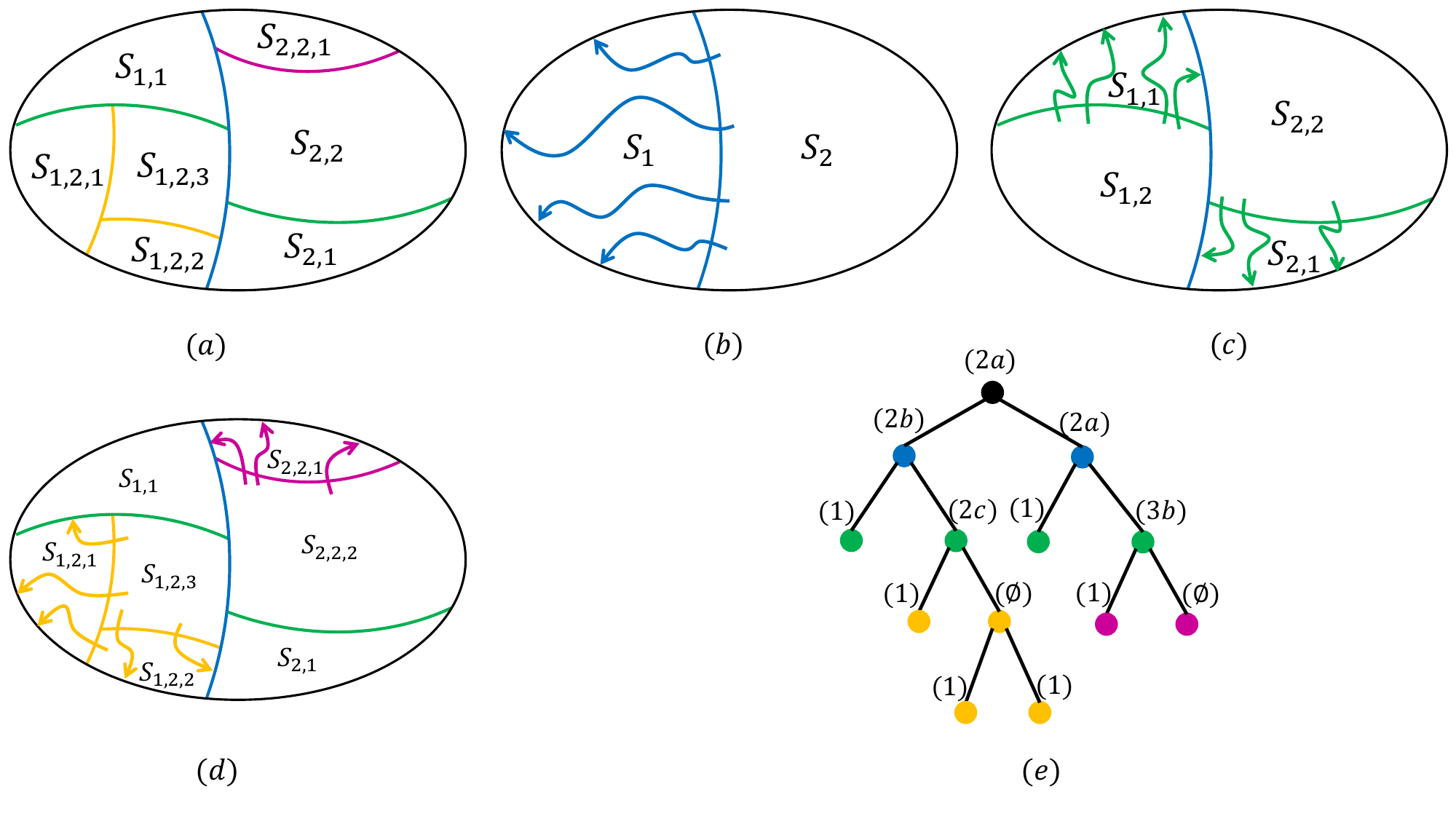}
            \caption{Figure~$(a)$ depicts several cuts computed by the refinement phase of a cluster $S$. Figures $(b)-(d)$ illustrate the flow routings from the cuts to the boundary of the smaller side, as guaranteed by Corollary~\ref{cor:balanced-cut-refined}. Figure~$(e)$ shows the binary tree that corresponds to these cuts: each node represents a sub-cluster obtained during the cutting process. The label on each node indicates the case of Corollary~\ref{cor:balanced-cut-refined} that was applied. A label of $(\emptyset)$ means that Corollary~\ref{cor:balanced-cut-refined} was not applied, and the node's state is instead inferred from its parent. The left child always corresponds to a side of the cut which is smaller (in terms of number of vertices) than the parent, by a constant factor.}
            \label{fig:binary-tree-refinement-phase}
        \end{figure}

    \smallskip
        We now show that the first claim of Theorem~\ref{theorem:refinement_phase} holds when we stop recurring on a sub-cluster. Assume that when running in some sub-cluster $R\subseteq S$, Corollary~\ref{cor:balanced-cut-refined} terminated with Case (1). In this case, $G_{\tilde{R}}$ is $\Omega(\phi_R)$-expander with respect to $\mu_R$, which means that $X_{B_R}$  $\Omega(\phi_R)$-expands in $G_{\tilde{R}}$. Therefore, by Remark~\ref{remark:mu_expansion_as_problem},  $G_{\tilde{R}}$ $\Omega(\phi_R)$-respects the all-to-all flow problem on the vertices of $X_{B_R}$, which establishes the first result.
        
        Now, assume that Corollary~\ref{cor:balanced-cut-refined} terminated with case (3b). This means that we have $A\subseteq R$ on which we do not recur, which satisfies that $G_{\tilde{A}}$ is an $\Omega(\phi_R)$-expander with respect to $\mu_R$.
        Equivalently, $\tilde{A}\cap X_{B_R}$ $\Omega(\phi_R)$-expands in $G_{\tilde{A}}$.
        Denote by $C = E(A, R\setminus A)$ the set of cut edges that were cut in this step. We have that we can route a flow in $G_{\tilde{A}}$ with congestion $O(1)$, such that each edge $e\in C$ sends one unit of flow,
        and each $v\in \tilde{A}$ gets at most $O(\phi_R \cdot \mu_R(v)) = O(\mu_R(v))$ units.\footnote{Recall the definition of routing a flow from the cut edges, introduced in the discussion following Theorem~\ref{theorem:balanced_cut_mu}.} 
        This means that only boundary nodes in $\tilde{A}\cap X_{B_R}$ can receive flow, and each receives at most $O(1)$ units. Let $X_C$ denote the new boundary edges in $G_{\tilde{A}}$, that correspond to the cut edges of $C$.
        Using Lemma~\ref{lemma:separator-general}, we get that $(\tilde{A}\cap X_{B_R})\cup X_C = X_{B_A}$ $\Omega(\phi_R)$-expands in $G_{\tilde{A}}$. 
        This means (Remark~\ref{remark:mu_expansion_as_problem}) that $G_{\tilde{A}}$ $\Omega(\phi_R)$-respects the all-to-all flow problem on the vertices of $X_{B_A}$.
        Finally, since $|A|\le |R|$, we have $\phi_R = \Omega(\phi_A)$. That is, $G_{\tilde{A}}$, $\Omega(\phi_A)$-respects the all-to-all flow problem on the vertices of $X_{B_A}$, as desired.
        
\smallskip
        Next, we show the second claim of Theorem~\ref{theorem:refinement_phase}. The argument follows the ideas of~\cite{harrelson2003polynomial,RackeS14}. Note that inter-cluster edges are only created as cut edges in cases (2a), (2b), (2c), (3a) and (3b) of Corollary~\ref{cor:balanced-cut-refined}.

        We view the cuts formed by cutting $S$ as a binary tree.\footnote{This tree is for analysis purposes and should not be confused with the outer hierarchical decomposition.} We call it the refinement tree, $T_S^{\mathrm{ref}}$, of $S$ (see Figure~\ref{fig:binary-tree-refinement-phase}). The root of the binary tree corresponds to $S$ itself and every vertex corresponds to a subset $D\subseteq S$. The children of $D$ in the binary tree are the subsets we recur on during the application of Corollary~\ref{cor:balanced-cut-refined} on $G_{\tilde{D}}$. Note that in case (2c) we have that $D$ splits into $A_1, D\setminus A_1$, and $A_1$ further splits into $A_2, A_1\setminus A_1$. These correspond to two levels of $T_S^{\mathrm{ref}}$. 
        We choose to set the left child of each node to be the side in which we can route flow from the cut edges: 
        \begin{enumerate}
            \item In cases (2a), (2b) and (3b), $D\setminus A$ is the left child of $D$. 
            \item In case (2c), $D\setminus A_1$ is the left child of $D$ and $A_2$ is the left child of $A_1$. See Figure~\ref{fig:binary-tree-refinement-phase}(e).
            \item In case (3a), $A$ is the left child of $D$.
        \end{enumerate}
        
        Note that in all cases the left child $L\subseteq D$ satisfies $|L|\le \frac{3}{4}|D|$, and additionally, we can route in $G_{\tilde{L}}$ a flow with congestion $O(1)$ such that each cut edge of $E(L, D\setminus L)$ sends one unit of flow and each vertex $v\in \tilde{L}$ receives at most $O(\phi_D\log n\cdot \mu_D(v))$ units.\footnote{In case (2c), this is true for $D$ directly, while for $A_1$ it's true with $\phi_D$ and $\mu_D$ instead of $\phi_{A_1}$ and $\mu_{A_1}$. However, because $|A_1| = \Theta(|D|)$, we have $\phi_{A_1} = \Theta(\phi_D)$. Additionally, we have $\mu_D \le \mu_{A_1}$.} That is, each boundary split node of $X_{B_D}\cap \tilde{L}$ gets at most $O(\phi_D\log n)$ units. Let $c_0\ge 1$ be such that in all cases each boundary split node in $X_{B_D}\cap \tilde{L}$ gets at most $c_0\phi_D\log n$ units. In the definition of $f$ in Theorem~\ref{theorem:refinement_phase}, we  choose $c_f \defeq \frac{4c_0}{\log(\frac{4}{3})}\ge1$.
    
        To describe the routing in $G'[S']$, where each inter-cluster split node sends one unit of flow to the boundary of $S$, we proceed as follows.\footnote{Observe that, unlike in other parts of the paper where we reason about demand states, here we explicitly construct and route a flow.}
        We shift our focus from $G_{\tilde{S}}$ back to $G'[S']$, effectively adding a split node in the middle of each inner edge. We identify flows from cut edges (in $G_{\tilde{S}}$) with flows that route mass from the split nodes (in $G'[S']$) corresponding to the cuts' edges.
        We start with one unit of mass placed on each inter-cluster split node.
        The mass is then routed according to the binary tree: beginning from the leaves and proceeding upward iteratively. At each binary cut $(L, D\setminus L)$, we use the flows described in Corollary~\ref{cor:balanced-cut-refined} to route the mass currently residing on the cut edges (\ie, split nodes of cut edges) to the boundary split nodes of $L$, \ie, to $X_{B_D}\cap \tilde{L}$. We scale the flow according to the current load on the cut edges. 
        Next, we analyze the load on each boundary split node as a result of these routings.
    
        We define the \emph{left depth} of a node $D\subseteq S$ (we identify the node with the corresponding subset of $S$) in the binary tree to be the number of left branches on the path from the root of the tree (corresponding to $S$) to the node $D$, plus one. That is, the left depth of $S$ is $1$. We denote by $\Phi(j)$ the maximum load on a boundary edge of a cluster of left depth $j$, after we process this sub-cluster (\ie, after we route from the cut edges $E(L,D\setminus L)$ induced at this cluster to its boundary edges $E(L, V\setminus D)\subseteq E(D, V\setminus D)$).
    
        Note that the left depth of a node is at most $\log_{4/3}{n}\le3\log n$.
        Therefore, any node of depth $3\log n$ must be a leaf node, so $\Phi(3\log n) = 1$ (no routings are done at the leaf level). 
        Following the analysis of~\cite{harrelson2003polynomial}, we now prove by induction on $j=3\log n, \ldots, 1$ the following claim:
        \begin{claim} [See \texorpdfstring{\cite[Lemma 4]{harrelson2003polynomial}}{[HHR03, Lemma 4]}]
        \label{claim:hhr-binary-tree}
            \[
                \Phi(j)\le \prod_{l=j}^{3\log n}{\left(1+\frac{1}{l\log \log n}\right)}
            \]    
        \end{claim}
        \begin{proof}
            Assume the claim is true for all $l\ge j+1$. Consider a node $D$ at left depth $j$ in the binary tree. Let $L\subseteq D$ be its left child. After processing $D$, the load on each boundary edge of $D\setminus L$ didn't change, while the load on each boundary edge of $L$ may have increased. Indeed, a boundary edge of $L$ has load of at most $\Phi(j+1)$ immediately after processing the sub-cluster $L$. This load then increases when processing the sub-cluster $D$, since we additionally route the load on the cut edges $E(L, D\setminus L)$ to it.
            
            During the processing of $D$, each cut edge in $E(L,D\setminus L)$ (\ie, split node of a cut edge) has load at most $\Phi(j)+\Phi(j+1)\le 2\Phi(j)$, because $\Phi(j+1)$ can come from the processing of $L$ and $\Phi(j)$ from $D\setminus L$ (the inequality is because $\Phi(j)\ge \Phi(j+1)$). 
            
            Therefore, the flow from $X_{E(L, D\setminus L)}$ to the boundary split nodes of $L$ is scaled by at most $2\Phi(j)$. Therefore, each boundary split node of $X_{E(L, V\setminus D)}$ gets at most $2\Phi(j)\cdot c_0\phi_D\log n$ units of flow. Therefore, we have:
            \begin{align*}
                \Phi(j)\le\Phi(j+1)+2c_0\phi_D\log n\cdot\Phi(j)\le\Phi(j+1)+\frac{2c_0\log n}{f(D)\cdot\log n}\Phi(j) = \Phi(j+1)+\frac{2c_0}{f(D)}\Phi(j) .
            \end{align*}
            Rearranging, we get:
            \begin{align*}
                \Phi(j)\le\Phi(j+1)\cdot\left(\frac{1}{1-\frac{2c_0}{f(D)}}\right)\le\Phi(j+1)\cdot\left(1+\frac{4c_0}{f(D)}\right) ,
            \end{align*}
            where the last inequality holds due to the inequality $1/(1-x)\le 1+2x$ for $x\in [0,\frac{1}{2}]$, and indeed $0<\frac{2c_0}{f(D)}\le\frac{1}{2}$ for large enough $n$. Because $D$ has left depth $j$, we must have $|D|\le\left(\frac{3}{4}\right)^{j-1}\cdot|S|$. Because $|S|\le\frac{2}{3}\sigma\le\frac{3}{4}\sigma$, we have $|D|\le\left(\frac{3}{4}\right)^j\cdot\sigma$. Therefore, 
            \begin{align*}
                f(D) &= c_f \log\left(\frac{\sigma}{|D|}\right)\log\log n \ge c_f \cdot j\log\left(\frac{4}{3}\right)\log\log n = \frac{4c_0}{\log\left(\frac{4}{3}\right)} \cdot j\log\left(\frac{4}{3}\right)\log\log n
                \\
                &= 4c_0\cdot j\log\log n .
            \end{align*}
            Plugging in, we have
            \begin{align*}
                \Phi(j) &\le \Phi(j+1)\cdot\left(1+\frac{4c_0}{4c_0\cdot j\log\log n}\right) = \Phi(j+1)\cdot\left(1+\frac{1}{j\log\log n}\right) 
                \\
                &\le \prod_{l=j+1}^{3\log n}{\left(1+\frac{1}{l\log \log n}\right)}\cdot\left(1+\frac{1}{j\log\log n}\right)=\prod_{l=j}^{3\log n}{\left(1+\frac{1}{l\log \log n}\right)} ,
            \end{align*}
            which proves the claim.
        \end{proof}
        We therefore have that the maximum load on an edge (\ie, a split node) in the routing scheme is $\Phi(1)\le\prod_{l=1}^{3\log n}{\left(1+\frac{1}{l\log \log n}\right)}$.
        Using Claim~\ref{claim:product-converges} we have that $\Phi(1)=O(1)$. 
        This means that the load on each edge is always constant, and the congestions of the scaled flows are also constant. An edge is congested by the flows when it is in the left side of the cut (\ie, $G_{\tilde{L}}$), which can happen at most $O(\log n)$ times. Therefore, the total congestion is $O(\log n)$, which completes the proof of Theorem~\ref{theorem:refinement_phase}.
    \end{proof}

	\newpage

	\bibliographystyle{alpha}
	\bibliography{refs}
\appendix

\section{Omitted Proofs}\label{appendix:omitted-proofs}

\subsection{Omitted Proofs from Section \texorpdfstring{\ref{section:techniques}}{2}}

    To prove Claim~\ref{claim:cut-sparsifier-implies-flow}, we use the following flow-cut gap theorems~\cite{leighton1999multicommodity, aumann1998log,okamura1981multicommodity}:
    \begin{theorem}[Flow-Cut Gap~\cite{aumann1998log}]
    \label{theorem:flow-cut-gap}
        Let $Q$ be a demand matrix on $G=(V,E)$. Define $\theta_G(Q)\ :=\ \sup\{\theta  : \text{$G$ $\theta$-respects $Q$}\}.$
        Then, for a universal constant $c>0$,
        \[ \frac{c}{\log n} \congestion_G(Q) \le \frac{1}{\theta_G(Q)} \le  \congestion_G(Q).\]
    \end{theorem}

    \begin{theorem}[Flow-Cut Gap of Trees~\cite{okamura1981multicommodity}]
    \label{theorem:flow-cut-gap-tree}
        Let $Q$ be a demand matrix on a tree $T=(V,E)$. Define $\theta_T(Q)\ :=\ \sup\{\theta  : \text{$T$ $\theta$-respects $Q$}\}.$
        Then,  $\frac{1}{\theta_T(Q)} =  \congestion_T(Q)$.
    \end{theorem}

    \begin{proof} [Proof of Claim~\ref{claim:cut-sparsifier-implies-flow}]  
        We begin by proving the first item. Let $T$ be a tree cut-sparsifier of $G$ of quality $\alpha$. Using Theorem \ref{theorem:flow-cut-gap}, we get that  a tree cut-sparsifier $T$ of $G$ of quality~$\alpha$ satisfies, for all flow problems $Q$,
        \[
            \congestion_T(Q) \numle{1} 
            \congestion_G(Q) \le
            \frac{\log n}{c \cdot \theta_G(Q)} \numle{2}
            \frac{\alpha\log n}{c \cdot \theta_T(Q)} \numeq{3} 
            \frac{\alpha\log n}{c}\congestion_T(Q),
        \]
        where Inequality (2) holds since $T$ is a cut-sparsifier of quality $\alpha$, and therefore 
        \[\frac{\mincut_T(S, \bar{S})}{\dem_Q(S, \bar{S})} \le  \frac{\alpha\cdot \capacity_G(S, \bar{S})}{\dem_Q(S, \bar{S})},\] 
        for every cut $S\subseteq V$ with $\dem_Q(S, \bar{S}) > 0$.
        Equality~(3) holds by Theorem~\ref{theorem:flow-cut-gap-tree}. 
        Finally, Inequality~(1) holds since, for every $e=\{S, P_T(S)\}\in T$, the congestion on $e$ is exactly  $\frac{\capacity_G(S, V\setminus S)}{\dem_Q(S,V\setminus S)}$, which trivially lower bounds $\congestion_G(Q)$.
        Therefore, $T$ is a flow-sparsifier of $G$ of quality $O(\alpha \log n)$.

        We proceed to prove the second item. Let $T$ be a tree flow-sparsifier of $G$ of quality $\alpha$ and fix a cut $(S,\bar{S})$ of $V$.
        We prove that $\capacity_G(S,\bar{S}) \le \mincut_T(S,\bar{S}) \le \alpha \cdot \capacity_G(S,\bar{S})$.
        The first inequality was already established in the proof of Claim~\ref{claim:equiv-def-cut} and relies only on the definition of the capacities in $T$. 
        For the second inequality, we reuse the same demand matrix $R$ constructed in the proof of Claim~\ref{claim:equiv-def-cut}.
        Namely, $R$ is the demand matrix induced by a maximum $(s,t)$-flow in the auxiliary network $T' = T \cup \{s,t\}$, where $s$ is connected to all vertices in $\bar{S}$ and $t$ to all vertices in $S$ with infinite-capacity edges. By construction, $\dem_R(S,\bar S) = \mincut_T(S,\bar S)$, and $\congestion_T(R) = 1$. Since $T$ is a flow-sparsifier of $G$ of quality $\alpha$, we get $\congestion_G(R) \le \alpha$. Therefore 

        \[
        \frac{\mincut_T(S,\bar{S})}{\capacity_G(S,\bar{S})}
        = \frac{\dem_R(S,\bar{S})}{\capacity_G(S,\bar{S})} 
        \le \congestion_G(R) 
         = \alpha.  \]
    \end{proof}

\subsection{Omitted Proofs from Section \texorpdfstring{\ref{section:preliminaries}}{3}}

\begin{proof}[Proof of Fact~\ref{fact:basic_flow_state}]
    Denote $P' = P-P^{\uparrow Q}$. We begin by proving that $P'$ is a valid demand state. Let $k \in \K$. Observe that: 
    \[
        \sum_{u\in V} P'(u,k) = - \sum_{u\in V}\sum_{v\in V}{\frac{P(u,k)}{\norm{P(u)}_1}D_Q(u,v)} + \sum_{u\in V}\sum_{v\in V}{\frac{P(v,k)}{\norm{P(v)}_1}D_Q(v,u)} = 0.
    \]
    We now prove that $\dem_{P'}(U,W)\le\dem_Q(U,W)$, for every cut $(U,W)$. Fix a cut $(U,W)$, we get that:
    \begin{align*}
        \dem_{P'}(U,W) &=
        \sum_{k \in \K}\left|\sum_{u\in U} P'(u,k)\right| \\&= 
        \sum_{k \in \K}\left| \sum_{u\in U}\sum_{v\in V}{\frac{P(v,k)}{\norm{P(v)}_1}D_Q(v,u)}  - \sum_{u\in U}\sum_{v\in V}{\frac{P(u,k)}{\norm{P(u)}_1}D_Q(u,v)}  \right| \\&=
        \sum_{k \in \K}
        \left|
        \sum_{u\in U}\sum_{v\in U}{\frac{P(v,k)}{\norm{P(v)}_1}D_Q(v,u)}
        + \sum_{u\in U}\sum_{v\in W}{\frac{P(v,k)}{\norm{P(v)}_1}D_Q(v,u)}
        \right.\\
        &\;\;\;\;\;\;\;\left.- \sum_{u\in U}\sum_{v\in U}{\frac{P(u,k)}{\norm{P(u)}_1}D_Q(u,v)} 
        - \sum_{u\in U}\sum_{v\in W}{\frac{P(u,k)}{\norm{P(u)}_1}D_Q(u,v)}
        \right| \\&=
        \sum_{k \in \K}
        \left|
        \sum_{u\in U}\sum_{v\in W}{\frac{P(v,k)}{\norm{P(u)}_1}D_Q(v,u)}
        - \sum_{u\in U}\sum_{v\in W}{\frac{P(u,k)}{\norm{P(u)}_1}D_Q(u,v)}
        \right| \\
        &\numle{1}
        \sum_{u\in U}\sum_{v\in W}\sum_{k \in \K}{\frac{|P(v,k)|}{\norm{P(v)}_1}D_Q(v,u)
        +\frac{|P(u,k)|}{\norm{P(u)}_1}D_Q(u,v)}
        \\&=
        \sum_{u\in U}\sum_{v\in W}{D_Q(v,u)
        +D_Q(u,v)} = \dem_{Q}(U,W),
    \end{align*}
    where Inequality~(1) is due to triangle inequality and change of order of summation.
\end{proof}

    \subsection{Omitted Proofs from Section \texorpdfstring{\ref{section:charging-scheme}}{5}}

    \begin{proof}[Proof of Claim~\ref{claim:initial_flow_state}]
        Recall Figure \ref{fig:U-W-cut}. Intuitively, when we spread the demands of $P(v)$ on the neighboring split nodes of $v$, we do not change the demand across the $(U',W')$ cut, except in the case that this spreading moved mass from one side of the cut to the other. In other words, this happens only if the edge $(x_e, v)$ crosses the cut. This may occur only if $e = (u,v)\in E(U,W)$, and $v\in W$. In this case, the amount of mass that crosses from $v$ to $x_e$ is at most $1$. Summing across all edges in $E(U,W)$, we only change the demand across the cut by at most $\capacity(U,W)$. 
    
        The formal proof is as follows. Recall that $\norm{P_{\{v\}'}(x_e)}_1 \le 1$ for every $v\in V$ and $e\in E$ adjacent to $v$. We get that
        \begin{align*}
            \dem_{P_{\Q_0}}(U',W') &= 
            \sum_{k\in \K} \left| \sum_{u\in U'} P_{\Q_0}(u,k) \right| =
            \sum_{k\in \K} \left| \sum_{u\in U'} \sum_{v\in V} P_{\{v\}'}(u,k) \right| \\
            &\numge{1} \sum_{k\in \K} \left|  \sum_{v\in U} \sum_{u\in U'} P_{\{v\}'}(u,k) \right|  - 
            \sum_{k\in \K} \left| \sum_{v\in W} \sum_{u\in U'} P_{\{v\}'}(u,k) \right| \\
            &\numge{2} \sum_{k\in \K} \left|  \sum_{v\in U}  P(v,k) \right|  - 
            \sum_{k\in \K}  \sum_{v\in W} \sum_{u\in U'} \left| P_{\{v\}'}(u,k) \right| \\
            &= 
            \dem_{P}(U,W) - 
               \sum_{v\in W} \sum_{k\in \K} \sum_{u\in U' \cap \{ v\}'} \left| P_{\{v\}'}(u,k) \right| \\
            &=
            \dem_{P}(U,W) - 
               \sum_{v\in W}  \sum_{u\in U' \cap \{ v\}'} \sum_{k\in \K}\left| P_{\{v\}'}(u,k) \right| \\
            &\numge{3} 
             \dem_{P}(U,W) - 
               \sum_{v\in W}  \sum_{u\in U' \cap \{ v\}'}  1 \\
            &\numeq{4} \dem_{P}(U,W) - \capacity_G(U,W),
        \end{align*}
        where Inequality~$(1)$ holds by swapping the order of summation and by triangle inequality; Inequality~$(2)$ holds by the way we defined $P_{\{v\}'}$ (spreading out the demands of $P(v)$ on the adjacent split nodes), as well as the way we defined $U'$ (that is, $U'$ contains all the split nodes defined by the cut $(U,W)$ in $G$) and triangle inequality; Inequality~$(3)$ holds by the mentioned observation that $\norm{P_{\{v\}'}(x_e)}_1 \le 1$, and Equality~$(4)$ holds by the way we defined $U'$.
        
    \end{proof}

    \subsection{Algebraic Tools}

    \begin{claim} [See \texorpdfstring{\cite[Lemma 10]{harrelson2003polynomial}}{[HHR03, Lemma 10]}]
    \label{claim:product-converges}
    For $k\ge 3$ and $c > 0$,
        \begin{align*}
            \prod_{\ell=1}^{k}{\left(1+\frac{c}{\ell}\right)}\le k^{2c}.
        \end{align*}
    \end{claim}

    \begin{proof}
        We have
        \[
        \prod_{\ell=1}^{k}\left(1+\frac{c}{\ell}\right)
        = \exp\!\left(\sum_{\ell=1}^{k}\ln\!\left(1+\frac{c}{\ell}\right)\right)
        \le \exp\!\left(\sum_{\ell=1}^{k}\frac{c}{\ell}\right)
        = \exp\!\bigl(c H_k\bigr),
        \]
        where we used $\ln(1+x)\le x$ for $x>-1$, and $H_k=\sum_{\ell=1}^{k}\frac{1}{\ell}$ denotes the $k$-th harmonic number.
        Since $H_k \le 2\ln k$, 
        \[
        \exp(c H_k)
        \le \exp\!\bigl(c\cdot2\ln k\bigr)
        = k^{2c}.
        \]
    \end{proof}
 
\end{document}